\documentclass[11pt]{article}

\usepackage[margin=1.1in]{geometry}
\usepackage{graphicx}

\usepackage{amsmath}
\usepackage{amssymb}
\usepackage{amsthm}

\usepackage{mathrsfs}
\usepackage{bm}

\usepackage{comment}
\newif\ifshow
\showfalse 
\ifshow
  \includecomment{comment}
\else
  \excludecomment{comment}
\fi
\usepackage{rotating}

\usepackage[dvipsnames]{xcolor}
\usepackage{ulem}
\usepackage{graphicx}
\usepackage{eso-pic}
\usepackage{grffile}
\usepackage{fancyhdr}
\usepackage{listings}
\usepackage{subcaption}
\usepackage{booktabs}
\usepackage{tabularx}
\usepackage{caption}
\usepackage{verbatim}
\usepackage{url}
 \usepackage[colorlinks=true, linkcolor=blue, citecolor=blue, urlcolor=blue]{hyperref}
\graphicspath{{./Figures/}}

\newtheorem{proposition}{Proposition}
\newtheorem{assumption}{Assumption}
\newtheorem{lemma}{Lemma}
\newcommand{\naA}{[Na^+]_A}
\newcommand{\kA}{[K^+]_A}
\newcommand{\clA}{[Cl^-]_A}

\newcommand{\naB}{[Na^+]_B}
\newcommand{\kB}{[K^+]_B}
\newcommand{\clB}{[Cl^-]_B}

\newcommand{\nae}{[Na^+]_e}
\newcommand{\ke}{[K^+]_e}
\newcommand{\cle}{[Cl^-]_e}
\newcommand{\nacle}{[NaCl]_e}
\newcommand{\ione}{[\text{ion}]_e}

\newcommand{\naj}{[Na^+]_j}
\newcommand{\kj}{[K^+]_j}
\newcommand{\clj}{[Cl^-]_j}
\newcommand{\ionj}{[\text{ion}]_j}

\newcommand{\wA}{w_A}
\newcommand{\wB}{w_B}

\newcommand{\wj}{w_j}

\newcommand{\wnorm}{w_{\text{Norm}}}

\newcommand{\vA}{V_A}
\newcommand{\vB}{V_B}
\newcommand{\vj}{V_j}

\newcommand{\parapath}{P}

\newcommand{\Arbl}{a_1}
\newcommand{\Arap}{a_2}
\newcommand{\Arpara}{a_{\parapath}}
\newcommand{\Arany}{a_{(\cdot)}}

\newcommand{\arbl}{\Arbl}%{\text{a}_1}
\newcommand{\arap}{\Arap}%{\text{a}_2}
\newcommand{\arany}{\Arany}

\newcommand{\eq}{^{eq}}

\newcommand{\zA}{z_{A}}
\newcommand{\zB}{z_{B}}
\newcommand{\zj}{z_{j}}
\newcommand{\zy}{z_{Y}}

\newcommand{\paramxA}{X_A^{\zA}}
\newcommand{\paramxB}{X_B^{\zB}}
\newcommand{\paramxj}{X_j^{\zj}}

\newcommand{\ye}{[Y^{\zy}]_e}
\newcommand{\xj}{[X^{\zj}]_j}

\newcommand{\enaA}{E_{Na,A}}
\newcommand{\ekA}{E_{K,A}}
\newcommand{\eclA}{E_{Cl,A}}
\newcommand{\enaB}{E_{Na,B}}
\newcommand{\ekB}{E_{K,B}}
\newcommand{\eclB}{E_{Cl,B}}
\newcommand{\eionj}{E_{\text{ion},j}}
\newcommand{\eionA}{E_{\text{ion},A}}
\newcommand{\eionB}{E_{\text{ion},B}}

\newcommand{\OsA}{\mathcal{O}_A}
\newcommand{\OsB}{\mathcal{O}_B}
\newcommand{\Ose}{\mathcal{O}_e}
\newcommand{\Osj}{\mathcal{O}_j}

\newcommand{\Cionjpbl}{\eta_{\text{ion},j}} % {\eta^{(1)}_{\text{ion},j}}
\newcommand{\Cnajpbl}{\eta_{Na,j}} % {\eta^{(1)}_{Na,j}}
\newcommand{\Ckjpbl}{\eta_{K,j}} % {\eta^{(1)}_{K,j}}
\newcommand{\Ccljpbl}{\eta_{Cl,j}} % {\eta^{(1)}_{Cl,j}}

\newcommand{\C}{\mathcal{C}\eq}
\newcommand{\Cpbl}{\mathcal{C}_{j}} % {\mathcal{C}^{(1)}_{j}}

\newcommand{\pnabl}{\text{\textbf{p}}_{\text{\textbf{Na,1}}}}
\newcommand{\pkbl}{\text{\textbf{p}}_{\text{\textbf{K,1}}}}
\newcommand{\pclbl}{\text{\textbf{p}}_{\text{\textbf{Cl,1}}}}

\newcommand{\pnaap}{\text{\textbf{p}}_{\text{\textbf{Na,2}}}}
\newcommand{\pkap}{\text{\textbf{p}}_{\text{\textbf{K,2}}}}
\newcommand{\pclap}{\text{\textbf{p}}_{\text{\textbf{Cl,2}}}}

\newcommand{\pionbl}{\text{\textbf{p}}_{\text{\textbf{ion,1}}}}
\newcommand{\pionap}{\text{\textbf{p}}_{\text{\textbf{ion,2}}}}

\newcommand{\pumprate}{p}

\newcommand{\jatp}[1]{J_{ATP}^{#1}}

\newcommand{\pMaxAbl}{\pumprate_{\max,A}} % {P_{\max,A}^{(1)}}
\newcommand{\pMaxBbl}{\pumprate_{\max,B}} % {P_{\max,B}^{(1)}}

\newcommand{\pMax}{\pumprate_{\max}}

    \newcommand{\Pmax}{\pMax} % {P_{\max}}
    
    \newcommand{\PmaxAbl}{\pMaxAbl} % {P_{\max,A}} % {P_{\max,A}^{(1)}}
    \newcommand{\PmaxBbl}{\pMaxBbl} % {P_{\max,B}} % {P_{\max,B}^{(1)}}
    
    \newcommand{\PmaxAap}{\PmaxAbl}%{P_{\max,A}^{(2)}}
    \newcommand{\PmaxBap}{\PmaxBbl}%{P_{\max,B}^{(2)}}

\newcommand{\pMinA}{\pumprate_{\min,A}}
\newcommand{\pMinB}{\pumprate_{\min,B}}

\newcommand{\zion}{z_{\text{ion}}}

\newcommand{\gnabl}{g_{Na,1}}
\newcommand{\gkbl}{g_{K,1}}
\newcommand{\gclbl}{g_{Cl,1}}

\newcommand{\gnaap}{g_{Na,2}}
\newcommand{\gkap}{g_{K,2}}
\newcommand{\gclap}{g_{Cl,2}}

\newcommand{\gnapara}{g_{Na,\parapath}}
\newcommand{\gkpara}{g_{K,\parapath}}
\newcommand{\gclpara}{g_{Cl,\parapath}}

\newcommand{\gionbl}{g_{\text{ion},1}}
\newcommand{\gionap}{g_{\text{ion},2}}
\newcommand{\gionpara}{g_{\text{ion},\parapath}}

\newcommand{\ggnabl}{\overline{g}_{Na,1}}
\newcommand{\ggkbl}{\overline{g}_{K,1}}
\newcommand{\ggclbl}{\overline{g}_{Cl,1}}

\newcommand{\ggnaap}{\overline{g}_{Na,2}}
\newcommand{\ggkap}{\overline{g}_{K,2}}
\newcommand{\ggclap}{\overline{g}_{Cl,2}}

\newcommand{\ggnapara}{\overline{g}_{Na,\parapath}}
\newcommand{\ggkpara}{\overline{g}_{K,\parapath}}
\newcommand{\ggclpara}{\overline{g}_{Cl,\parapath}}

\newcommand{\ggionap}{\overline{g}_{\text{ion},2}}
\newcommand{\ggionpara}{\overline{g}_{\text{ion},\parapath}}

\newcommand{\ggnaany}{\overline{g}_{Na,\cdot}}
\newcommand{\ggkany}{\overline{g}_{K,\cdot}}
\newcommand{\ggclany}{\overline{g}_{Cl,\cdot}}

\newcommand{\gionany}{{g}_{\text{ion},(\cdot)}}
\newcommand{\ggionany}{\overline{g}_{\text{ion},(\cdot)}}

\newcommand{\GionAbl}{G_{\text{ion},A}} % {G_{\text{ion},A}^{(1)}}
\newcommand{\GionBbl}{G_{\text{ion},B}} % {G_{\text{ion},B}^{(1)}}
\newcommand{\GionAap}{G_{\text{ion},A}^{(2)}}
\newcommand{\GionBap}{G_{\text{ion},B}^{(2)}}
\newcommand{\Gionjbl}{G_{\text{ion},j}} % {G_{\text{ion},j}^{(1)}}

\newcommand{\Gnajbl}{G_{Na,j}} % {G_{Na,j}^{(1)}}
\newcommand{\Gkjbl}{G_{K,j}} % {G_{K,j}^{(1)}}

    \newcommand{\fluxionbl}{\Phi_{\text{ion},1}}
    \newcommand{\fluxionap}{\Phi_{\text{ion},2}}
    \newcommand{\fluxionpara}{\Phi_{\text{ion},\parapath}}

    \newcommand{\fluxnabl}{\Phi_{\text{Na},1}}
    \newcommand{\fluxnaap}{\Phi_{\text{Na},2}}
    \newcommand{\fluxnapara}{\Phi_{\text{Na},\parapath}}

\newcommand{\dna}{d_{Na}}
\newcommand{\dk}{d_{K}}
\newcommand{\dcl}{d_{Cl}}

\newcommand{\uA}{u_{1}}
\newcommand{\uB}{u_{2}}

\definecolor{myblue}{RGB}{20,120,245}
\definecolor{myorange}{RGB}{243,108,10}
\definecolor{myred}{RGB}{237,28,36}
\definecolor{mygreen}{RGB}{0,128,0}%{0,166,81}
\definecolor{magenta}{RGB}{255,26,255}

\newcommand{\ZA}[1]{{\color{black}{#1}}}
\newcommand{\AK}[1]{{\color{black}{#1}}}
\newcommand{\KT}[1]{{\color{black}{#1}}}

\title{{Stability and robustness of a generalized pump-leak model \\for epithelial cell and lumen volume regulation}}

\author{Kerry Tarrant\thanks{Department of Mathematics, University of Iowa, IA,  USA,  
  kerry-tarrant@uiowa.edu.}
\; and\; Alan R. Kay \thanks{Department of  Biology, University of Iowa, IA, USA,
  alan-kay@uiowa.edu.}
\; and\; Zahra Aminzare\thanks{Department of  Mathematics, University of Iowa, IA, USA,
  zahra-aminzare@uiowa.edu.}
  }
\date{}

\begin{document}

\maketitle

\begin{abstract}
Epithelial cells regulate ion concentrations and volume through coordinated membrane pumps, ion channels, and paracellular pathways which can be modeled by classical single-compartment pump-leak equations (PLEs). Many epithelial functions, however, depend on the interaction between a cell and an enclosed luminal space, a geometry that cannot be captured by classical PLEs. To address this, we develop a two-compartment model consisting of an intracellular compartment  coupled to a luminal compartment through the apical membrane, with both compartments  {interfacing} an infinite extracellular bath and connected to it through the basolateral membrane and a paracellular pathway. Building on the five-dimensional single-cell PLEs, we formulate a ten-dimensional PLE system for this geometry and derive analytical equilibria and steady-state formulas for both the passive system and the Na\textsuperscript{+}/K\textsuperscript{+}-ATPase (NKA) driven active system. We characterize how these equilibria depend on physiologically relevant parameters, analyze local stability across wide parameter ranges, and apply global sensitivity and robustness methods to identify the principal determinants of ion and volume homeostasis. The model reveals fundamental differences between basolateral and apical placement of the NKA, including the onset of luminal volume blow-up when apical potassium recycling is insufficient.  More broadly, this framework provides a mathematically tractable and physiologically grounded foundation for studying epithelial transport and for predicting conditions under which pump localization and conductance changes lead to stable function or pathological lumen expansion.
\end{abstract}

 \textbf{keywords}
Pump-leak model, pump-leak equations, ABp system, volume regulation, NKA pump, sensitivity analysis, Sobol index, robustness.

%~~~~~~~~~~~~~~~~~~~~~~~~~~~~~~~~~~~~~~~~~~~~~~~~~~~~~~~~~~~~~~~~~~~
\section{Introduction}

%Donnan effect & osmotic pressure

\AK{All live cells function under the constant threat of inundation by osmotic water fluxes that are driven by impermeant molecules within cells. This so-called \textit{Donnan effect} is generated because cells have membranes that are permeable to water and contain impermeant molecules, which ensures that water will flow continuously  into the cell unless countermeasures are taken \cite{Sperelakis_2012}. % add to bib
 Eukaryotes employ {a variety of} mechanisms to {counter} the Donnan effect. Plant cells have rigid cell walls that {constrain} swelling and allow them to sustain a turgor pressure \cite{hofte2017plant} {that opposes the osmotic influx of water.}  In contrast animal cells pump $Na^+$ ions out, which establishes a dynamic steady state, where the intracellular osmolarity balances that of the extracellular osmolarity. This pump-leak mechanism (PLM) also ensures that there is osmotic ``room" in the cell to accommodate about 100 mM of a variety of metabolites, which are needed for the cell to survive \cite{Kay_Blaustein_2019, tosteson1960regulation}.  In all animal cells the sodium pump is the Na\textsuperscript{+}/K\textsuperscript{+}-ATPase (NKA), which transports three $Na^+$ out \KT{of} the cell for every two $K^+$  pumped in, at the expense of one ATP molecule \cite{Fedosova_2021}.} 

%PLEs for single cells
\AK{Much work has been done to characterize the behavior of the single cell PLM, showing that it establishes a robust means for stabilizing volume, as well as accommodating metabolites by moving chloride out of the cell, as well as establishing a negative membrane potential.  }
The PLM is described by a system of five algebraic-differential equations--collectively known as the pump-leak equations (PLEs)--that model the time-dependent changes in ionic concentrations, cell volume, and voltage. These equations have been rigorously analyzed for the single-cell setting  \cite{hoppensteadt2002modeling, FRASER2007336, keener2009mathematical, mori2012mathematical}. In this case, the cell is bathed in an infinite, well-mixed extracellular medium and exchanges ions and water across a single membrane, and more recent descriptions and mathematical inquiries have incorporated cation-chloride cotransporters \cite{aminzare2024mathematical} (Figure~\ref{fig:diagram_ANBp}, left). \AK{The PLEs create a system that ensures that a steady state emerges that is not dependent on a very narrow range of parameters.}

\begin{figure}[h!]
    \centering
    \includegraphics[width=1\textwidth]{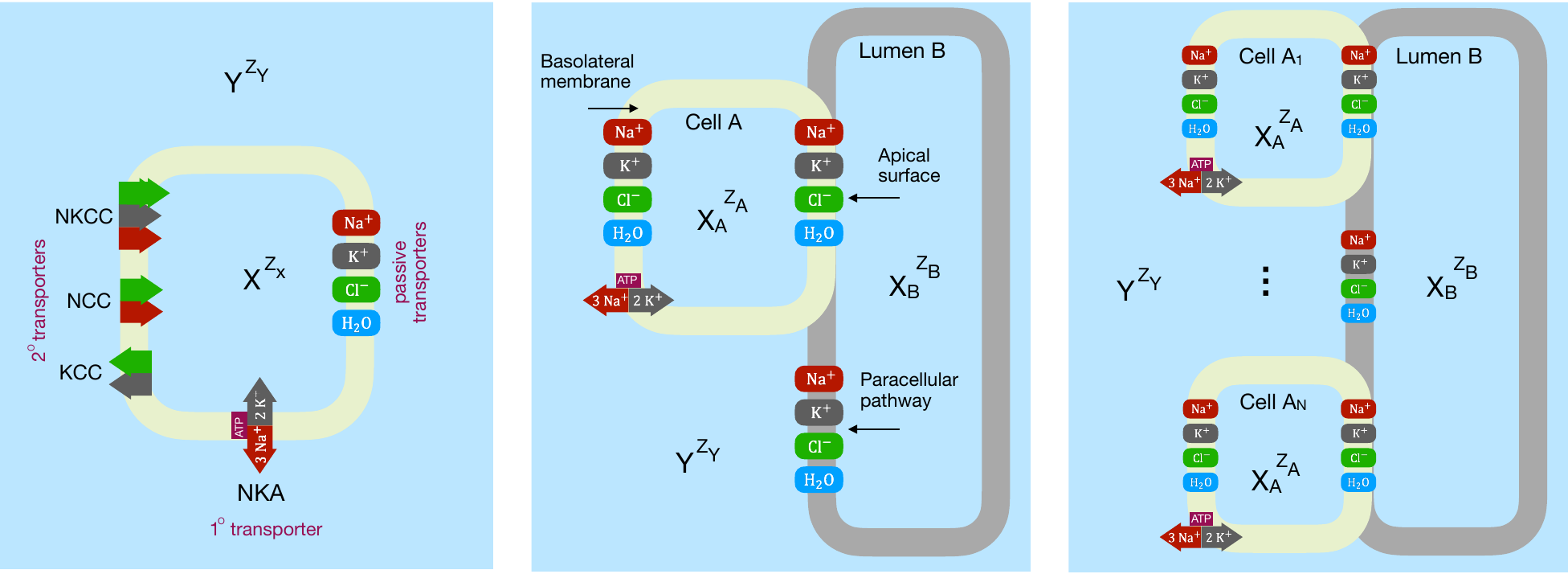}
    \caption{(Left) Schematic of a single-cell featuring Na\textsuperscript{+}, K\textsuperscript{+}, Cl\textsuperscript{-}, and water channels; the Na\textsuperscript{+}/K\textsuperscript{+}-ATPase (NKA) pump (a primary transporter); secondary transporters including KCC, NCC, and NKCC; and impermeant molecules inside (X) and outside (Y) the cell. (Middle) Schematic of the ABp model, where cell A connects to a lumen B via the apical surface and to the extracellular fluid via the basolateral membrane. The lumen B is also connected to the extracellular fluid through a paracellular pathway (p in ABp system refers to the paracellular pathway). The apical surface, basolateral membrane, and paracellular pathway contain passive ion channels and water transporters. Additionally, the basolateral membrane includes the NKA pump and other active transporters. (Right) Generalization of the ABp model to the $A_1\cdots A_NBp$ model, in which $N$ cells are connected to the common lumen B.}
    \label{fig:diagram_ANBp}
\end{figure} 

\AK{\textbf{Epithelia} are ubiquitous in animals, often providing the working end of an organ, whose function is to actively transport ions,  molecules and water, or to serve as \KT{a} barrier \cite{boron2016medical}. Simple epithelia consist of cells arranged in a tiled monolayer \cite{Caceres_2017}. 
 Junctions between adjacent epithelial cells are formed by transmembrane proteins that connect adjacent cells. Physiologists classify epithelia as either ``tight" or ``leaky", which is determined by the permeability characteristics of these intercellular junctions \cite{sackin2013electrophysiological}. Epithelia exhibit structural and functional asymmetry, with distinct populations of ion channels and transporters distributed between the apical membrane (facing the lumen) and the basolateral membrane (facing the interstitial space). This asymmetric distribution enables epithelia to transport ions and water vectorially.

Our objective is to model a rather specialized epithelium, which we term an ``epithelial vesicle". In these structures a layer of epithelial cells surrounds a lumen, which contains, water, ions and small metabolites (Figure~\ref{fig:diagram_ANBp}, right). The following are examples of epithelial vesicles: the ventricular system in the brain \cite{Kandel_2013}, 
 the scala media in the cochlea \cite{Patuzzi_2011}, 
  follicles of the thyroid, ovarian follicles, and embryonic structures like, the blastula, and the otic and optic vesicles \cite{Wolpert_2015}. 
    In all these cases the size of the lumen remains constant, at least in the short term. If all cells in the epithelial vesicle are the same, and the cells and lumen are small (i.e. of the order of a few microns), the cells can be collapsed into a single compartment A (Figure~\ref{fig:diagram_ANBp}, middle), since diffusion is rapid over short distances \cite{Berg_1993}.  
    We will term this a static ABp, where the volumes of A and B reach a constant steady state. This contrasts with what happens in a transporting epithelial system where the volume of B grows continuously as a fluid is secreted into it. 
}

\AK{
In this paper we develop a two compartment PLM where compartment A represents an epithelial cell, which borders the extracellular 
\ZA{interstitial fluid (ISF)} which is infinite and a compartment B that represents the lumen. Furthermore, the lumen can be coupled to the extracellular space by a paracellular pathway (p), hence the model is termed the ABp model. The permeability of the paracellular pathway is determined by occludin and claudin molecules with span the gap between epithelial cells \cite{Tsukita_2018}. 
Claudin molecules (23 genes in humans), can establish ion selective pathways in the space between epithelial cells. Claudins allow the selective permeation of ions and water, depending on the claudin gene expressed in the epithelium, or as a barrier to ion or water fluxes.  

In the absence of a paracellular pathway, the ABp model reduces to an AB system, where now compartment B is contained within A (\ZA{Figure~\ref{fig:AB_system}}). The AB system can serve as a model for intracellular organelles like, lysosomes, nuclei, endoplasmic reticulum etc. 

Similar models have been developed to simulate various kinds of epithelial transport \cite{Larsen_2000, Latta_1984, lew1979behaviour, Weinstein_Stephenson_1981}, 
however crucially in these models the compartments flanking the epithelium have been assumed to be infinite and have fixed concentrations. In contrast, here we allow the concentrations of ions in B to vary, as well as its volume and voltage. 

In the static ABp model, where in the steady state the volumes of all compartments reach a stable volume, the fluxes of any ion or water into and out a compartment must be balanced. We will show that it is possible for there to be persistent cyclic flows of ions through the model, which do not need to be balanced.}

Understanding the collective behavior of these multicellular arrangements is essential for explaining coordinated volume regulation, fluid secretion, and absorption, and may reveal emergent properties that do not appear in single-cell models.

\medskip 
\ZA{
Our main contributions in this work are as follows:

\noindent\textbf{Development of a mathematical framework.}
We generalize the classical five-dimensional pump-leak equations for a single cell to a ten-dimensional system describing the coupled dynamics of ion concentrations, volume, and membrane potential in a two-compartment ABp system.

\noindent\textbf{Existence and local stability of steady states.} 
We derive parameter conditions guaranteeing the existence of equilibria and steady states, obtain explicit analytical expressions for these states, and establish their local stability.

\noindent\textbf{Robustness and sensitivity analysis in high-dimensional parameter space.} 
The ABp model depends on a large number of parameters. To assess robustness of the steady states and long-term behavior, we analyze model sensitivity across a broad parameter range using Latin hypercube sampling combined with Sobol variance-based sensitivity indices. This approach quantifies how uncertainty in individual parameters and their interactions contributes to variability in key model outputs.

\noindent\textbf{Analysis of pump localization.}
We analyze two distinct pump configurations—placement on the basolateral membrane versus the apical surface—and characterize their effects on the existence and stability of steady states.
\medskip

Our analysis leads to the following outcomes:

\noindent\textbf{1.} As in the classical single-cell pump--leak model, the Na$^+$/K$^+$-ATPase plays a central role in stabilizing the two-compartment system and regulating cell and lumen volume. Moreover, we show that the qualitative behavior of the ABp system is insensitive to the precise mathematical form of the pump, justifying the use of a constant pump rate to obtain tractable analytical results.

\noindent\textbf{2.}
Global sensitivity analysis reveals a pronounced low-dimensional structure in parameter space. In particular, sodium conductances, extracellular sodium concentration, and, in some regimes, temperature account for a dominant fraction of output variance and can partially compensate for reduced pump activity, thereby shifting volume steady states.

\noindent\textbf{3.}
Although potassium and chloride conductances are necessary for physiological operation of the ABp system, the steady-state behavior is comparatively insensitive to variations in these parameters over wide ranges.

\noindent\textbf{4.}
Pump localization has qualitatively distinct consequences: while basolateral pump placement supports stable regulation of both compartments, apical pump placement leads to divergence of the luminal volume once the pump rate exceeds a critical threshold.
}
\medskip 

%stucture
The structure of the paper is as follows.
Section~\ref{sec:model} introduces the coupled pump-leak equations for a general two-compartment system.
Section~\ref{sec:passive} derives analytical expressions for the equilibria of the purely passive system, in which ionic fluxes driven by electrochemical gradients are the sole transport mechanisms.
Section~\ref{sec:active} extends this analysis by incorporating active transport via the NKA pump on the basolateral membrane and establishes the corresponding steady-state solutions.
Section~\ref{sec:dependencyOnParams} employs sensitivity analysis and Sobol indices to investigate the robustness of the coupled PLEs under parameter fluctuations and to identify mechanisms that promote volume regulation.
Section~\ref{subsection:pumpAS} considers the complementary case in which the NKA pump is located on the apical surface, yielding steady-state expressions that depend explicitly on membrane conductances and the pump rate.
We conclude in Section~\ref{sec:discussion}. All tables and numerical details are provided in the Appendix.

%~~~~~~~~~~~~~~~~~~~~~~~~~~~~~~~~~~~~~~~~~~~~~~~~~~~~~~~~~~~~~~~~~~~
\section{Pump-Leak Equations for Two-Compartment Systems}
\label{sec:model}

The classical pump-leak equations (PLEs) comprise a system of four differential equations and one algebraic constraint describing the intracellular concentrations of $Na^+$, $K^+$, and $Cl^-$, together with cellular volume and membrane potential for a single cell immersed in an infinite bath. Introduced by Tosteson and Hoffman in 1960 \cite{tosteson1960regulation}, this framework was subsequently extended to incorporate mechanisms of cell volume regulation \cite{jakobsson1980interactions} and epithelial transport \cite{lew1979behaviour}. Mori \cite{mori2012mathematical} later established the existence and uniqueness of an asymptotically stable steady state, followed by additional analytical developments \cite{keener2009mathematical, aminzare2024mathematical}.

In this section, we generalize the PLE framework to describe ionic transport, compartmental volumes, and membrane potentials in two-compartment systems, with particular emphasis on ABp models (illustrated in the middle panel of Figure~\ref{fig:diagram_ANBp}). The resulting formulation consists of ten coupled differential-algebraic equations, which we refer to as the \emph{coupled PLEs}.

To accommodate the additional compartment, we introduce notations that distinguishes quantities in the cytoplasm (compartment $A$), lumen (compartment $B$), and extracellular space (ISF). Variables and parameters associated with compartments $A$ and $B$ carry the subscripts $_A$ and $_B$, respectively, while extracellular quantities carry the subscript $_e$. We further use the subscripts $_1$, $_2$, and \KT{$_\parapath$} to identify the basolateral membrane (between $A$ and ISF), the apical surface (between $A$ and $B$), and the paracellular pathway (between $B$ and ISF).

\medskip 

\noindent\textbf{Volume.} We make the standard assumption that the solution in each compartment is composed entirely of water, allowing us to treat osmolarity and osmolality as equivalent measures of total solute concentration \cite{boron2016medical}. Since the membrane is permeable to water, if there are differences in osmolarity across it, water will move by osmosis \cite{manning2023physical}. The osmolarities of the extra- and intracellular solutions are respectively $\Ose$ and $\Osj$ (for $j=A,B$):
\begin{subequations}
\begin{align}   
&\Ose = \nae + \ke + \cle + \ye \label{eq:ose},\\
    &\Osj = \naj + \kj + \clj + \xj,\label{eq:osj}
\end{align}
\end{subequations}
where $\ione$  and $\ionj$ represent the extracellular concentrations and intracellular concentrations for compartment $j$ ($j=A,B$), respectively. 
$\xj$ describes the concentration of impermeant molecules inside compartment $j$ with average charge $\zj$ which include metabolites and macromolecules. Similarly, $\ye$ with average charge $\zy$ represents the extracellular impermeant molecules. Because osmosis is a colligative effect, one need only consider the number of moles of these molecules. Note that we assume $\ye$ is constant and the number of moles of the intracellular impermeant molecule is constant, i.e., $\paramxj$ is constant. If $\wj$ is the compartment volume then $\xj={\paramxj}/{\wj}.$

{The Starling} equation governs water flux in the system \cite{garcia2013biophysical,boron2016medical,manning2023physical}, which describes changes in volume as proportional to the osmolarity differences across membranes{--the Starling equation normally includes the transmembrane pressure, which we omit since we assume that the compartments can stretch freely without developing tension and, thus, any pressure}. The following equations express the change in volumes in compartments $A$ and $B$:
\begin{subequations}
\begin{align}
    &\frac{ d \wA }{ dt } = \underbrace{{\color{black}\nu_1} (\OsA-\Ose)}_{\text{baso. mem.}} + \underbrace{{\color{black}\nu_2} (\OsA-\OsB)}_{\text{ap. sur.}} \label{eq:dwA} \\
    &\frac{ d \wB }{ dt } = \underbrace{{\color{black}\nu_p}(\OsB-\Ose)}_{\text{para. pathway}} - \underbrace{{\color{black}\nu_2} (\OsA-\OsB)}_{\text{ap. sur.}} \label{eq:dwB}
\end{align}
\label{eq:ode_water}
\end{subequations}
where the parameters $\nu_1$, $\nu_2$, and $\nu_p$ are the osmotic permeability coefficients times the molar volume of water.

\medskip 

\noindent\textbf{Voltage.}
The total concentration of charge inside the compartments and in the extracellular space  is given by:
 \begin{subequations}\label{eq:charge:constraint}
\begin{align}
Q_j&=\naj +\kj -\clj +\zj \xj\label{eq:elnu_i},\\
Q_e &=\nae+\ke- \cle + z_Y\ye\label{eq:elnu_o}. 
\end{align}
\end{subequations}

Because of the energetic cost of separating charges, isolated solutions will have a net charge of zero, so we can set $Q_e=0$
\begin{align}
    &\nae + \ke - \cle + \zy\ye = 0 \label{eq:sum_e=0} .
\end{align}

 For $j\in\{A,B\}$, the membrane potential in each compartment can be modeled exactly by the following  algebraic equation \cite{varghese1997conservation}: 
\begin{equation}
    \vj = \frac{F \wj}{C_{m,j}} \left( \naj + \kj -\clj + z_j \xj\right) ,
\label{eq:alg_volt}
\end{equation}
where $F$ is Faraday's constant and $C_{m,j}$ is the total compartment capacitance.  

The voltage difference between the spaces is negligible, allowing us to use the electroneutrality condition to approximate the net charge concentration \cite{keener2009mathematical, ostby2009astrocytic}, given by $Q_j=Q_e$. 

\medskip 

\noindent\textbf{Ion concentrations.}
In the presence of an electrochemical gradient, ionic channels facilitate the passive transport of ions by allowing them to move freely down their electrochemical gradients, thereby reducing the difference in charge and concentration across the membrane. Channels act as selective gateways, each attuned to a specific ionic species, such as sodium, potassium, and chloride. Passive ionic currents through these channels are described using Ohm's law, chosen {because it is a reasonable approximation for how conductances actually behave. Additionally, its linear dependence on electrochemical potential makes} subsequent analyses more tractable. Ohm's law for passive ionic current is given by:
\begin{equation*}
    i_{\text{ion}} = - z_{\text{ion}} \gionany (\vj-\eionj - V_e) \label{eq:ohmslaw}
\end{equation*}
where $z_{\text{ion}}$ is the ionic valence, $\gionany$ is the channel conductance on a given membrane or pathway, and the driving force is the transmembrane potential difference $\vj - \eionj - V_e$. The Nernst potential $\eionj$ is
\begin{align}
\eionj &= \frac{RT}{z_{\text{ion}}F} \ln\left(\frac{\ione}{\ionj}\right), \label{eq:E_ion_j}
\end{align}
with $R$ the universal gas constant and $T$ the absolute temperature. When $\vj - \eionj > V_e$, cations (respectively, anions) tend to move out of (respectively, into) the compartment. Throughout, the extracellular voltage serves as the reference potential and is fixed at $V_e = 0$.

The net ionic flux into each compartment results from the combined effects of passive transport, described by Ohm's law, and active transport processes localized to the membrane of compartment~$A$. The differential equations governing the concentrations of Na\textsuperscript{+}, K\textsuperscript{+}, and Cl\textsuperscript{-} in compartment $A$ are given by
\begin{subequations}
\begin{align} 
    &F\frac{d \left( \wA \naA \right) }{dt} = -\underbrace{{\color{black}\gnabl} \left( \vA - \enaA \right) + {\color{black}\pnabl}}_{} + \underbrace{{\color{black}\pnaap} + {\color{black} \dna}}_{} , \label{eq:dnaA} \\
    &F\frac{d(\wA\kA)}{dt}  = -\underbrace{{\color{black}\gkbl} \left( \vA - \ekA \right) + {\color{black}\pkbl}}_{} + \underbrace{{\color{black}\pkap} + {\color{black} \dk}}_{} , \label{eq:dkA} \\
    &F\frac{d(\wA\clA)}{dt} = \underbrace{{\color{black}\gclbl} \left( \vA - \eclA \right)}_{\text{ion total flux across basolateral}} + \underbrace{ {\color{black} \dcl}}_{\text{ion total flux across apical}}. \label{eq:dclA}
\end{align}
\label{eq:ode_con_a}
\end{subequations}
Similarly, the differential equations for the ionic concentrations in compartment $B$ are
\begin{subequations}
\begin{align}
    &F\frac{d \left( \wB \naB \right) }{dt} = -\underbrace{{\color{black}\gnapara} \left( \vB - \enaB \right)}_{} - \underbrace{{\color{black}\pnaap} - {\color{black} \dna}}_{} , \label{eq:dnaB} \\
    &F\frac{d(\wB \kB)}{dt}  = -\underbrace{{\color{black}\gkpara}\left( \vB - \ekB \right)}_{} - \underbrace{{\color{black}\pkap} - {\color{black} \dk}}_{} , \label{eq:dkB} \\
    &F\frac{d(\wB \clB)}{dt} = \underbrace{{\color{black}\gclpara} \left( \vB - \eclB \right)}_{\text{ion total flux across paracellular}} - \underbrace{ {\color{black} \dcl}}_{\text{ion total flux across apical}} . \label{eq:dclB}
\end{align}
\label{eq:ode_con_b}
\end{subequations}
In Equations \eqref{eq:ode_con_a}-\eqref{eq:ode_con_b}, $g_{\text{ion},1}$ and $g_{\text{ion},p}$ denote the conductance parameters along the basolateral membrane and the paracellular pathway, respectively. Active flux is given by the bolded terms \textbf{p}$_\text{ion,1}$ and \textbf{p}$_\text{ion,2}$. We will define the active transport terms at the beginning of subsequent sections. The term $d_{\text{ion}}$ is the passive flux across the apical surface and is defined as
\begin{align}
    \label{eq:dION}
    &d_{\text{ion}} = -z_{\text{ion}}{\color{black} g_{\text{ion},2}}\left( \left(\vA-E_{\text{ion},A})-(\vB-E_{\text{ion},B}\right) \right)
\end{align}
where $g_{\text{ion},2}$ is the conductance for the ionic channels located on apical surface. Passive flux along the apical membrane, in Equation \eqref{eq:dION}, is a function of the compartment concentrations and electrical potentials.

To maintain a constant extracellular osmolarity $\Ose$, we fix $\ke$. The concentrations $\cle$ and $\nae$ are determined by the following relationships: 
\begin{subequations}
    \begin{align}
        &\cle = \frac{1}{2}\left(\Ose + \left(\zy - 1\right)\ye \right), \label{eq:cle} \\
        &\nae = -\ke +\cle - \zy\ye \label{eq:nae}. 
    \end{align}
    \label{eq:ione}
\end{subequations}
Equation \eqref{eq:cle} is derived by combining the extracellular osmolarity in Equation \eqref{eq:ose} and the electroneutrality Equation \eqref{eq:sum_e=0}. Equation \eqref{eq:nae} comes directly from the assumption of electroneutrality. Table~\ref{tab:extracellular} gives the default values of solute concentrations.

From this point forward, we refer to the system consisting of the volume equation \eqref{eq:ode_water}, the voltage relation \eqref{eq:alg_volt}, and the ion concentration equations \eqref{eq:ode_con_a}--\eqref{eq:ode_con_b} collectively as the coupled PLEs for a general ABp system.

\medskip 

The middle illustration in Figure \ref{fig:diagram_ANBp} has some limitations. First, the volumes of compartments $A$ and $B$ are not {drawn to scale}; for a more precise representation, compartment $B$ should be at least ten times larger than compartment $A$. Second, the paracellular pathway is a narrow intercellular space through which ions and water flow. Despite these simplifications, the schematic highlights some features that will be useful later.

Default parameter values, sometimes called nominal values, are provided in Section \ref{sec:appendix_1}. Tables \ref{tab:constants}-\ref{tab:pump} provide default values and descriptions for the parameters used in our model. Table \ref{tab:constants} specifies the surface area for the basolateral and apical membranes, which are assumed to be identical unless otherwise stated, and the paracellular pathway. The product of area and conductance determines the effective number of channels in each membrane. Surface area calculates the distribution of leak channels and active transporters along surfaces and pathways. We assume that the number of channels and transporters does not change as the volume of the compartment changes. Area for the basolateral membrane, apical surface, and paracellular pathway will be denoted with $\Arbl$, $\Arap$, and $\Arpara$, respectively.

The ionic and water flux rates are asymmetrical along the surfaces and pathway since the total ionic conductance and water permeability are proportional to the surface area (given in Table~\ref{tab:constants}). For the ionic species on each surface or pathway, we will express total ionic conductance as $\gionany = \ggionany \cdot \Arany$, where $\ggionany$ is the ionic conductance per unit area. In this paper, equations and expressions will use total conductance, while figures will use conductance per unit area.

Active transport mechanisms, specifically the terms $\pionbl$ and $\pionap$ as seen in Equations \eqref{eq:ode_con_a} and \eqref{eq:ode_con_b}, are fixed at 0 in Section \ref{sec:passive} while each will be varied in Sections \ref{sec:active}-\ref{subsection:pumpAS}. Table \ref{tab:pump} lists the default values for total NKA pump rates and the stoichiometry coefficients for sodium and potassium on NKA pumps for the basolateral membrane and the apical surface. In this paper, we denote the NKA pump rate per unit area as $\pumprate$. The stoichiometry of the NKA pump is typically found to pump out three Na\textsuperscript{+} for every two K\textsuperscript{+} pumped into the cell. Under some circumstances in nature, the sodium to potassium ratios can differ \cite{artigas2023pump}, but, for most {animal cells}, the values given in Table \ref{tab:pump} reduce the electrochemical work performed by the NKA pump \cite{peluffo2023na}.

%~~~~~~~~~~~~~~~~~~~~~~~~~~~~~~~~~~~~~~~~~~~~~~~~~~~~~~~~~~~~~~~~~~~
\section{Passive Two-Compartment Systems}\label{sec:passive}

In this section, we present  conditions that guarantee the existence of an \textit{equilibrium} in a 10-dimensional \textit{passive} ABp system {where there is no active transport.}

\begin{assumption}\label{assumption:equilibrium}
To establish the existence of equilibrium states in passive 2-compartment systems, we impose the following assumptions.\begin{enumerate}
    \item $\pnabl =  \pkbl = \pnaap= \pkap = 0$. That says no active transport mechanisms are involved.

    \item For each ion, at least two of $g_{\text{ion},1}, g_{\text{ion},2}, g_{\text{ion},p}$ are strictly positive.

    \item   At least two of $\nu_1,\nu_2, \nu_p$ are strictly positive.

    \item $\ke,\nae,\cle,\ye > 0$.

\end{enumerate}
\end{assumption}

Note that to achieve a positive and finite volume at equilibrium, the concentration of impermeant molecules, $\xj\eq$, must be strictly positive. The next lemma provides the conditions for this.

\begin{lemma}\label{lemma:equilibrium}
{For $j\in\{A,B\}$,} let 
\begin{equation}
    {\xi_j} := \begin{cases}
                \dfrac{\Ose - \sqrt{4\left( 1 - z_j^2 \right) \C + \Ose^2 z_j^2}}{1 - z_j^2} & \text{ if } z_j^2 \neq 1\\
                \dfrac{\Ose^2 - 4\C}{2 \Ose} & \text{ if } z_j^2 = 1.
           \end{cases} \label{eq:xjeq}
\end{equation}
where
 \begin{align}
    \C = \cle \left( \nae  + \ke \right){=\cle\left(\cle-\zy\ye\right)} \label{eq:C}.
\end{align}
and assume that  $\ye > 0$. Then, $4 \C - \Ose^2 < 0$ and $\xi_j>0$. 
\end{lemma}

\begin{proof}
Substituting \eqref{eq:cle} into the third term of Equation \eqref{eq:C} gives
    \begin{align*}
        \C &= \left(\frac{1}{2}\Ose + \frac{1}{2}\left(\zy - 1\right)\ye \right)\left(\left(\frac{1}{2}\Ose + \frac{1}{2}\left(\zy - 1\right)\ye \right)-\zy\ye\right)\\ 
        &= \frac{1}{4}\left( \left(\Ose - \ye \right)^2  -  \zy^2\ye^2 \right) .
    \end{align*}
    Hence 
    $4 \C - \Ose^2=\ye\left(\ye - \zy^2 \ye - 2\Ose\right).$
    Since $0<\ye <\Ose$, we conclude that $4 \C - \Ose^2<0$. For $z_j^2 = 1$, it is straightforward to show that if $4 \C - \Ose^2<0$ then $\xi_j >0$. 
    For $z_j^2 < 1$ ($> 1$), since $4 \C - \Ose^2<0$, we can conclude that 
   $4(1-\zj^2)\C+\zj^2\Ose^2 < \Ose^2$, ($> \Ose^2$). Hence $\xi_j >0$. 
\end{proof}

\begin{proposition}\label{proposition:equilibrium}
Consider the coupled PLEs \eqref{eq:ode_con_a}--\eqref{eq:ode_water} for 2-compartment systems under Assumption~\ref{assumption:equilibrium}. Then, the system admits a stable equilibrium point, given explicitly as follows:
\begin{subequations}
    \begin{align}
    &\naj\eq = \dfrac{2 \nae \cle }{\Ose + \left( z_j - 1 \right) \xj\eq}, \label{eq:sodium_EQ} \\
    &\kj\eq  = \dfrac{2 \ke \cle }{\Ose + \left( z_j - 1 \right) \xj\eq}, \label{eq:potassium_EQ} \\
    &\clj\eq = \frac{1}{2}\left( \Ose + \left( z_j - 1 \right) \xj\eq \right), \label{eq:chloride_EQ} \\
    &\wj\eq = \frac{\paramxj}{\xj\eq}, \label{eq:volume_EQ} \\
    &\vj\eq = \frac{RT}{F}\ln\left( \frac{\Ose + \left( z_j - 1 \right)\xj\eq}{2 \cle} \right) . \label{eq:voltage_EQ}
    \end{align}
    \label{eq:EQ}
\end{subequations}
where $ \xj\eq = \xi_j$ is given in Equation \eqref{eq:xjeq}.
\end{proposition}

\begin{proof}
At equilibrium (using Assumption~\ref{assumption:equilibrium}.1), we set the right-hand side of Equations~\eqref{eq:ode_con_a}--\eqref{eq:ode_con_b} equal to zero:
\begin{subequations}
\begin{align}
0&=-g_{\text{ion},1}\left(\vA^{eq}-\eionA^{eq}\right)-g_{\text{ion},2}\left(\left(\vA^{eq}-\eionA^{eq}\right)-\left(\vB^{eq}-\eionB^{eq}\right)\right)\label{equ:equilibrium1}\\
0&=-g_{\text{ion},p}\left(\vB^{eq}-\eionB^{eq}\right)+g_{\text{ion},2}\left(\left(\vA^{eq}-\eionA^{eq}\right)-\left(\vB^{eq}-\eionB^{eq}\right)\right)\label{equ:equilibrium2}.
\end{align}
\end{subequations}
By summing \eqref{equ:equilibrium1} and \eqref{equ:equilibrium2}, and assuming that at least one of 
$g_{\text{ion},1}$ or $g_{\text{ion},p}$ is strictly positive (Assumption~\ref{assumption:equilibrium}.2), 
we can express $\vA^{eq} - \eionA^{eq}$ in terms of $\vB^{eq} - \eionB^{eq}$, or vice versa. 
We can then use \eqref{equ:equilibrium1} or \eqref{equ:equilibrium2} to solve for both quantities. 
For example, assume that $g_{\text{ion},1} > 0$, then
\begin{align}\label{equ:equilibrium3}
\vA^{eq}-\eionA^{eq}=-\frac{g_{\text{ion},2}}{g_{\text{ion},1}}\left(\vB^{eq}-\eionB^{eq}\right). 
\end{align}
{Substituting \eqref{equ:equilibrium3} into \eqref{equ:equilibrium1} yields}
\[0=\left(g_{\text{ion},1}g_{\text{ion},2}+g_{\text{ion},1}g_{\text{ion},\parapath}+g_{\text{ion}2}g_{\text{ion},\parapath}\right)\left(\vB^{eq}-\eionB^{eq}\right).\]
Since by Assumption~\ref{assumption:equilibrium}.2, the first term of the product is non-zero, we get $\vB^{eq}-\eionB^{eq}=0,$ {and hence, we conclude that} for $j\in\{A,B\}$,  the following electrochemical potential condition holds for sodium, potassium, and chloride:
\begin{align}
    &\vj\eq = \eionj\eq \label{eq:vj-EQ}. 
\end{align}
At equilibrium, the chemical and electrical potentials counterbalance each other, leaving the system without a net flux of ions.

Next,  setting the right-hand side of Equation \eqref{eq:ode_water} to zero we obtain 
\begin{subequations}
\begin{align}
&0=\nu_1\left(\OsA^{eq}-\Ose\right)+\nu_2\left(\OsA^{eq}-\OsB^{eq}\right)\label{equ:equilibrium4}\\
&0=\nu_p\left(\OsB^{eq}-\Ose\right)-\nu_2\left(\OsA^{eq}-\OsB^{eq}\right)\label{equ:equilibrium5}.
\end{align}
\end{subequations}
By summing \eqref{equ:equilibrium4} and \eqref{equ:equilibrium5}, and assuming that at least one of $\nu_1$ or $\nu_p$ is strictly positive, we can express  $\OsA^{eq}-\Ose$ in terms of $\OsB^{eq}-\Ose$, or vice versa. We can then use \eqref{equ:equilibrium4} or \eqref{equ:equilibrium5} to solve for both quantities. 
For example, assume that 
$\nu_1>0$. Then, 
\begin{align}\label{equ:equilibrium6}
\OsA^{eq}-\Ose=-\frac{\nu_p}{\nu_1}\left(\OsB^{eq}-\Ose\right).
\end{align}
{Substituting \eqref{equ:equilibrium6} into \eqref{equ:equilibrium4} yields}
$\left(\nu_1\nu_2+\nu_1\nu_p+\nu_2\nu_p\right)\left(\OsB^{eq}-\Ose\right)=0,$ {and hence,}
for $j\in\{A,B\}$, we get $\mathcal{O}_j\eq = \Ose$ or equivalently, 
\begin{align}   
 \naj\eq+\kj\eq+\clj\eq+\xj\eq = \Ose \label{eq:ose_EQ}. 
\end{align}

At equilibrium, the compartments have equal osmolarity. 
As a result, the surfaces and pathways lack any net water flow. Although the osmolarities are the same in both spaces, the solute composition of each compartment may differ \cite{boron2016medical}.

Electroneutrality requires 
\begin{align}
    &\naj\eq+\kj\eq-\clj\eq+z_j\xj\eq = 0 \label{eq:EN_EQ} .
\end{align}
Adding Equations \eqref{eq:ose_EQ} and \eqref{eq:EN_EQ} gives us 
\begin{align}
    &\naj\eq + \kj\eq = \frac{1}{2}\left( \Ose - \left( z_{j} + 1 \right) \xj\eq \right) . \label{eq:Ose_p_EN_EQ}
\end{align}
Subtracting Equation \eqref{eq:EN_EQ} from Equation \eqref{eq:ose_EQ} gives us the equilibrium value of chloride as in \eqref{eq:chloride_EQ}.

From Equation \eqref{eq:vj-EQ}, we can express the intracellular electrical potential in terms of $\clj\eq$, which gives us $\vj\eq = -\frac{RT}{F}\ln\left(\frac{\cle}{\clj\eq}\right)$. Chloride concentration $\clj\eq$ can be substituted with the right-hand side of Equation \eqref{eq:chloride_EQ}. This allows us to describe {the equilibrium value of} electrical potential in terms of $\xj\eq$, as given in \eqref{eq:voltage_EQ}.

The intracellular sodium and potassium concentrations $\naj$ and $\kj$ are isolated from the potential in Equation \eqref{eq:vj-EQ}, and expressed as
\begin{align}
    \ionj\eq &= [\text{ion}]_e \exp\left( -\frac{z_{\text{ion}}F}{RT} \vj\eq  \right) . \label{eq:ion_EQ}
\end{align}
Notice that both $\naj\eq$ and $\kj\eq$ depend on the $\vj\eq$, which in turn depends on $\xj\eq$ as expressed in {\eqref{eq:sodium_EQ}--\eqref{eq:potassium_EQ}}. {Finally, by definition, the equilibrium value of the volume can be expressed as \eqref{eq:volume_EQ}.}

To finish the proof of existence, we now solve for the impermeant concentration $\xj\eq$. The ionic concentrations on the left-hand side of Equation \eqref{eq:Ose_p_EN_EQ} can be replaced with the expressions found on the right-hand side of Equation \eqref{eq:ion_EQ}. The voltage in Equation \eqref{eq:Ose_p_EN_EQ} can be replaced with the right-hand side of Equation \eqref{eq:voltage_EQ}. This gives us
\begin{align}
    \frac{1}{2}\left( \Ose - \left( 1+z_j \right)\xj\eq \right) &= 
    \nae \exp\left( -\ln\left( \frac{\Ose + \left( z_j - 1 \right)\xj\eq}{2 \cle} \right)\right) \nonumber \\
    &\quad+ \ke \exp\left(  -\ln\left( \frac{\Ose + \left( z_j - 1 \right)\xj\eq}{2 \cle} \right) \right) . \label{eq:volt_ion_sum_EQ}
\end{align}
Solving for $\xj\eq$ in Equation \eqref{eq:volt_ion_sum_EQ} gives us an expression for impermeant concentration as given in \eqref{eq:xjeq}. Note that Assumption~\ref{assumption:equilibrium}.4 and Lemma~\ref{lemma:equilibrium} guarantee that $\xj\eq$ is strictly positive and hence the corresponding volume is finite. 
The rest of Equation~\eqref{eq:EQ} can be derived from \eqref{eq:ion_EQ} and the fact that $w_j=\paramxj/[X_j]$.

{The proof of stability is identical to that of the active 2-compartment system with $\pumprate = 0$. We will discuss this case in detail in Section~\ref{sec:localstability}.}
\end{proof}

Equation~\eqref{eq:EQ} provides analytical expressions for the equilibrium solutions of the coupled PLEs described in Section~\ref{sec:model}. The expressions for each compartment depend on the shared extracellular concentrations as well as on the number of impermeant molecules, $\paramxj$, and their average charges, $\zj$.  

If these compartment-specific parameters are identical, i.e., $\paramxA = \paramxB$ and $\zA = \zB$, then the equilibrium states of the two compartments become indistinguishable, and the ABp system behaves like a single compartment at equilibrium. See \cite{aminzare2024mathematical} for more details on the single-cell case. 

Note that the equilibrium volume $\wj^{eq}$ is the only state variable that depends on the number of impermeant molecules {$\paramxj$}, and the equilibrium values for the ionic concentrations and voltages only depend on $\zj$. {A well-defined, locally asymptotically stable equilibrium requires impermeant molecules in the external compartment, i.e., $\ye>0$, to provide the osmotic constraint needed for volume regulation.} As mentioned before, for $\zA = \zB$, the equilibrium values of ion concentrations and voltages of compartments $A$ and $B$ are identical, and their volumes depend on $\paramxA$ and $\paramxB$. 

%~~~~~~~~~~~~~~~~~~~~~~~~~~~~~~~~~~~~~~~~~~~~~~~~~~~~~~~~~~~~~~~~~~~
\section{Active Two-Compartment Systems:  NKA Pump on Basolateral Membrane} \label{sec:active}

In this section, we introduce a Na\textsuperscript{+}/K\textsuperscript{+} ATPase (NKA) pump to the cellular membrane and derive the steady states of coupled PLEs analytically. In most epithelial systems, the NKA pump is confined to the basolateral membrane {or on the lateral borders of epithelial cells} \cite{shoshani2005polarized}. In rarer cases, the NKA pump is expressed on the apical surface \cite{pollay1985choroid}. We first consider the NKA pump on the basolateral membrane and later in  in Section \ref{subsection:pumpAS}, we will explore the effect of the NKA pump on the apical surface.  An active NKA pump mechanism will pump $\gamma_{Na}$ Na\textsuperscript{+} from the intracellular space to the extracellular space for every $\gamma_{K}$ K\textsuperscript{+} ions pumped into the intracellular space. Here, $\gamma_{Na}$ and $\gamma_{K}$ denote the stoichiometries of the NKA pump -- typically 3 and 2, respectively \cite{patel2023asymmetric}. 
{Active transport by the basolateral NKA directly regulates ionic concentrations in compartment $A$ and, through transcellular coupling and paracellular leak, also influences concentrations in compartment $B$.} We will use Equations \eqref{eq:ode_con_a}--\eqref{eq:ode_con_b} with the \textit{constant} active pump mechanism given by:
\[
    \pnabl = -\gamma_{Na} \, \pumprate \, \Arbl, \quad \pkbl = \gamma_{K} \, \pumprate \, \Arbl, \quad \pnaap=0, \quad \pkap = 0, 
\]
where $\pumprate$ and $\arbl$ are non-negative constants. The schematic diagrams for ABp systems in Figure~\ref{fig:diagram_ANBp}(middle) illustrate this configuration, with the NKA pump located on the interface between the extracellular space and compartment $A$. Passive transport mechanisms are the sole means of ion transport at the confluence of compartments $A$ and $B$, effectively coupling the two compartments and enabling the NKA pump to influence both. The ion composition observed in compartment $B$ almost mirrors that in compartment $A$, while volume changes are also proportional. 

Note that the constant NKA pump used here represents the simplest form. More complex nonlinear formulations exist, such as 
\begin{align}\label{eq:nonlinear_pump}
\pumprate \, \Arbl\left(\frac{\ke}{k_p+\ke}\right)^2\left(\frac{\naA}{k_{Na}+\naA}\right)^3,
\end{align}
first introduced in \cite{GarayGarrahan1973}, where $k_p = 0.883$ mM and $k_{Na} = 3.56$ mM are the apparent dissociation constants for $K^+$ and $Na^+$, respectively. Simpler nonlinear forms have also been proposed, including $\pumprate \, \Arbl \naA$ \cite{RN30260}, $\pumprate \, \Arbl \left(\tfrac{\naA}{\nae}\right)^3$ \cite{keener2009mathematical}, and $\pumprate \, \Arbl \left(\tfrac{\ke}{\kA}\right)^2\left(\tfrac{\naA}{\nae}\right)^3$ \cite{Manicka_Levin_2019}.
In \cite{aminzare2024mathematical}, Aminzare and Kay analytically demonstrated that the exact mathematical form of the NKA pump does not qualitatively alter the steady states of a single cell; only the stoichiometry of the pump affects them. 

 {The hydrolysis of ATP (Adenosine Triphosphate) releases free energy} that drives a wide range of cellular processes, including active ion transport \cite{DunnGrider2023}. {The rate at which the NKA pump operates determines the ATP consumption rate.} 
 In particular, the ATP consumption rate for a constant NKA pump is defined as $\jatp{\pumprate}=\frac{\pumprate \, \Arbl}{F}$.
For the nonlinear NKA pump in Equation \eqref{eq:nonlinear_pump}, this rate is $\jatp{\pumprate}=\frac{\pumprate \, \Arbl}{F}\left(\frac{\ke}{k_p+\ke}\right)^2\left(\tfrac{\naA}{k_{Na}+\naA}\right)^3.$ 
In Figure~\ref{fig:constant_vs_nonlinear_pumps}, we plot the steady states of the ABp system as a function of the ATP consumption rate for both the constant and nonlinear models of the NKA pump. As the figure shows, the steady state values are qualitatively similar, consistent with those found in single cells \cite{aminzare2024mathematical}. Therefore, it is reasonable to use a constant pump, as it does not alter the qualitative behavior of the steady states and is more mathematically tractable.
Note that while we use analytical expressions for the steady states in the constant pump case (derived in Section~\ref{sec:existence} below), we compute them numerically (using the stiff ODE solver \texttt{ode15s} in MATLAB) for the nonlinear pump models.

\begin{figure}[h!]
   \centering
   \includegraphics[width=.7\linewidth]{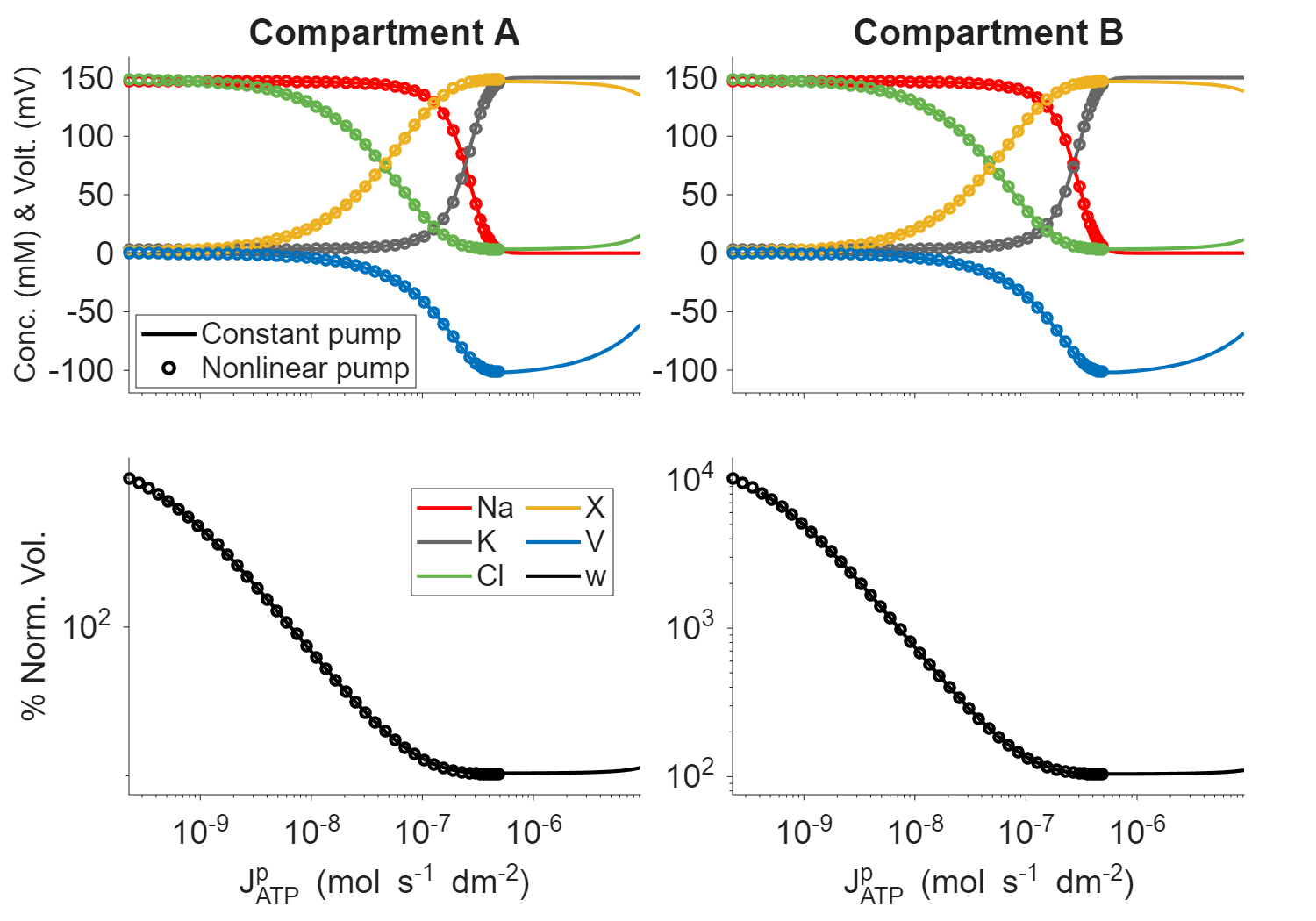}
    \caption{
   Steady state values of ABp systems are plotted as functions of the ATP consumption rate, for the constant NKA (solid curve) and the nonlinear NKA model given in \eqref{eq:nonlinear_pump} (circles).
   }
   \label{fig:constant_vs_nonlinear_pumps}
\end{figure}

%~~~~~~~~~~~~~~~~~~~~~~~~~~~~~~~~~~~~~~~~~~~~~~~~~~~~~~~~~~~~~~~~~~~
\subsection{Existence of the steady states of coupled PLEs} \label{sec:existence}

To establish the existence of steady states {(denoted by the superscript $^{ss}$)} in 2-compartment systems where the NKA pump is located on basolateral membrane, we prove the following lemma which will provide a possible range for the pump rate $\pumprate$ and impose the following assumptions.

\begin{lemma}\label{lemma:existence_ss}
For ions Na\textsuperscript{+} and K\textsuperscript{+}, assume that at least two of $g_{\text{ion},1}$, $g_{\text{ion},2}$, and $g_{\text{ion},p}$ are strictly positive. For 
 $j\in\{A,B\}$ let 
\begin{align}
    \Cpbl &= \cle \left( \nae e^{- \frac{F}{RT}\gamma_{Na} \, \pumprate \, \Arbl \Gnajbl} + \ke e^{\frac{F}{RT} \gamma_{K} \, \pumprate \, \Arbl\Gkjbl} \right), \label{eq:Cpbl}
\end{align}
where the constants $\Gionjbl$ are defined as
\begin{subequations}
    \begin{align}
    &\GionAbl = \dfrac{\gionap + \gionpara}{\gionbl \, \gionap  + \gionbl \, \gionpara  + \gionap \, \gionpara} \\
    &\GionBbl = \dfrac{\gionap}{\gionbl \, \gionap  + \gionbl \, \gionpara  + \gionap \, \gionpara} . \label{eq:GionB}
    \end{align}
    \label{eq:Gionjp1}
\end{subequations}
Then, if 
\begin{equation}
    \gamma_{Na}\nae\Gnajbl - \gamma_{K}\ke\Gkjbl > 0 \label{eq:na-k>0_bl}, 
\end{equation}
there exists $\pMax>0$ such that for $\pumprate<\pMax$, 
\begin{equation}\label{eq:necessary_for_existence}
4\Cpbl-\Ose^2<0.
\end{equation}
Hence, for  $0<\pumprate<\pMax$,  $\xj^{ss}>0$, where
\begin{align}
    \xj^{ss} &= \begin{cases}
                \dfrac{\Ose - \sqrt{4\left( 1 - \zj^2 \right) \Cpbl + \Ose^2 \zj^2}}{1 - \zj^2} & \text{ if } \zj^2 \neq 1\\
                \dfrac{\Ose^2 - 4\Cpbl}{2 \Ose} & \text{ if } \zj^2 = 1 .
           \end{cases}
    \label{eq:xjssp1}
\end{align}
\end{lemma}

In Lemma~\ref{lemma:existence_ss}, the conditions on $g_{\text{ion},1}$, $g_{\text{ion},2}$, and $g_{\text{ion},p}$ ensure that the denominator of $\Gionjbl$ is non-zero. For instance, it is permissible to have $\ggnabl = \ggkap = 0$, provided that $\ggnaap\;\ggnapara > 0$ and $\ggkbl \; \ggkpara > 0$.

Note that it is possible to assume $\ggionap = 0$, in which case compartments $A$ and $B$ become decoupled--$A$ transport mechanism remains active while $B$ is passive. Although we exclude this case from our analysis, we note that the steady state values for $A$ and the equilibrium values for $B$ are consistent with those derived in previous studies; see \cite{aminzare2024mathematical}.

\begin{proof}
To show that there exists a $\pMax > 0$ such that for all $\pumprate < \pMax$ we have  
\[
f_{j,1} (\pumprate) := 4\Cpbl - \Ose^2 < 0,
\]
we need to verify the following two conditions:  
(1) $f_{j,1}(0) \leq 0$, and  
(2) $f'_{j,1}(0) < 0$.  
Then, since $f_{j,1}$ is continuous, we conclude that $f_{j,1}(\pumprate) < 0$ for all sufficiently small values of $\pumprate > 0$. We denote the supremum of such values by $\pMax$.

(1) Note that
$
f_{j,1}(0) = 4\C - \Ose^2 = \ye (\ye - \zy^2 \ye - 2\Ose),
$
which, by Lemma~\ref{lemma:equilibrium}, is negative for $\ye > 0$ and is zero when $\ye = 0$.

(2) Next, we compute the derivative:
\[
f'_{j,1}(\pumprate) = 4\cle \frac{F}{RT} \Arbl \left( 
    -\nae \gamma_{Na} \Gnajbl e^{\frac{F}{RT} \gamma_{Na} \, \pumprate \, \Arbl \Gnajbl} 
    + \ke \gamma_K \Gkjbl e^{\frac{F}{RT} \gamma_K \pumprate \, \Arbl \Gkjbl} 
\right).
\]
Evaluating at $\pumprate = 0$, we obtain
\[
f'_{j,1}(0) = 4\cle \frac{F}{RT} \Arbl \left( 
    -\nae \gamma_{Na} \Gnajbl 
    + \ke \gamma_K \Gkjbl 
\right).
\]
By \eqref{eq:na-k>0_bl}, the term in parentheses is negative, and hence $f'_{j,1}(0) < 0$.

A similar argument to that in Lemma~\ref{lemma:equilibrium} shows that $\xj^{ss} > 0$.
\end{proof}

\begin{proposition}\label{proposition:steadystate}
Consider the coupled PLEs \eqref{eq:ode_con_a}--\eqref{eq:ode_water} for 2-compartment systems with 
\[\pnabl = -\gamma_{Na} \, \pumprate \, \Arbl, \pkbl = \gamma_{K} \, \pumprate \, \Arbl, \pnaap=\pkap = 0.\] 
Under the assumptions of Lemma~\ref{lemma:existence_ss},  and the following two assumptions:
\begin{enumerate}
\item at least two of $g_{\text{Cl},1}$, $g_{\text{Cl},2}, g_{\text{Cl},p}$ are strictly positive
\item at least two of $\nu_1,\nu_2,\nu_p$ are strictly positive 
\end{enumerate}
the system admits a stable steady state for $0<\pumprate<\pMax$ , where $\pMax$ is defined in Lemma~\ref{lemma:existence_ss}. The steady states can be described explicitly as follows. 
For $j\in\{A,B\}$:
\begin{subequations}
\begin{align}
    &\naj^{ss} = \dfrac{2 \nae \cle \exp\left(- \frac{F}{RT}\gamma_{Na} \, \pumprate \, \Arbl \Gnajbl \right)}{\Ose + \left( z_j - 1 \right) \xj^{ss}} \label{eq:najss_p1} \\
    &\kj^{ss} = \dfrac{2 \ke \cle \exp\left(\frac{F}{RT} \gamma_{K} \, \pumprate \, \Arbl\Gkjbl\right)}{\Ose + \left( z_j - 1 \right) \xj^{ss}} \label{eq:kjss_p1} \\
    &\clj^{ss} = \frac{1}{2}\left( \Ose + \left( z_j - 1 \right) \xj^{ss} \right) \\
    &\wj^{ss} = \frac{\paramxj}{\xj^{ss}}  \label{eq:wjss_p1} \\
    &\vj^{ss} = \frac{RT}{F}\ln\left( \frac{\Ose + \left( z_j - 1 \right)\xj^{ss}}{2 \cle} \right), \label{eq:vjss_p1}
\end{align}
\label{eq:jss_p1}
\end{subequations}
where $\GionAbl,\GionBbl$ are given in \eqref{eq:Gionjp1} and $\xj^{ss}$ is given in \eqref{eq:xjssp1}. 
\end{proposition}

Note that, unlike in the passive case, $\ye$ is not required to be positive. Indeed, as long as the NKA pump remains active, it compensates for $\ye$ and maintains regulation of the 2-compartment system.

\begin{proof}
 To derive the steady states, we set the right-hand side of Equations \eqref{eq:ode_con_a} and \eqref{eq:ode_con_b} to zero. For sodium, this yields:
\begin{subequations}
\begin{align}
    &0=-\gnabl \left(\vA^{ss}-\enaA^{ss}\right) - \gnaap\left(\left(\vA^{ss}-\enaA^{ss}\right)-\left(\vB^{ss}-\enaB^{ss}\right)\right)-\gamma_{Na} \, \pumprate \, \Arbl \label{eq:set_dnaA_0} \\
    &0=-\gnapara\left(\vB^{ss}-\enaB^{ss}\right)+\gnaap\left(\left(\vA^{ss}-\enaA^{ss}\right)-\left(\vB^{ss}-\enaB^{ss}\right)\right). \label{eq:set_dnaB_0}
\end{align}
\end{subequations}
Summing \eqref{eq:set_dnaA_0} and \eqref{eq:set_dnaB_0}, and assuming that at least one of $\gnabl$, $\gnapara$ is non-zero, we can solve for the electrochemical potential difference in the corresponding compartment. For instance if $\gnabl>0$, then: 
\[\vA^{ss} - \enaA^{ss} = \frac{-\gnapara\left(\vB^{ss} - \enaB^{ss}\right) - \gamma_{Na} \, \pumprate \, \Arbl}{\gnabl}.\]
 Substituting this expression into Equation \eqref{eq:set_dnaB_0} yields 
\begin{subequations}
\begin{align*}
    &0=-(\vB^{ss}-\enaB^{ss})(\gnabl\gnaap+\gnabl\gnapara+\gnaap\gnapara)- \gnaap\gamma_{Na} \, \pumprate \, \Arbl .
\end{align*}
\end{subequations}
We can now solve for  $\vB^{ss}-\enaB^{ss}$ and then for $\vA^{ss}-\enaA^{ss}$. Similarly,  
we can solve for the electrochemical potential of potassium and chloride in each compartment.
In summary, for $j\in\{A,B\}$,  the electrochemical potential of sodium,  potassium, and chloride become:
\begin{subequations}
    \begin{align}
        &\vj^{ss} - E_{Na,j}^{ss} = -\gamma_{Na} \, \pumprate \, \Arbl \Gnajbl  =: \Cnajpbl \label{eq:vj-enaj}\\
        &\vj^{ss} - E_{K,j}^{ss} = \gamma_{Na} \, \pumprate \, \Arbl\Gkjbl  =: \Ckjpbl \label{eq:vj-ekj}\\
        &\vj^{ss} - E_{Cl,j}^{ss}  = 0 =: \Ccljpbl\label{eq:vj-eclj}
    \end{align}
    \label{eq:vj-eionj1}
\end{subequations}
where, for ions Na\textsuperscript{+} and K\textsuperscript{+}, the constants $\Gionjbl$ are defined in \eqref{eq:Gionjp1} as
    \begin{align*}
    &\GionAbl = \dfrac{\gionap + \gionpara}{\gionbl \, \gionap  + \gionbl \, \gionpara  + \gionap \, \gionpara} \\
    &\GionBbl = \dfrac{\gionap}{\gionbl \, \gionap  + \gionbl \, \gionpara  + \gionap \, \gionpara} .
    \end{align*}
  
  We denote the right-hand side of the expressions in Equation~\eqref{eq:vj-eionj1} by $\Cionjpbl$. By isolating the Nernst potentials in Equation~\eqref{eq:vj-eionj1} and solving for the intracellular concentrations $\ionj^{ss}$, we obtain the following expression for $[\text{ion}]_j^{ss}$ in terms of $\vj^{ss}$:
\begin{align}
    [\text{ion}]_j^{ss} = [\text{ion}]_e \exp\left( \frac{F}{RT}\left( \Cionjpbl - \vj^{ss} \right) \right). \label{eq:ionj1}
\end{align}

Next, we derive an expression for $\vj^{ss}$ in terms of fixed parameters. Since the addition of the pump mechanism does not affect the volume equations, we can proceed similarly to the previous section and derive Equation~\eqref{eq:Ose_p_EN_EQ} for the steady state values using the expression for $[\text{ion}]_j^{ss}$ given in Equation~\eqref{eq:ionj1}:
 \begin{align}
    \naj^{ss} + \kj^{ss} = \frac{1}{2}\left( \Ose - \left( z_{j} + 1 \right) \xj^{ss} \right) .\label{eq:Ose_p_EN_SS}
\end{align}

Similarly, since $\vj^{ss} = E_{\text{Cl},j}^{ss}$, the voltage $\vj^{ss}$ can be derived analogously to Equation~\eqref{eq:voltage_EQ}, resulting in Equation~\eqref{eq:vjss_p1}.

Substituting Equations~\eqref{eq:vjss_p1} and~\eqref{eq:ionj1} into Equation~\eqref{eq:Ose_p_EN_SS} yields:
\begin{equation*}
\begin{split}
    \frac{1}{2}\left( \Ose - \left( 1 + z_j \right) \xj^{ss} \right) =\; &\nae \exp\left( \frac{F}{RT} \left( \Cnajpbl - \frac{RT}{F} \ln\left( \frac{\Ose + \left( z_j - 1 \right) \xj^{ss}}{2 \cle} \right) \right) \right) \\
    &+ \ke \exp\left( \frac{F}{RT} \left( \Ckjpbl - \frac{RT}{F} \ln\left( \frac{\Ose + \left( z_j - 1 \right) \xj^{ss}}{2 \cle} \right) \right) \right),
\end{split}
\end{equation*}
\noindent
which simplifies to Equation~\eqref{eq:xjssp1} with the definition in Equation~\eqref{eq:Cpbl}. Since $0 < \pumprate < \pMax$, Lemma~\ref{lemma:existence_ss} guarantees that $\xj^{ss}$ is strictly positive, ensuring the steady states are well-defined. 

The remainder of Equation~\eqref{eq:jss_p1} follows from Equations~\eqref{eq:vjss_p1},~\eqref{eq:ionj1}, and the identity $w_j = \paramxj / [X_j]$.

The proof of stability is discussed in detail in Section~\ref{sec:localstability}.
\end{proof}

Cellular volume can be challenging to compare across figures because it is measured in \KT{liters}  
and varies widely in magnitude between systems. To address this, we normalize volume by $\wnorm$, defined as $\wA^{ss}$ evaluated at default parameter values with the pump rate fixed at $\pumprate = 1$ $\mu$A dm$^{-2}$ on the basolateral membrane:
\begin{equation}
    \wnorm = \left.\wA^{ss}\right|_{\text{defaults},\, \pumprate = 1 \text{ $\mu$A dm$^{-2}$}} . \label{eq:wnorm}
\end{equation}

%~~~~~~~~~~~~
\subsubsection{A note on the possible range of the pump rate}

In Lemma~\ref{lemma:existence_ss}, we showed that if $\gamma_{Na}\nae\Gnajbl - \gamma_{K}\ke\Gkjbl > 0$, then there exists a range for the pump rate per unit area $\pumprate$, denoted by $(0,\pMax)$, for which steady states exist and, in particular, the compartment volumes remain finite. Define
\begin{equation} \label{eq:fj1} 
    f_{j,1}\left(\pumprate\right):= 4 \, \Cpbl(\pumprate)-\Ose^2. 
\end{equation}
Our default parameter values are chosen so that \eqref{eq:na-k>0_bl} holds. Therefore, we expect that $f_{j,1}(\pumprate) < 0$ over some interval of $\pumprate$. Figure~\ref{fig:fj1fig} shows the plots of $f_{j,1}(\pumprate)$ for compartments $A$ and $B$, which appear nearly identical. The right panel provides a zoomed-in view near their positive roots, denoted $p_{\max,A}$ and $p_{\max,B}$. Solving for these roots numerically yields $p_{\max,A} \approx 3,402$ $\mu$A dm$^{-2}$ and $p_{\max,B} \approx 4,082$ $\mu$A dm$^{-2}$. In the simulations presented in the following sections, we fix $\pMax = \min\{\pMaxAbl, \pMaxBbl\}$ and vary $\pumprate$ within the range $(0, \pMax)$.

\begin{figure}[h!]
    \centering
    \includegraphics[width=0.6\linewidth]{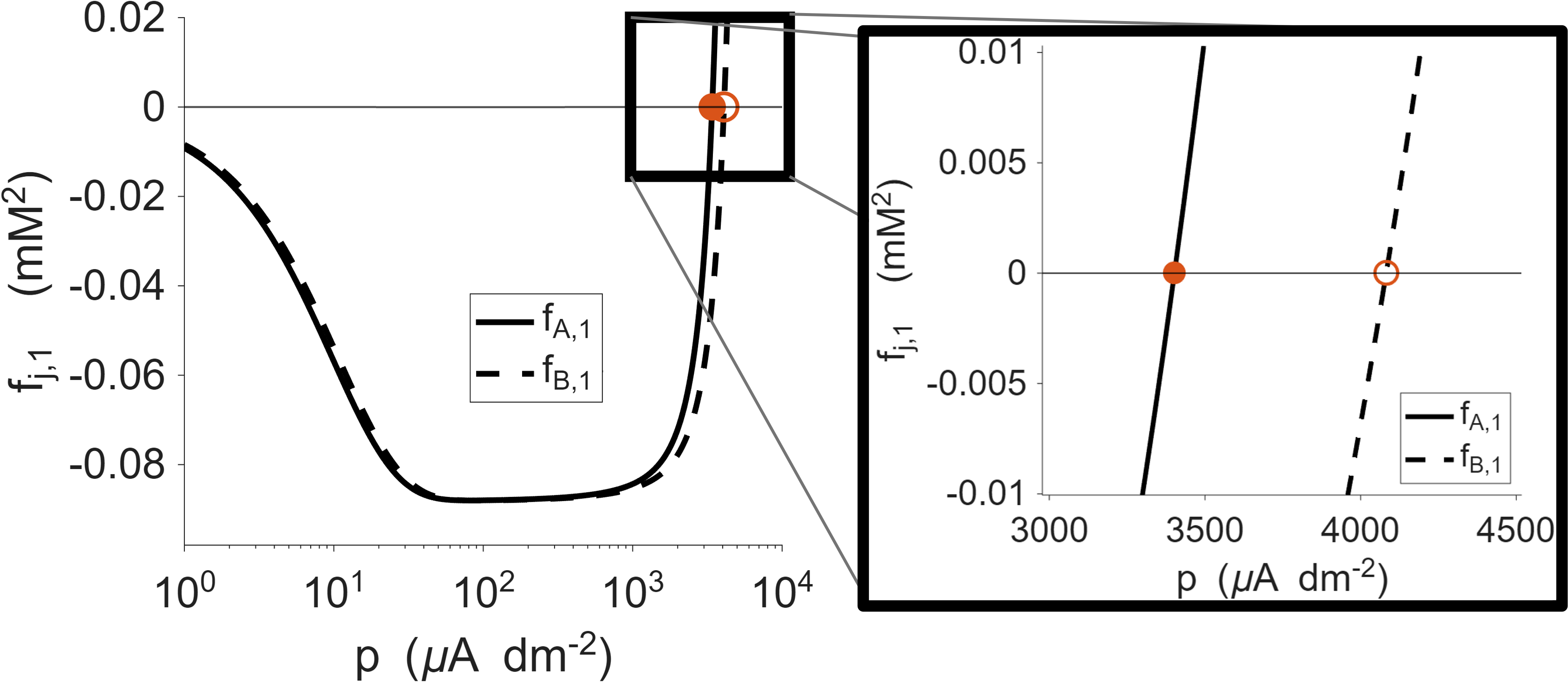}
    \caption{
    The admissible range of $\pumprate$ is computed numerically  
    by plotting $f_{j,1}(\pumprate) := 4 \, \Cpbl(\pumprate) - \Ose^2$ and identifying the interval $(0, \pMax)$ where $f_{j,1}(\pumprate) < 0$. In the zoomed-in panel on the right, the open and closed circles denote the roots of $f_{A,1}$ and $f_{B,1}$, respectively.} 
    \label{fig:fj1fig}
\end{figure}

The optimal pump rate $\pumprate$, denoted by $p_{\min}$, is defined as the value of $\pumprate$ that minimizes the steady state compartment volume $\wj^{ss}$. To compute $p_{\min}$, we differentiate $\wj^{ss}(\pumprate)$ from Equation~\eqref{eq:wjss_p1} with respect to $\pumprate$ and solve for the value at which the derivative is zero. For any $z_j$ and $j \in {A, B}$, the minimizing pump rate $p_{\min,j}$ is given by:
\begin{align}
p_{\min,j} &:= \frac{RT}{F} \left( \dfrac{1}{\gamma_K \Gkjbl + \gamma_{Na} \Gnajbl} \right) \ln\left( \dfrac{\gamma_{Na} \nae \Gnajbl}{\gamma_{K} \ke \Gkjbl} \right) \; / \; \arbl.
\label{eq:pminp1j}
\end{align}
For $p_{\min,j}$ to be positive, the argument of the logarithm must exceed 1; that is,
 \[\dfrac{\gamma_{Na} \nae \Gnajbl}{\gamma_{K} \ke \Gkjbl}>1,\]
  which is precisely the condition in Equation~\eqref{eq:na-k>0_bl}, a necessary condition for the existence of steady states. An analogous condition has also been used for single-compartment models in previous studies \cite{mori2012mathematical, aminzare2024mathematical}.
Finally, we note that $p_{\min,j}$ depends on the sodium and potassium conductances, the stoichiometry of the NKA pump, and the extracellular sodium and potassium concentrations. Since the values of $p_{\min,j}$ play a key role in understanding the system's behavior, we compute them numerically using our default parameter values and use them in the simulations presented in the following sections:
$
p_{\min,A} \approx 80 \text{ $\mu$A dm$^{-2}$}$ and  $p_{\min,B} \approx 90 \text{ $\mu$A dm$^{-2}$}.
$

%~~~~~~~~~~~~
\subsubsection{A note on ion flows of a ABp system at steady states}
For an ABp system with the NKA pump located on the apical surface, the ionic fluxes across the three interfaces are
\begin{subequations}
\begin{align}
    &\fluxionbl = \zion \gionbl \left( \vA - \eionA \right)+ \pionbl, \\
    &\fluxionap = - \zion \gionap \left[ \left( \vA - \eionA \right) - \left( \vB - \eionB \right) \right] , \\
    &\fluxionpara = - \zion \gionpara \left( \vB - \eionB \right),
\end{align}
\label{eq:fluxes_Baso}
\end{subequations}
where 
$ \pnabl=-\gamma_{Na} \, \pumprate \, \Arbl$, $\pkbl = \gamma_{K} \, \pumprate \, \Arbl,$ and $\pclbl=0.$

Using the steady state values derived in Equation~\eqref{eq:jss_p1}, we show that the absolute values of total ion and water flows remain equal across the apical, basolateral, and paracellular pathways, that is,
\[
\fluxionbl^{ss} = \fluxionap^{ss} = \fluxionpara^{ss}.
\]
As illustrated in Figure~\ref{fig:total_flow} (left panel), sodium flow forms a counterclockwise loop: sodium moves from compartment $B$ to $A$ through the apical surface, exits compartment $A$ through the basolateral membrane, and returns to compartment $B$ via the paracellular pathway.
The total fluxes across these three membranes are plotted in the right panel as functions of the pump rate $\pumprate$. As $\pumprate$ increases, the total flux magnitude across each interface increases while remaining equal, confirming conservation of flow at the steady states.
Similarly, the potassium flow follows the opposite direction (clockwise), moving from the extracellular environment into compartment $A$, then to $B$ through the basolateral membrane and apical surface, respectively, and finally leaving compartment $B$ through the paracellular pathway. As with sodium, the absolute value of the total potassium flow increases with the NKA pump rate $\pumprate$, as seen in the right panel. {
Moreover, \ZA{due to the $3{:}2$ pump stoichiometry,} the magnitude of the Na$^+$ flux is slightly larger than that of the K$^+$ flux over the plotted range at each interface.
}
The flows of chloride and water remain zero across all membranes (data not shown).
\begin{figure}[h!]
    \centering
    \includegraphics[width=0.7\linewidth]{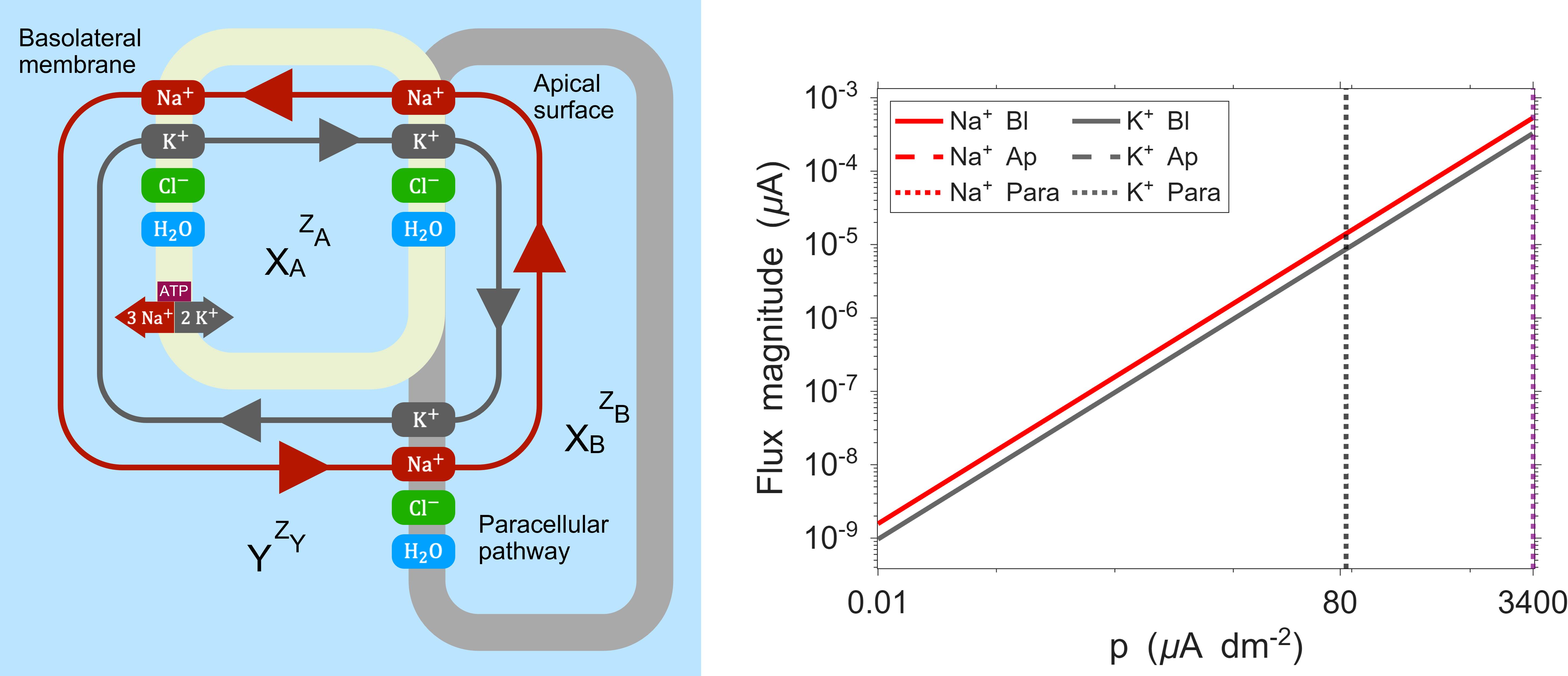}
   \caption{
    \textbf{Sodium {and potassium} flow across the membranes at steady state.}
   The arrows in the left schematic of the ABp system illustrate counterclockwise sodium and clockwise potassium transport loops. Sodium moves from compartment $B$ to $A$ across the apical surface, exits $A$ through the basolateral membrane, and returns to $B$ via the paracellular pathway, whereas potassium moves from $A$ to $B$ across the apical surface, from $B$ to the ISF through the paracellular pathway, and back to $A$ across the basolateral membrane. 
The right panel shows that the magnitudes of the total sodium (red) and potassium (gray) fluxes across the three interfaces increase with the pump rate $\pumprate$, {with Na$^+$ fluxes slightly larger than K$^+$ fluxes over the plotted range and equal across their respective interfaces}. 
{Here, \emph{total} flux denotes the area-integrated flux (in $\mu$A) and is not normalized by membrane or pathway area.}
}
    \label{fig:total_flow}
\end{figure}

%~~~~~~~~~~~~~~~~~~~~~~~~~~~~~~~~~~~~~~~~~~~~~~~~~~~~~~~~~~
\subsection{Local stability analysis of the steady states of coupled PLEs}
\label{sec:localstability}

In this section, we show that the ABp system's steady state is locally stable, i.e., the eigenvalues of the corresponding Jacobian matrix of PLEs evaluated at the steady state derived in Proposition~\ref{proposition:steadystate} have negative real parts.

 We first perform a change of variable and rewrite PLEs in these new variables. Let 
\begin{align*}
 &x_1=F\wA\naA,\; y_1=F\wA\kA,
\; z_1=F\wA\clA,\; \uA=F\wA, \\
 &x_2=F\wB\naB,\; y_2=F\wB\kB, \;z_2=F\wB\clB, \; \uB=F\wB.
\end{align*} 
Here, $x_i$, $y_i$, and $z_i$ have units of coulombs and represent the signed amounts of ionic charge per species in each compartment, and \KT{$u_i$} is the amount of charge you get per unit molarity in the compartment. This improves numerical conditioning as terms begin to vary.
In the new variables, the PLEs become:
\begin{subequations}
\begin{align}
    &\frac{d x_1}{dt} = -\gnabl \left( \vA - \enaA \right) + \dna - \gamma_{Na} \, \pumprate \, \Arbl ,  \\
    &\frac{d y_1}{dt}  = -\gkbl \left( \vA - \ekA \right) + \dk + \gamma_{K} \, \pumprate \, \Arbl ,  \\
    &\frac{d z_1}{dt} = \gclbl \left( \vA - \eclA \right) +  \dcl, \\
    &\frac{ d \uA }{ dt } = \nu_1 (\OsA-F\Ose) + \nu_2 (\OsA-\OsB) \label{eq:dwA_cv}, \\
    &\frac{d x_2 }{dt} = -\gnapara \left( \vB - \enaB \right) - \dna ,  \\
    &\frac{dy_2}{dt}  = -\gkpara\left( \vB - \ekB \right) -\dk , \label{eq:dkB_cv} \\
    &\frac{d z_2}{dt} = \gclpara \left( \vB - \eclB \right) - \dcl, \\
    &\frac{ d \uB }{ dt } = \nu_p(\OsB-F\Ose) - \nu_2 (\OsA-\OsB) .
\end{align}
\label{eq:ode_changevariable}
\end{subequations}
In terms of the new variables, we have the following expressions for membrane potential and osmolarity in compartments A and B:
\begin{align*}
&\vA=\frac{1}{C_{m,A}}(x_1+y_1+z_1+F\paramxA),\; \vB=\frac{1}{C_{m,B}}(x_2+y_2+z_2+F\paramxB),\\
&\OsA=\frac{F}{\uA}(x_1+y_1+z_1+F\paramxA), \; \OsB=\frac{F}{\uB}(x_2+y_2+z_2+F\paramxB).
\end{align*}
The Nernst potential for each ion in compartments A and B can be expressed as follows:
\[\eionA=\frac{RT}{F}\ln\left(\frac{\uA\ione}{\text{ion}_1}\right), \; \eionB=\frac{RT}{F}\ln\left(\frac{\uB\ione}{\text{ion}_2}\right).\]

Linearizing Equation~\eqref{eq:ode_changevariable} around the steady state (derived in Equation~\eqref{eq:jss_p1}) yields an $8\times 8$ Jacobian matrix $M$, which can be written in block form as
\begin{align}\label{eq:Jacobian}
M = \begin{bmatrix}M_A & M_{AB} \\ M_{BA} & M_B \end{bmatrix}
\end{align}
where $M_A, M_{AB}, M_{BA},$ and $M_B$ are $4\times 4$ matrices defined below. 
\[M_A= {\begin{bmatrix}
        -G_{A}^{Na}(\frac{1}{C_{m,A}}+\frac{RT}{F}\frac{1}{x_1}) & -G_{A}^{Na}\frac{1}{C_{m,A}} & G_{A}^{Na}\frac{1}{C_{m,A}} & G_{A}^{Na}\frac{RT}{F} \frac{1}{\uA} \\
        -G_{A}^{K}\frac{1}{C_{m,A}} & -G_{A}^{K}(\frac{1}{C_{m,A}}+\frac{RT}{F}\frac{1}{y_1}) & G_{A}^{K}\frac{1}{C_{m,A}} & G_{A}^{K}\frac{RT}{F}\frac{1}{\uA} \\
        G_{A}^{Cl}\frac{1}{C_{m,A}} & G_{A}^{Cl}\frac{1}{C_{m,A}} & -G_{A}^{Cl}(\frac{1}{C_{m,A}}+\frac{RT}{F}\frac{1}{z_1}) & G_{A}^{Cl}\frac{RT}{F}\frac{1}{\uA} \\
        \nu_A\frac{F}{\uA} & \nu_A\frac{F}{\uA} & \nu_A\frac{F}{\uA} & -\nu_A\sigma_A\frac{F}{\uA^2}
    \end{bmatrix}}
\]
where $G_{A}^{Na}=\gnabl+\gnaap$, $G_{A}^{K}=\gkbl+\gkap$, $G_{A}^{Cl}=\gclbl+\gclap$, $\nu_A=\nu_1+\nu_2$, and $\sigma_A=x_1+y_1+z_1+F\paramxA$.
\[
    M_B={\begin{bmatrix}
        -G_{B}^{Na}(\frac{1}{C_{m,B}}+\frac{RT}{F}\frac{1}{x_2}) & -G_{B}^{Na}\frac{1}{C_{m,B}} & G_{B}^{Na}\frac{1}{C_{m,B}} & G_{B}^{Na}\frac{RT}{F}\frac{1}{\uB} \\
        -G_{B}^{K}\frac{1}{C_{m,B}} & -G_{B}^{K}(\frac{1}{C_{m,B}}+\frac{RT}{F}\frac{1}{y_2}) & G_{B}^{K}\frac{1}{C_{m,B}} & G_{B}^{K}\frac{RT}{F}\frac{1}{\uB} \\
        G_{B}^{Cl}\frac{1}{C_{m,B}} & G_{B}^{Cl}\frac{1}{C_{m,B}} & -G_{B}^{Cl}(\frac{1}{C_{m,B}}+\frac{RT}{F}\frac{1}{z_2}) & G_{B}^{Cl}\frac{RT}{F}\frac{1}{\uB} \\
        \nu_B\frac{F}{\uB} & \nu_B\frac{F}{\uB} & \nu_B\frac{F}{\uB} & -\nu_B\sigma_B\frac{F}{\uB^2}
    \end{bmatrix}},
\]
where $G_{B}^{Na}=\gnapara+\gnaap$, $G_{B}^{K}=\gkpara+\gkap$, $G_{B}^{Cl}=\gclpara+\gclap$, $\nu_B=\nu_p+\nu_2$, and $\sigma_B=x_2+y_2+z_2+F\paramxB$.
\[
   M_{AB}= \begin{bmatrix}
        \gnaap(\frac{1}{C_{m,B}} + \frac{RT}{F}\frac{1}{x_2}) & \gnaap\frac{1}{C_{m,B}} & -\gnaap\frac{1}{C_{m,B}} & -\gnaap\frac{RT}{F}\frac{1}{\uB} \\
        \gkap\frac{1}{C_{m,B}} & \gkap(\frac{1}{C_{m,B}}+\frac{RT}{F}\frac{1}{y_2}) & -\gkap\frac{1}{C_{m,B}} & -\gkap\frac{RT}{F}\frac{1}{\uB} \\ 
        -\gclap\frac{1}{C_{m,B}} & -\gclap\frac{1}{C_{m,B}} & \gclap(\frac{1}{C_{m,B}}+\frac{RT}{F}\frac{1}{z_2}) & -\gclap\frac{RT}{F}\frac{1}{\uB} \\ 
        -\nu_2\frac{F}{\uB} & -\nu_2\frac{F}{\uB} & -\nu_2\frac{F}{\uB} & \nu_2\sigma_B\frac{F}{\uB^2} 
    \end{bmatrix} ,
\]
and 
\[
   M_{BA}= \begin{bmatrix}
        \gnaap(\frac{1}{C_{m,A}}+\frac{RT}{F}\frac{1}{x_1}) & \gnaap\frac{1}{C_{m,A}} & -\gnaap\frac{1}{C_{m,A}} & -\gnaap\frac{RT}{F}\frac{1}{\uA} \\
        \gkap\frac{1}{C_{m,A}} & \gkap(\frac{1}{C_{m,A}}+\frac{RT}{F}\frac{1}{y_1}) & -\gkap\frac{1}{C_{m,A}} & -\gkap\frac{RT}{F}\frac{1}{\uA} \\ 
        -\gclap\frac{1}{C_{m,A}} & -\gclap\frac{1}{C_{m,A}} & \gclap(\frac{1}{C_{m,A}}+\frac{RT}{F}\frac{1}{z_1}) & -\gclap\frac{RT}{F}\frac{1}{\uA} \\ 
        -\nu_2\frac{F}{\uA} & -\nu_2\frac{F}{\uA} & -\nu_2\frac{F}{\uA} & \nu_2\sigma_A\frac{F}{\uA^2} 
            \end{bmatrix} .
\]

For now, we fix all parameters at their default values (Tables~\ref{tab:constants}--\ref{tab:pump}) and vary only the pump rate $\pumprate$ from near zero to $p_{max}:= \Pmax/a_1 = \max\{\PmaxAbl, \PmaxBbl\}$, the full range over which steady states 
exist (See Figure~\ref{fig:eigval_p1}, left panel). Later, in Section~\ref{sec:dependencyOnParams}, we will analyze local stability as the other parameters vary. For any fixed $\pumprate$, the  right panel of Figure~\ref{fig:eigval_p1} shows the largest eigenvalue 
among the eight real eigenvalues of $M$. 
For all values of $\pumprate$ between $0$ and $p_{max}$, the eigenvalues are real and negative, indicating local stability. {Since the Jacobian $M$ is nearly singular throughout this range, we compute its eigenvalues symbolically using variable-precision arithmetic to reduce round-off errors caused by ill-conditioning~\cite{stor2015accurate}.}

\begin{figure}[h!]
    \centering
    \includegraphics[width=0.43\linewidth]{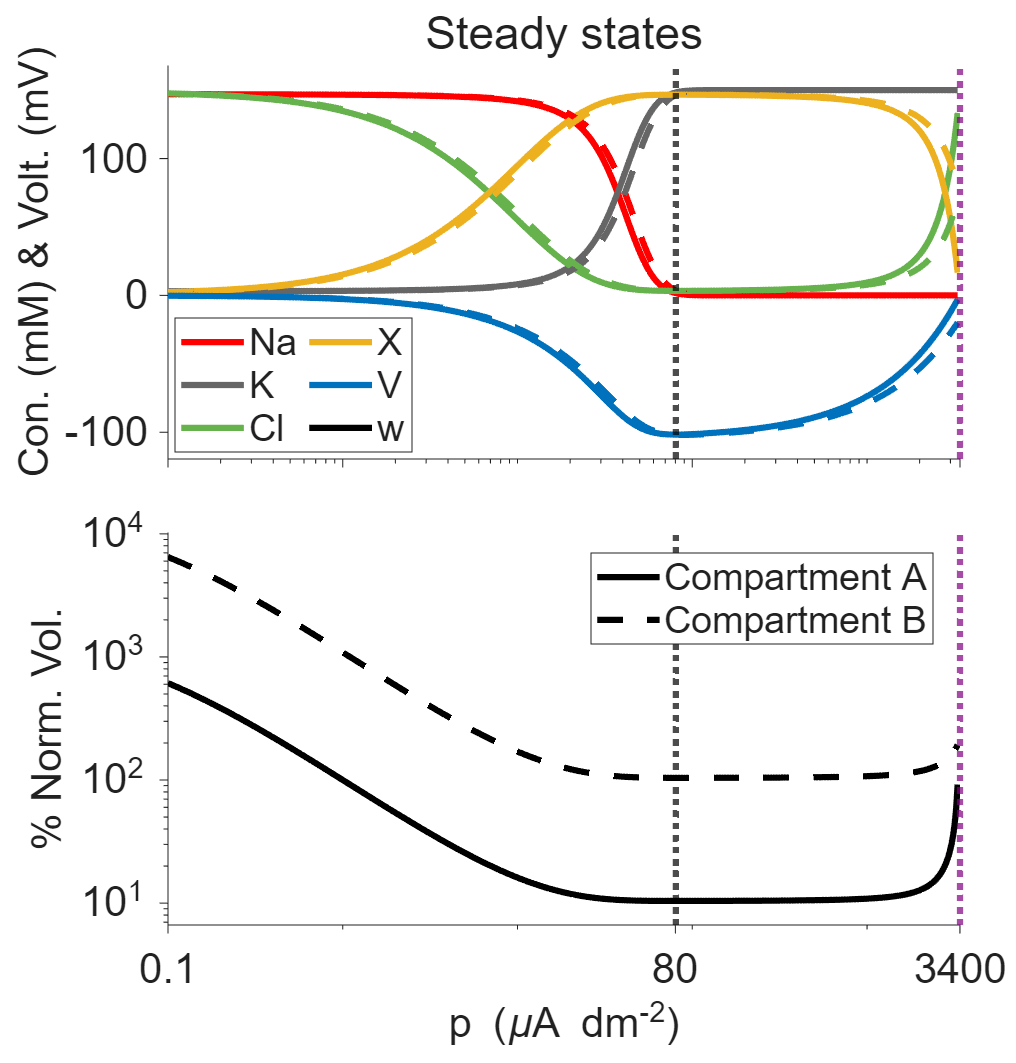} \quad
    \includegraphics[width=0.46\linewidth]{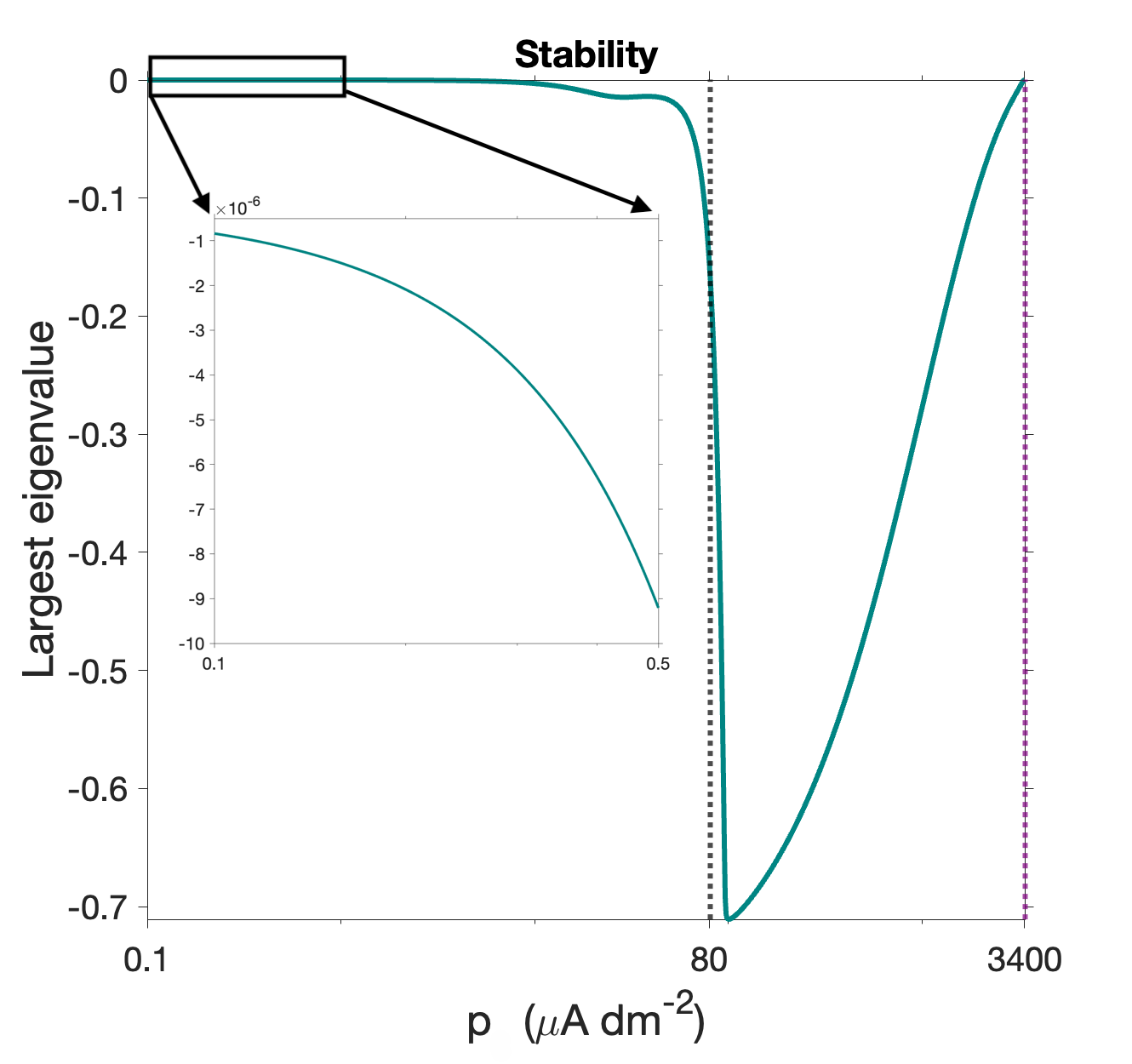} 
    % C:\Users\Lenovo Device\Dropbox\Zahra-Kerry\Summer 25_Eigenvalues\After-Zoom-New\Eigenvalues_vs_p1.m
  \caption{
        Steady state values of the ABp system  (Equation~\eqref{eq:jss_p1})  and the largest eigenvalue of the corresponding Jacobian matrix (Equation~\eqref{eq:Jacobian})   are plotted as functions of the pump rate $\pumprate$. 
        As shown in the right panel (and in the zoomed-in inset), all eigenvalues are negative for any given $\pumprate$, indicating that the steady states are locally stable. 
        The vertical lines mark $p_{\min}$ and $\pMax$, corresponding to the pump rates at which the steady state volume reaches its minimum and to the largest value of $\pumprate$ for which steady states exist, respectively. 
        The horizontal axis is shown on a logarithmic scale, and the vertical axis on a linear scale. All other parameters are fixed at their default values given in Tables~\ref{tab:constants}--\ref{tab:pump}.
    }
 \label{fig:eigval_p1}
\end{figure}

%~~~~~~~~~~~~~~~~~~~~~~~~~~~~~~~~~~~~~~~~~~~~~~~~~~~~~~~~~~~~~~~~~~~
\subsection{Special ABp systems} \label{sec:special_systems}

  In what follows, we examine two biologically relevant classes of ABp systems: \AK{The Koefoed-Johnsen-Ussing (KJU) model of epithelial transport} and cell-organelle systems. The KJU case is represented by an ABp model with $\ggnabl = \ggkap = 0$. The cell-organelle case corresponds to an AB model, obtained from the ABp framework by removing the permeable paracellular pathway, i.e., by setting $\ggionpara = \nu_p = 0$. These two cases parallel classical studies of epithelial transport and pump-leak cell volume regulation \cite{tosteson1960regulation}, and they provide mathematically tractable model systems for our analysis.

 \medskip 

\noindent \textbf{KJU model of epithelial transport.}  
In an epithelial system, there is typically a tube-like structure composed of a single layer of epithelial cells joined by tight junctions, forming a paracellular pathway. This arrangement is common in the gut, kidney, and other organs. The apical surface, often called the mucosal surface, faces the lumen of the tube, while the basolateral membrane, also known as the serosal surface, faces the interstitial space. Epithelia can be studied using an Ussing chamber, where the tissue is clamped between two chambers that allow control of the voltage across the epithelium and measurement of ionic current and sometimes volume flux. When the voltage is controlled and the current is measured, the system operates in the short-circuit configuration.

The foundational model for epithelial transport is the Koefoed-Johnsen-Ussing (KJU) model, which was extended by \cite[Section 2.8.2]{keener2009mathematical} to incorporate impermeant intracellular solutes, which were not included in the original formulation. In the KJU model, sodium conductance on the basolateral membrane  and potassium conductance on the apical surface are set to zero ($\ggnabl=\ggkap = 0$), reflecting physiological asymmetries. In this configuration, sodium and water are transported from the apical to the basolateral side and can flow continuously at steady state. 
A related single-cell epithelial model that includes a paracellular pathway was developed by Weinstein and collaborators \cite{Strieter1990}, capturing volume regulation via volume-activated chloride permeability.

In what follows, we compare the steady-state values of an KJU system with those of a general ABp system. To ensure that Assumption~\ref{assumption:equilibrium} holds, we require $\ggkbl,\,\ggkpara > 0$ and $\ggnaap,\,\ggnapara > 0$, which guarantees that Equation~\eqref{eq:na-k>0_bl} is satisfied. We further assume that at least two of $\nu_1$, $\nu_2$, and $\nu_p$ are nonzero. Under these conditions, the steady states of the KJU system---like those of the ABp system---are given explicitly by Equation~\eqref{eq:jss_p1} and therefore depend on several model parameters. Since $\ggkap = 0$, we have $G_{K,A} = 1/\ggkbl$ and $G_{K,B} = 0$. Therefore, even though $\ggkpara > 0$, the steady-state values of compartments $A$ and $B$ depend only on the potassium conductance along the basolateral membrane. In particular, no $K^+$ leaks into or out of compartment~$B$.

In the KJU system, $f_{B,1}$ has a unique negative root, so $f_{B,1}(\pumprate) < 0$ for all $\pumprate \ge 0$, that is, $\pMaxBbl = \infty$. In this case, we set $\pMax = \pMaxAbl < \infty$. Furthermore, since $\ggkap = 0$, we obtain $\pMinB = \infty$, and since $\ggnabl = 0$, $\pMinA$ becomes much smaller (approximately ten times smaller) than $\pMinA$ in a general ABp system. This indicates that the KJU system requires less energy to minimize $\wA^{ss}$ compared to the ABp system. This observation is consistent with the results shown in Figures~\ref{fig:heatmapA_p_v_gna1} and \ref{fig:robustness_gna1}, where small $\ggnabl$ decreases the volume steady states in the presence of a weak pump mechanism (i.e., small $\pumprate$).

Figure~\ref{fig:epithelial} depicts a schematic diagram for the KJU system (left), in which the sodium channels on the basolateral membrane and the potassium channels on the apical surface are removed. The steady states are shown as functions of the pump rate $\pumprate$ for $\pumprate < \pMaxAbl$. The value of $\pMinA$ is indicated by a vertical dashed line and is very small. Finally, the spectral abscissa {(the maximum real part among the eigenvalues of the Jacobian at steady state)} is plotted as a function of $\pumprate$, remaining negative and thereby reflecting the stability of the steady states.

\begin{figure}[h!]
    \centering
    \includegraphics[width=1\linewidth]{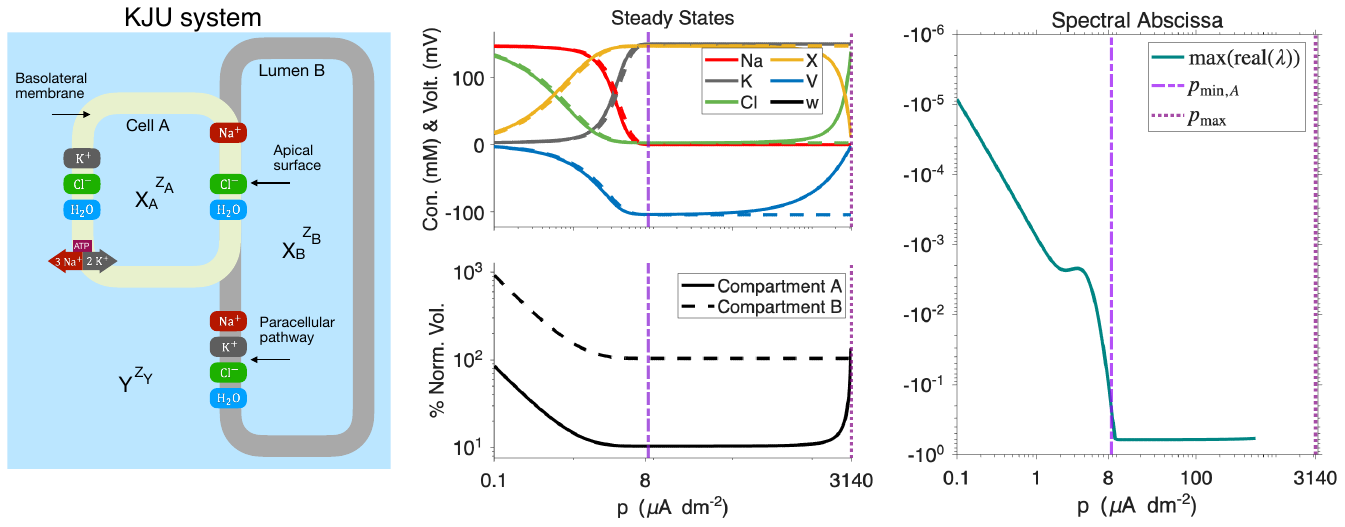}
    \caption{
        \textbf{KJU system.}
        The schematic diagram is shown, and the steady-state values together with the corresponding spectral abscissa are plotted as functions of the pump rate~$\pumprate$. 
        {Conductances are fixed at $\ggnabl=\ggkap=0$ with remaining parameters fixed at their default value.}
        }
    \label{fig:epithelial}
\end{figure}

\bigskip
\noindent \textbf{Cell-organelle AB systems. } 
Intracellular organelles such as the nucleus, mitochondria, and endoplasmic reticulum are {enclosed by lipid bilayer membranes that are topologically separate from the plasma membrane and that partition their interiors from the cytoplasm}. Because these organelles are surrounded by selective membranes and exchange ions with the cytoplasm, their volumes, ionic concentrations, and voltage can be modeled using pump-leak equations. In our framework, the organelle interior plays the role of compartment $B$ in an AB system with no paracellular pathway, that is, with $\ggionpara = \nu_p = 0$. 

In what follows, we consider an AB system, a special case of the ABp system without a paracellular pathway. Since in an AB system we have $\ggionpara = \nu_p = 0$, the existence of a steady state requires that all other ion conductances are positive and that $\nu_1 > 0$, $\nu_2 > 0$.

A simple calculation shows that when $\ggionpara = 0$, we have $\GionAbl = \GionBbl= 1/\gionbl$, and therefore $\mathcal{C}_A = \mathcal{C}_B$ (see Equation \eqref{eq:Cpbl}). The steady states depend only on two ion conductances, $\ggnabl$ and $\ggkbl$. Increasing $\ggnabl$ increases the steady state volume, while increasing $\ggkbl$ decreases it. Hence, if $\zA = \zB$, the two compartments admit identical steady states (except for volume, which differs by a factor $\paramxA/\paramxB$). We conclude that the two-compartment AB system effectively behaves like a single cell. Thus, similar to the equilibrium values of Equation \eqref{eq:EQ}, the only parameters that distinguish the two compartments are $\zA$ and $\zB$.

Figure~\ref{fig:AB_system} depicts a schematic diagram for the AB system (left), in which the channels on the paracellular pathway are removed. The steady states are shown as functions of the pump rate $\pumprate$.  The value of $\pMinA$ is indicated by a vertical dashed line and is very small. Finally, the spectral abscissa is plotted as a function of $\pumprate$, remaining negative and thereby reflecting the stability of the steady states.

\begin{figure}[h!]
    \centering
   \includegraphics[width=1\linewidth]{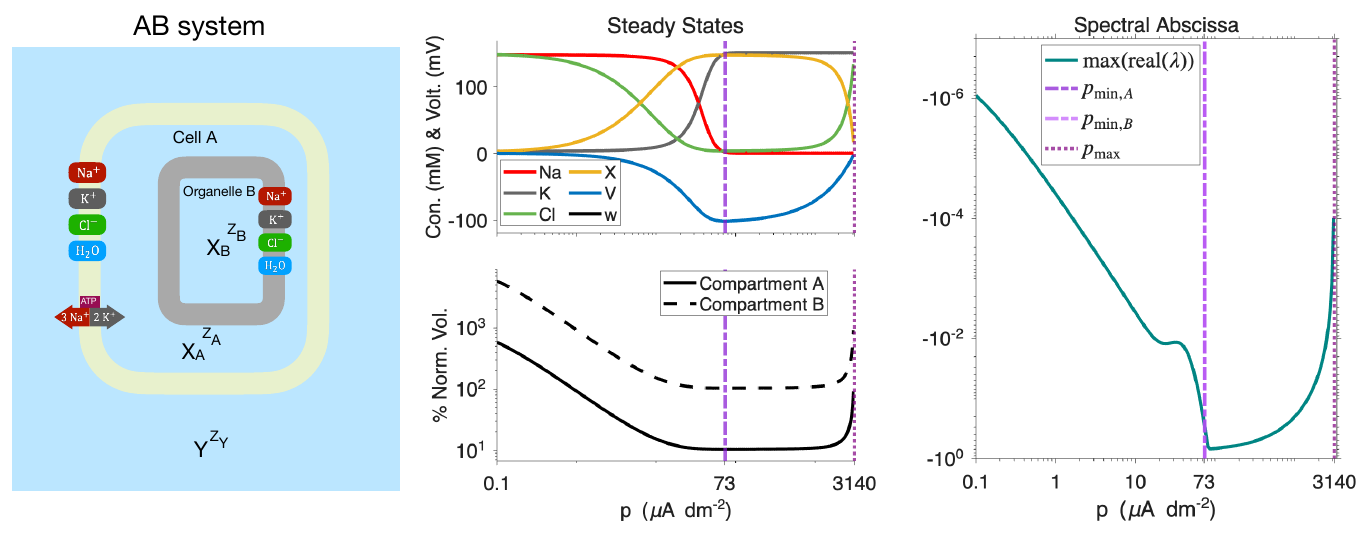}
        \caption{
\textbf{AB system.}
The schematic diagram is shown, and the steady-state values together with the corresponding spectral abscissa are plotted as functions of the pump rate~$\pumprate$.
{Assigning $\paramxB=\paramxA/10$ will make volume of compartment B a tenth of A's.}
       }
    \label{fig:AB_system}
\end{figure}

%~~~~~~~~~~~~~~~~~~~~~~~~~~~~~~~~~~~~~~~~~~~~~~~~~~~~~~~~~~~~~~~~~~~
\section{Effects of Biophysical Parameters on ABp System}
 \label{sec:dependencyOnParams}
In Equation~\eqref{eq:jss_p1} of Section~\ref{sec:active}, we explicitly derived the steady-state values of the coupled PLEs for general ABp systems. These expressions depend on many parameters in a highly nonlinear way, making it difficult to identify parameter influences by direct inspection. 

In the following three subsections, we use sensitivity analysis to:  
(1) identify the key parameters that alter the steady-state behavior of ABp systems, with emphasis on mechanisms that regulate compartmental volumes;  
(2) examine the robustness of these behaviors under parameter fluctuations; and  
(3) determine the parameters for which the ABp system maintains homeostasis 
 when they vary.

\subsection{\ZA{Key mechanisms that regulate the volumes in ABp systems}}
 \label{sec:dependencyOnParams_key_mechanisms}
 
We already saw that the NKA pump is a key mechanism that regulates cell and lumen volumes, and our results in the ABp system are consistent with those shown before in single-compartment models \cite{aminzare2024mathematical, mori2012mathematical, keener2009mathematical, FRASER2007336}. In this section, we use variance-based global sensitivity analysis (VBGSA) to identify the key model parameters that exhibit compensatory mechanisms for maintaining volume regulation when the NKA pump activity is reduced.

We consider the interaction among nine key parameters: the pump rate; sodium and potassium conductances across the basolateral  membrane, apical surface, and paracellular pathway; extracellular concentrations--primarily the combination of sodium and chloride--and the temperature. Note that chloride conductances do not appear in the steady state formulas, although they are essential for the operation of the PLM. Moreover, to preserve electroneutrality and positive osmolarity, not all extracellular concentrations can be varied simultaneously. All other model parameters are fixed at their default values. 

To vary $\nae$ and $\cle$ in a charge-balanced manner, we introduce the reparameterization
\begin{equation}
    \nacle = \nae + \cle \label{eq:nacle} .
\end{equation}
This additional parameter enables us to simultaneously vary $\nae$ and $\cle$ while maintaining electroneutrality since the charges of the two ions cancel each other out. The perturbed parameters are $\nae = \nae^{\text{default}} + \nacle$ and $\cle = \cle^{\text{default}} + \nacle$. 

For the numerical work, we consider a $9$-dimensional parameter vector $\theta=(\theta_k)\in\mathbb{R}^9$ where the index $k=1,\ldots,9$ refers respectively to the following parameters 
\[\pumprate,\; 
\ggnabl,\; 
\ggkbl,\;
\ggnaap,\;  
\ggkap,\;
\ggnapara,\;  
\ggkpara,\;
T,\;
\nacle,
\]
with default values denoted by $\theta_k^*$. Although we did not assign a default value to $\pumprate$ (unlike the other parameters), because we are investigating compensatory parameters under reduced pump activity, we set a small pump rate value here, namely $\theta_{\pumprate}^* = 1\,\mu\text{A}\,\text{dm}^{-2}$. 

For each parameter, we define a perturbation interval around its default value as
\begin{subequations}
\begin{align}
    %\theta_k \in 
    \Theta_k &:= \left[\theta_k^*\,\left(1-\beta\,\epsilon_k\right),\; \theta_k^* \,\left(1-\beta\,\epsilon_k\right)^{-1}\right], \quad\quad \mbox{for}\quad k\notin\left\{T, \, \nacle\right\},
    \label{eq:parameter_interval_log}\\
    %\theta_k\in
    \Theta_k &:= \left[\theta_k^*-\beta\,\epsilon_k,\; \theta_k^*+\beta\,\epsilon_k\right], \hspace{2.6cm} \mbox{for}\quad k\in\left\{T, \, \nacle\right\}.
    \label{eq:parameter_interval_linear}
\end{align}
\label{eq:parameter_interval}  
\end{subequations}
Here the scaling factor $\beta\in\KT{[0,1]}$ serves as a uniform perturbation level across all nine parameters in the model, while the values of $\epsilon_k$ are fixed as
\begin{equation}\label{eq:epsilon_values}
\begin{aligned}
    &\epsilon_1 = 0.1, \; \epsilon_2=0.25, \;
    \epsilon_3= 
    \epsilon_4=
    \epsilon_5= 
    \epsilon_6= \epsilon_7=0.5, \; \epsilon_8=25 \,^\circ\text{C}, \; \epsilon_9=25\, \text{mM}.
\end{aligned}
\end{equation}

We then define the parameter hypercube, denoted by $\mathcal{H}^{9}$, as
\begin{equation}
    \mathcal{H}^{9} \;:=\; \{ \theta \in \mathbb{R}^9 \, : \, \theta_k\in\Theta_k, \;\; \forall\, k\}.
    \label{eq:parameter_hypercube}
\end{equation}
From the parameter hypercube $\mathcal{H}^{9}$, we draw $N = 10^6$ parameter points for VBGSA and $N = 10^3$ parameter points for the robustness 
which will be used later in Section~\ref{sec:robustness}. 
We use Latin hypercube sampling (LHS) to sample from $\mathcal{H}^9$.
The sampling procedure is described in Section~\ref{sec:appendix_sampling_VBGSA}. 

VBGSA quantifies how uncertainty in parameters contributes to the variance of an output of interest (here, a steady state value of cell volume). It is a global, rather than local, sensitivity analysis, as the perturbation of multiple parameters is performed simultaneously \cite{JANSEN199935,saltelli2008}. For a given outcome value, the total variance is decomposed into contributions from each parameter or group of parameters and compared to the total variance in the output. This decomposition is summarized by Sobol indices, which are dimensionless numbers ranging from 0 to 1. A Sobol index close to 1 indicates that the parameter has a strong influence on the steady state, while an index close to 0 indicates a weak influence. Sobol indices partition output variance into first- and total-order indices (shown by blue and orange 
bars in Figure~\ref{fig:sobol_wA__0}), which indicate the contributions of each parameter \cite{SALTELLI2010259,DELA2022111159}.
\begin{figure}[h!]
    \centering
    \includegraphics[width=0.43\linewidth]{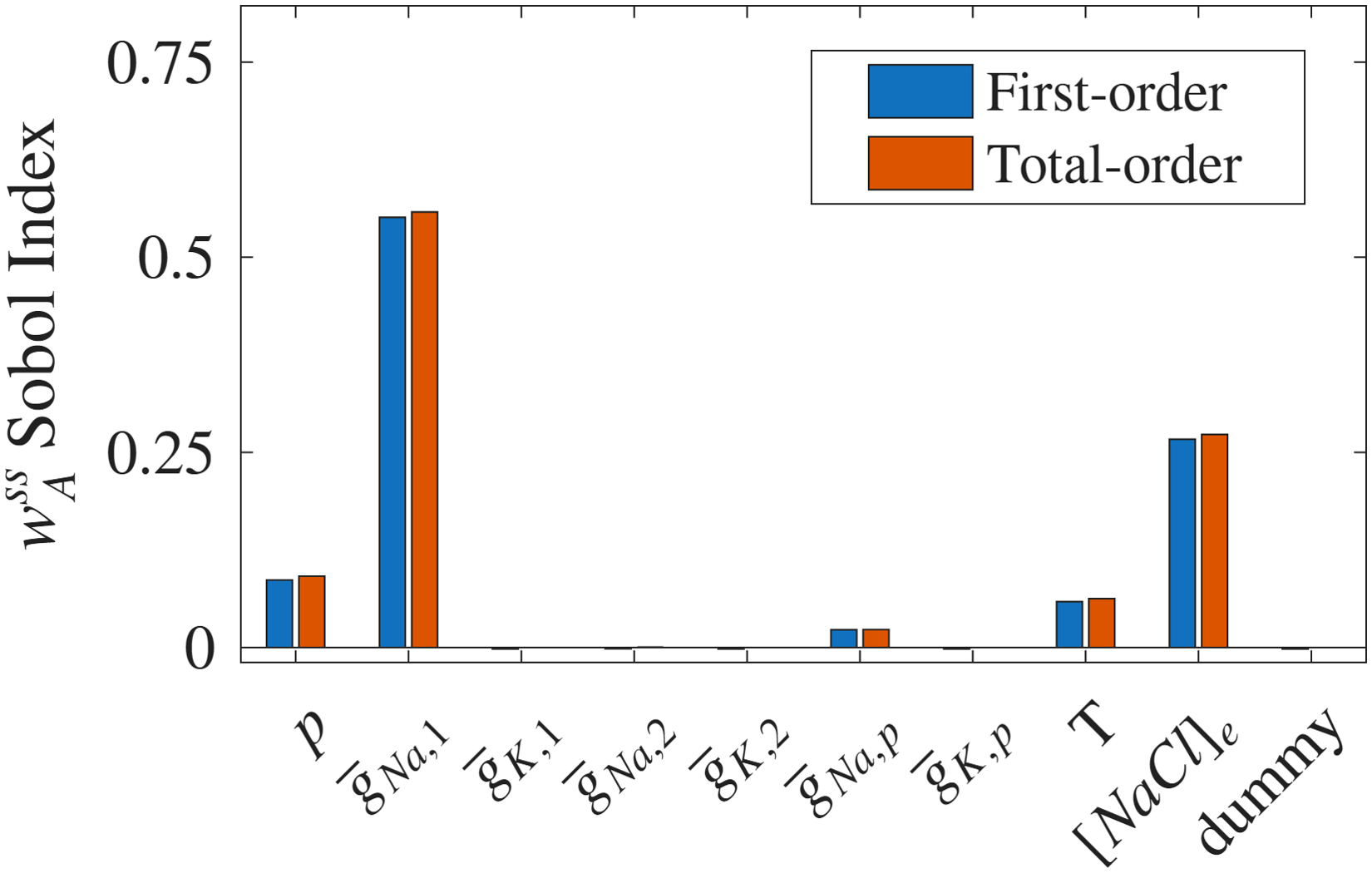}\qquad
    \includegraphics[width=0.43\linewidth]{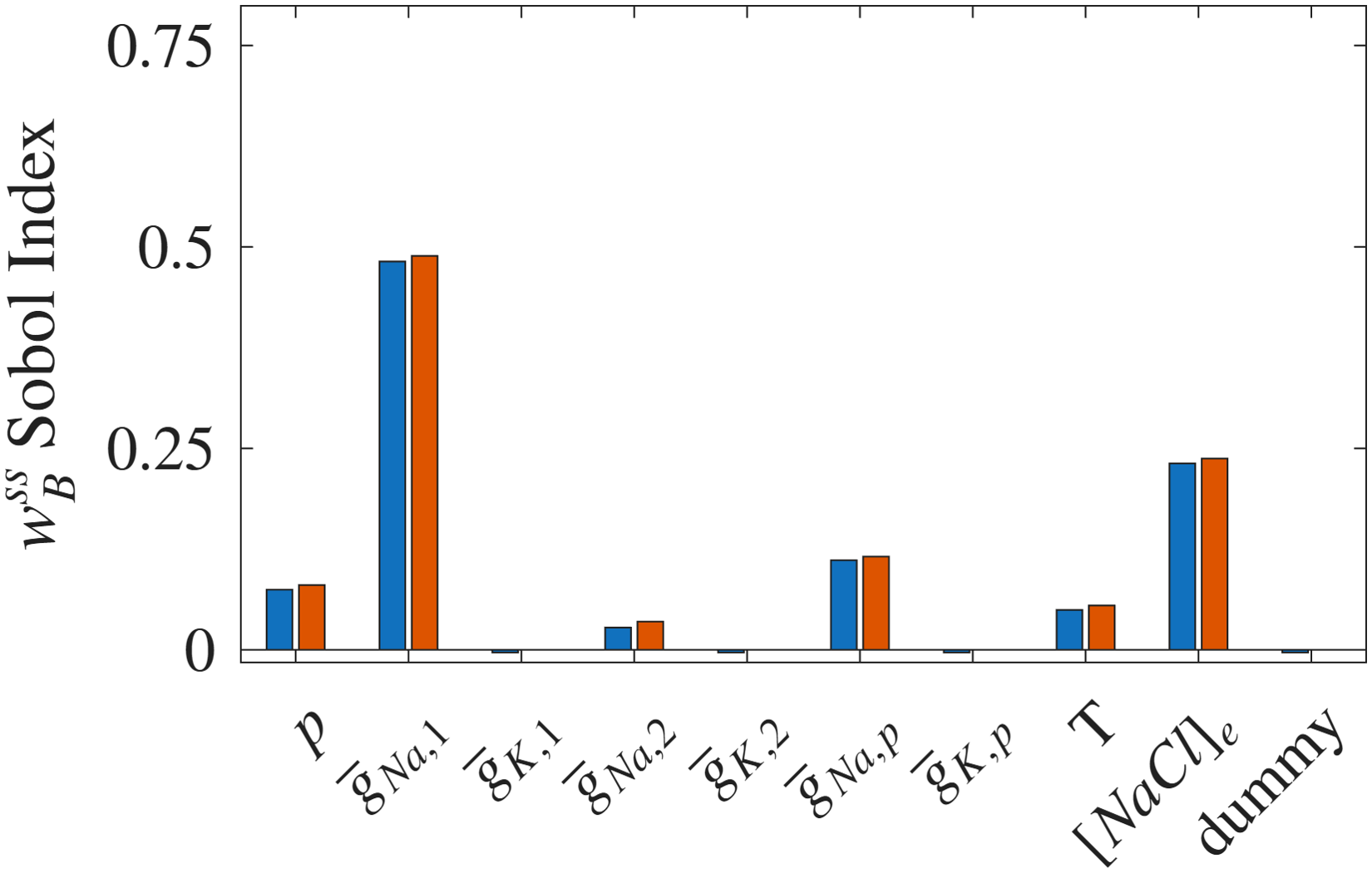}
    \caption{
        First- and total-order Sobol indices of the steady-state volumes of compartments $A$ (left) and $B$ (right) are shown as the nine parameters are varied as described in~\eqref{eq:parameter_interval}. The Sobol indices indicate that $\pumprate$, $\ggnabl$, $T$, and $\nacle$ contribute most to the variance in $\wA^{ss}$, with the largest contribution arising from $\ggnabl$ (highest bars) when $\pumprate$ is small {(using default value $1$ $\mu$A dm$^{-2}$)}. Similarly, the Sobol indices of $\wB^{ss}$ show that these same parameters are the most influential, with additional contributions from $\ggnaap$ and substantially larger contributions from $\ggnapara$.}
    \label{fig:sobol_wA__0}
\end{figure}

In Figure~\ref{fig:sobol_wA__0}, we use VBGSA with $10^6$ parameter points to compute the Sobol indices and show that the variance in the steady-state volume of compartment~A is primarily attributable to variations in $\pumprate$, $\ggnabl$, $\nacle$, and $T$. Similarly, the variance in the steady-state volume of compartment~B is mainly attributable to changes in $\pumprate$, $\ggnabl$, $\ggnaap$, $\ggnapara$, $\nacle$, and $T$.

Although the Sobol index diagram identifies which parameters influence the steady-state volume, it does not indicate whether these changes increase or decrease the volume. To address this, we perturb only two parameters at a time---$\pumprate$ and one of $\ggnabl$, $\ggnaap$, $\ggnapara$, $\nacle$, or $T$---and plot the resulting steady-state values as heat maps.

\ZA{
The Sobol indices shown in Figure~\ref{fig:sobol_wA__0} are computed over a restricted range of the pump rate, namely $\pumprate \in [0.1, 10]$~($\mu$A\,dm$^{-2}$). This restriction is imposed because, for larger values of $\pumprate$, variations in the pump rate account for a disproportionately large fraction of the output variance—particularly in volume-related quantities—thereby masking the contributions of other parameters in the Sobol index analysis.

In contrast, for the heat-map analysis that follows, we consider a broader range of pump rates, $\pumprate \in [0.1, 40]$~($\mu$A\,dm$^{-2}$), in order to ensure physiologically meaningful system behavior. In this regime, the model yields ion concentrations and membrane potentials within acceptable ranges, specifically sodium concentrations between 1--30~mM, potassium concentrations between 50--100~mM, chloride concentrations between 1--30~mM, and membrane potentials between $-100$ and $-40$~mV.
}

{\ZA{Motivated by the Sobol sensitivity analysis, we next investigate} whether modulating sodium conductance can serve as an additional or alternative mechanism for modulating volume, particularly when the pump rate is low. This is indeed viable, as seen in in vivo experiments, which often use tetrodotoxin and saxitoxin to act on voltage-gated sodium channels, thereby attenuating sodium conductance \cite{narahashi1964tetrodotoxin, chrachri1997voltage, maltsev1998novel}. We vary the basolateral sodium conductance $\ggnabl$ between $10^{-2}$ to $10^{2}$ mS dm$^{-2}$, which is a physiologically relevant range \cite{weinstein2012mathematical,weinstein2015mathematical}. In separate simulations, we perform the same variation for the apical and paracellular sodium conductances, $\ggnaap$ and $\ggnapara$.}

Figures~\ref{fig:heatmapA_p_v_gna1}--\ref{fig:heatmapA_p_v_gnaP} depict the steady-state values as heat maps over small pump rates, namely $\pumprate \in [0.1, 150]$ 
 (in $\mu$A\,dm$^{-2}$), and sodium conductances $\ggnabl$, $\ggnaap$, and $\ggnapara \in [0.01, 100]$ (in {mS dm$^{-2}$}), respectively. The left four panels in each figure show the changes in intracellular ion concentrations and voltage, and the right panel shows the corresponding volume. These figures allow us to examine how modulation of sodium conductances affects steady state values and whether such modulation could serve as an alternative mechanism to the pump for changing the volume.
\begin{figure}[h!]
    \centering
    \includegraphics[width=1\linewidth]{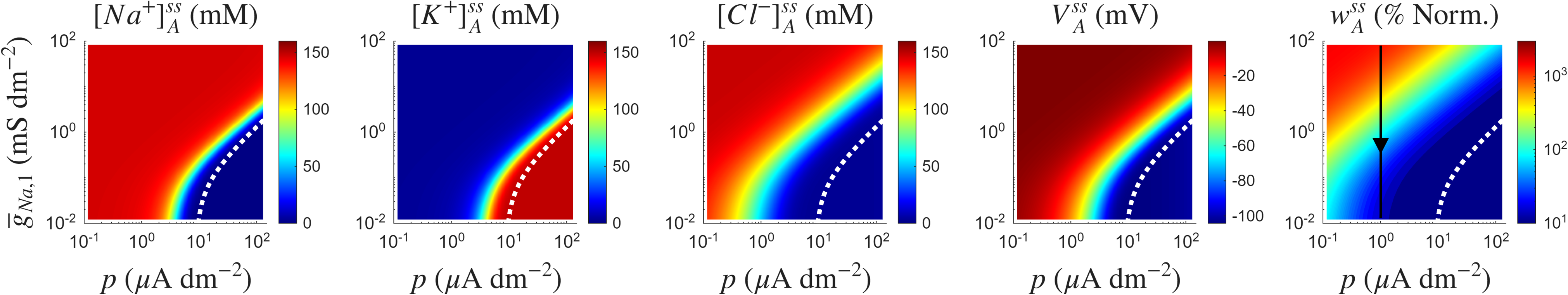}
    \caption{
          Steady-state values of compartment $A$ as functions of $\pumprate$, $\ggnabl$, shown as heat maps. The black vertical  line indicates that when $\pumprate$ is small, decreasing $\ggnabl$ decreases $\wA^{ss}$.
          {The dashed white curve traces the locus $\pumprate=\pMinA(\ggnabl)$ at which $\wA^{ss}$ is minimized---computed from Equation \eqref{eq:pminp1j}.}
       }
    \label{fig:heatmapA_p_v_gna1}
\end{figure}

{Focusing first on the basolateral conductance, the right panel of Figure~\ref{fig:heatmapA_p_v_gna1} shows that decreasing $\ggnabl$ decreases the steady-state volume of compartment $A$ (indicated by the vertical black line and arrow).} A similar result holds for compartment $B$ (figure not shown). This suggests that blocking the sodium channel on the basolateral membrane could assist the pump mechanism in reducing the volume of the ABp system. This result is consistent with our previous work on a single cell \cite{aminzare2024mathematical}, where we showed that reducing the sodium conductance lowers $p_{\min}$, the pump-rate value that minimizes the steady-state volume.

%\AK{(The values of the concentrations and voltages in fig 10-13 are not in the right range in compartment A. You should have the following Na+ 1-30 mM, K+  50-100 mM, Cl 1-30 mM and V -100 to -40 mV. These are all approximate but you need to get them in the right ball park.)}

As indicated by the Sobol indices corresponding to $\ggnaap$, the volume of compartment~$A$ remains nearly unchanged, while the volume of compartment~$B$ varies. Figure~\ref{fig:heatmapA_p_v_gna2} shows the steady-state values of both compartments. For a fixed $\pumprate$, decreasing $\ggnaap$ slightly decreases the volume of $A$ but substantially increases the volume of $B$ (see arrows on the right panels).

\begin{figure}[h!]
    \centering
    \includegraphics[width=1\linewidth]{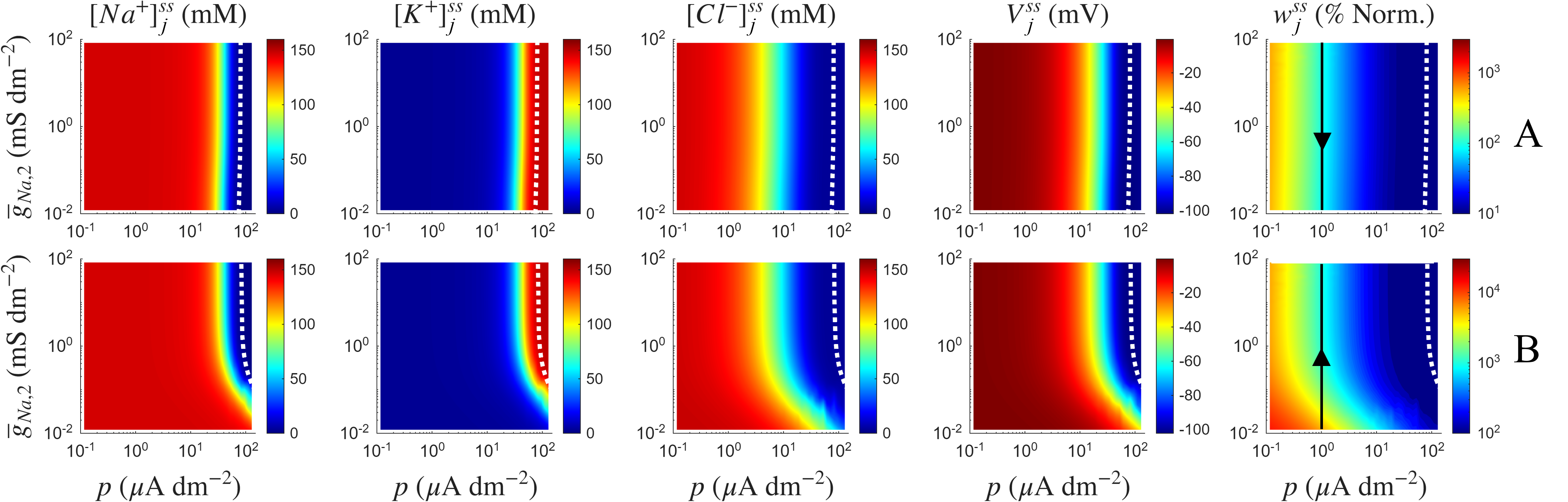}
    \caption{
       Steady-state values of compartments $A$ (top) and $B$ (bottom) as functions of $\pumprate$ and $\ggnaap$, shown as heat maps. The black vertical line and arrow indicate that when $\pumprate$ is small, decreasing $\ggnaap$ slightly decreases $\wA^{ss}$ while substantially increasing $\wB^{ss}$.
       {The dashed white curve overlays the locus $\pumprate=\pMinA(\ggnaap)$ in the top row and $\pumprate=\pMinB(\ggnaap)$ in the bottom row, where $\wA^{ss}$ and $\wB^{ss}$ attain their respective minima---curves computed from Equation \eqref{eq:pminp1j}.}
       }
    \label{fig:heatmapA_p_v_gna2}
\end{figure}

Finally, the Sobol indices corresponding to $\ggnapara$ indicate that the volume of $A$ changes only slightly, while the volume of $B$ is significantly affected. Figure~\ref{fig:heatmapA_p_v_gnaP} shows that the effect of $\ggnapara$ is similar to that of $\ggnabl$: decreasing either $\ggnabl$ or $\ggnapara$ decreases the volume of both compartments, although $\ggnapara$ reduces the volume of $B$ more than that of $A$.

The left panels of Figures~\ref{fig:heatmapA_p_v_gna1}--\ref{fig:heatmapA_p_v_gnaP} show that, for small pump rates, changes in sodium conductance do not appreciably alter the concentrations or membrane potential; a larger pump rate is required to observe such effects.

In summary, all sodium conductances can be used to control the volume of either or both compartments.
%\AK{Leave the following speculation for the discussion - do this for all the cases where you quote experimental data }
%\ZA{Discussed over zoom}
Physiologically, such effects of sodium conductances make sense because limiting sodium entry reduces intracellular sodium accumulation, thereby lowering the osmotic gradient that draws water into the cell. This not only decreases the need for compensatory ion pumping but also prevents osmotic swelling, especially under stress or pathological conditions. {Experimental studies support this prediction: in cardiac myocytes, reducing sodium influx along the basolateral membrane delays cell swelling during NKA pump inhibition \cite{takeuchi2006ionic}; in epithelial and astrocytic cells, amiloride-sensitive ENaC channels on the apical surface have been shown to mediate swelling-activated sodium conductance that can be pharmacologically suppressed to reduce volume in the cell \cite{kizer1997reconstitution}; and knockout mice lacking claudin-2 had a depressed sodium flux across the paracellular pathway and showed a decrease in reabsorption of water into the lumenal space \cite{muto2010claudin}. These studies confirm that sodium conductance is a key regulator of cell volume and a potential therapeutic target.}

\begin{figure}[h!]
    \centering
    \includegraphics[width=1\linewidth]{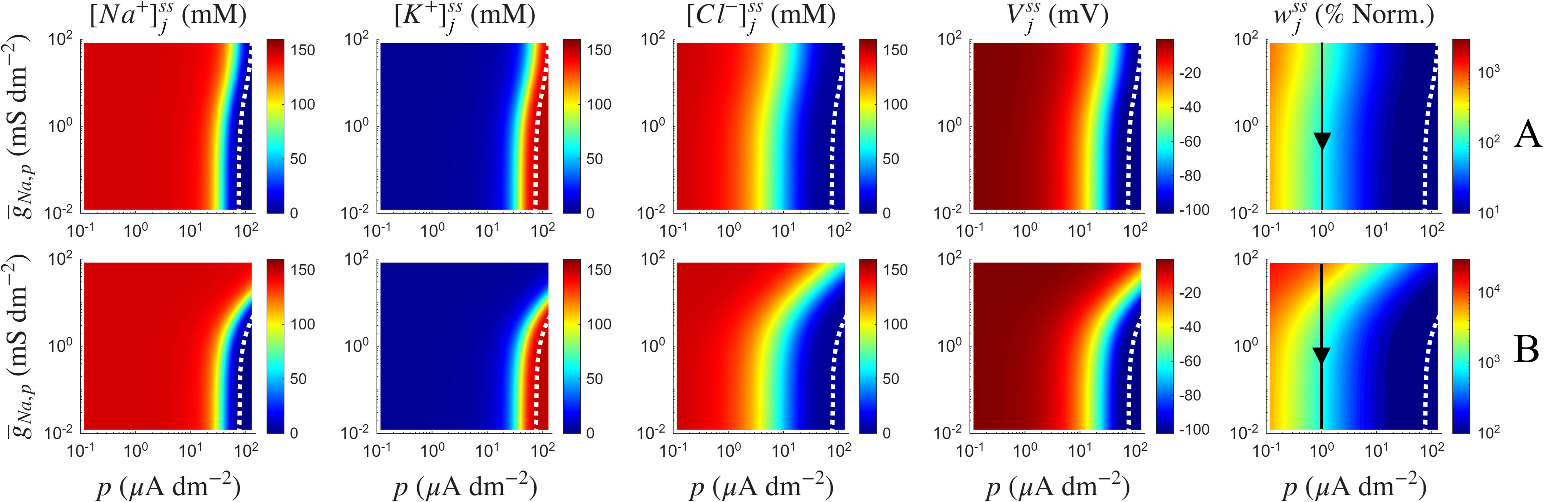}
  \caption{
    Steady-state values of compartments $A$ (top) and $B$ (bottom) as functions of $\pumprate$ and $\ggnapara$, shown as heat maps. The vertical black lines indicate that, for small $\pumprate$, although decreasing $\ggnapara$ decreases both $\wA^{ss}$ and $\wB^{ss}$, its effect on $\wB^{ss}$ is larger.
    {The dashed white curve overlays the locus $\pumprate=\pMinA(\ggnapara)$ in the top row and $\pumprate=\pMinB(\ggnapara)$ in the bottom row, where $\wA^{ss}$ and $\wB^{ss}$ attain their respective minima---curves computed from Equation \eqref{eq:pminp1j}.}
}
    \label{fig:heatmapA_p_v_gnaP}
\end{figure}

Furthermore, the Sobol index diagram predicts that perturbing the extracellular sodium and chloride concentrations---i.e., perturbing $\nacle$ as defined in \eqref{eq:nacle}---affects the steady-state cell volume. Figure~\ref{fig:heatmap_p_vs_Na_e}, where $\pumprate \in [0.1, 10]$ and $\nacle \in [-50,\,50]$, shows that increasing $\nacle$ leads to a decrease in the steady-state cell volume, especially when the pump rate is low, as indicated by the vertical black arrow in the right panel. In addition, the figure shows that intracellular sodium and chloride concentrations increase monotonically as the extracellular sodium-chloride concentration increases. Due to the balanced charges carried by sodium and chloride, the model predicts that the membrane voltage remains unaffected by changes in extracellular sodium-chloride concentration. The intracellular potassium concentration remains constant in order to maintain electroneutrality and positive osmolarity.

\ZA{Changes in extracellular sodium and chloride concentrations alter the osmotic balance across the membrane and drive water influx or efflux, a fundamental principle of cell volume regulation \cite{hoffmann2009physiology}. While extracellular ion concentrations are tightly regulated \emph{in vivo}, such perturbations can be imposed experimentally or clinically, for example through controlled changes in perfusate composition or osmotherapy. This principle underlies clinical interventions such as hypertonic saline treatment for cerebral edema \cite{Simard2007}. 

Within this experimentally controlled context, our theoretical results predict that modulation of extracellular $\nae$ and $\cle$ provides an effective means to influence cell volume. This prediction is consistent with experimental observations in T cells perfused with $2\times$PBS \cite{shu2016study} and with in silico time-series simulations of U937 cells, where the addition of 100~mM NaCl induces a reduction in cell volume \cite{yurinskaya2022cation}. Similar behavior is observed in the steady-state volume $\wB^{ss}$ in our model (results not shown).}

\begin{figure}[h!]
    \centering
    \includegraphics[width=1\linewidth]{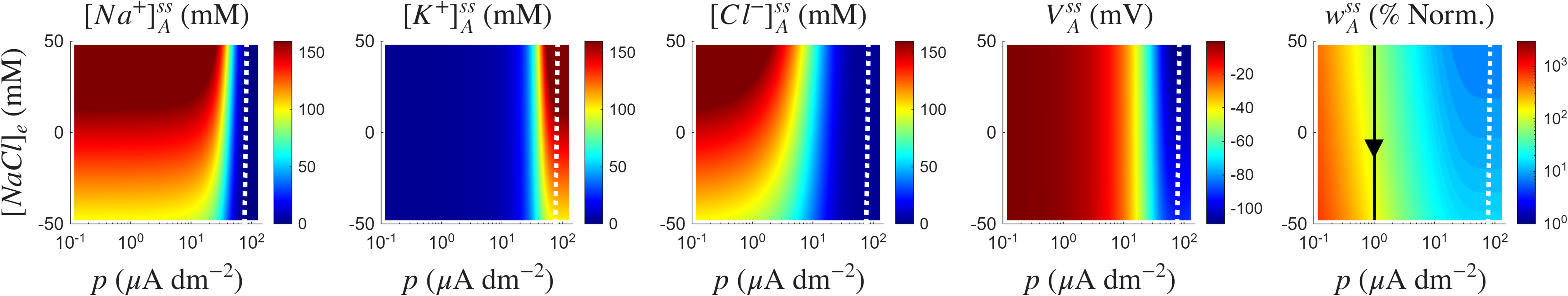}
    \caption{
       Steady-state values of compartment A as functions of $\pumprate$ and $\nacle$, shown as heat maps. The black vertical line indicates that when $\pumprate$ is small, increasing $\nacle$ decreases $\wA^{ss}$.
       {The dashed white curve traces the locus $\pumprate=\pMinA(\nacle)$ at which $\wA^{ss}$ is minimized---computed from Equation \eqref{eq:pminp1j}.}
    }
    \label{fig:heatmap_p_vs_Na_e}
\end{figure}

Finally, as suggested by the Sobol index diagram, we plot the steady states as functions of $\pumprate$ and $T$ in Figure~\ref{fig:heatmap_p_vs_T}. 
{To probe the model's sensitivity, we sweep a wide temperature range ($T\in[-100,100]\;^\circ\text{C}$); biological interpretation is restricted to temperatures where the system remains aqueous and active transport is meaningful.}
Our model predicts that, for a fixed $\pumprate$, lowering the temperature leaves the voltage and intracellular sodium and potassium concentrations essentially unchanged. We observe a modest decrease in intracellular chloride and in cell volume as $T$ decreases. As shown in the right panel of Figure~\ref{fig:heatmap_p_vs_T}, the effect on $\wA^{ss}$ is most pronounced when $\pumprate$ is small (see the vertical arrow). For larger values of $\pumprate$, temperature has only a minor effect on $\wA^{ss}$. The small magnitude of this response over a physiologically reasonable range suggests that temperature modulation may serve as only a modest volume-directed intervention when NKA activity is limited. Similar trends hold for $\wB^{ss}$ (not shown).

\begin{figure}[h!]
    \centering
    \includegraphics[width=1\linewidth]{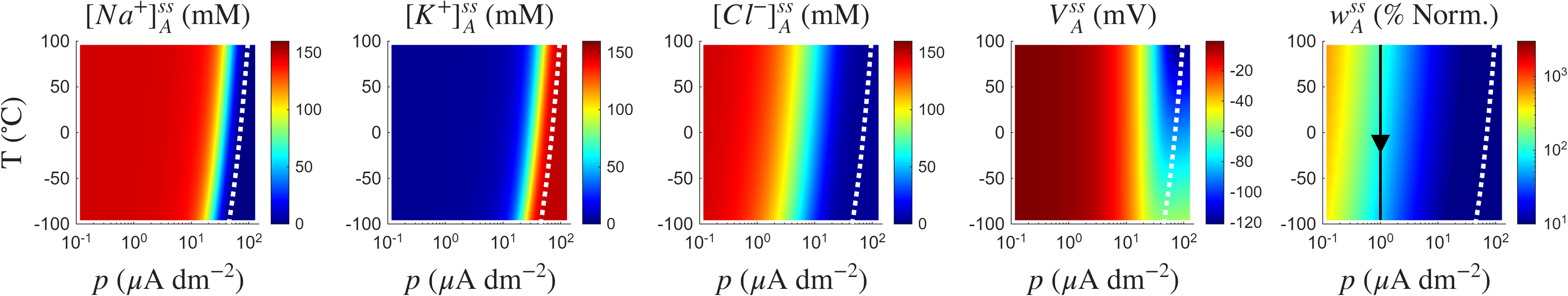}
\caption{
    Steady-state values of compartment $A$ as functions of $\pumprate$ and $T$, shown as heat maps. The vertical black line indicates that, for small $\pumprate$, decreasing $T$ decreases $\wA^{ss}$.
    {The dashed white curve traces the locus $\pumprate=\pMinA(T)$ at which $\wA^{ss}$ is minimized---computed from Equation \eqref{eq:pminp1j}.}
}
    \label{fig:heatmap_p_vs_T}
\end{figure}

In vitro recordings from rat type IIb omohyoid skeletal muscle fibers showed that, at $19$ $^\circ$C, voltages did not significantly differ from voltages at $37$ $^\circ$C or any other tested temperature in between \cite{ruff1999effects}. Preoptic-anterior hypothalamic neurons measured between $30$ and $40$ $^\circ$C showed no significant differences in their resting potential \cite{zhao2005temperature}. Direct mammalian, isomotic steady state evidence for cooling-induced cellular shrinkage is sparse; however, related settings report smaller cellular water content at lower temperatures. For instance, rat hepatocytes shrink during hypothermic storage, particularly when colloid support is absent \cite{neveux1997deletion}, and mild hypothermia reduces cytotoxic cellular swelling in brain tissue \cite{schwab1998mild}.

%~~~~~~~~~~~~~~~~~~~~~~~~~~~~~~~~~~~~~~~~~~~~~~~~~~~~~~~~~~~~~~~~~~~
\subsection{Robustness of multiple volume control mechanisms} \label{sec:robustness}

Robustness is the ability of a system to maintain its functions against perturbations \cite{kitano2004biological}. In biological systems, robustness is essential for survival in changing environments, such as the homeostatic mechanisms that maintain intracellular volume despite fluctuations in extracellular conditions \cite{alberts2022molecular}. 
In Figure~\ref{fig:eigval_p1}, we showed that an ABp system relies on the NKA pump as its primary mechanism for controlling {ion concentrations, voltage and }volume. 
{Furthermore, in Section~\ref{sec:dependencyOnParams_key_mechanisms}, we identified alternative mechanisms that can regulate volume when the pump is weak.} 
In particular, we observed that sodium conductances, extracellular concentrations, and temperature can modulate the volume. Thus far, we have varied only one (in Figure~\ref{fig:eigval_p1}) or two (in Section~\ref{sec:dependencyOnParams_key_mechanisms}) of these parameters at a time to examine their effects. However, in nature, simultaneous fluctuations in multiple parameters are inevitable. In this section, we study the robustness of these behaviors, in the sense of determining whether they persist under perturbations to multiple parameters.

Similar to Figure~\ref{fig:sobol_wA__0}, we simultaneously perturb $\ggionany$, $\nacle$, and $T$ as described in Equation~\eqref{eq:parameter_hypercube}, and plot the analytically computed steady-state values (from Equations~\eqref{eq:xjssp1}--\eqref{eq:jss_p1}) of the ABp system as functions of $\pumprate$ (Figure~\ref{fig:pump_perturb}), {$\ggnaany$} (Figures~\ref{fig:robustness_gna1} and \ref{fig:robustness_gna2_gnap}), and $\nacle$ (Figure~\ref{fig:robustness_Nae}) to assess the robustness of the behaviors observed in Section~\ref{sec:dependencyOnParams_key_mechanisms}.

In these figures, the solid curves, the $10^3$ scatter points, and the dashed curves represent the steady-state values and spectral abscissa of the Jacobian (from \eqref{eq:Jacobian}) evaluated at each parameter value for the default parameters, the $10^3$ parameter vectors (each of dimension 9) sampled via {LHS} as described in Section~\ref{sec:dependencyOnParams_key_mechanisms}, and the Gaussian Process (GP) regressions \cite{williams2006gaussian} fitted to the scatter points, respectively. Note that the regression curve (dashed) serves as a surrogate model for the perturbed data points (scatter). In addition, we plot an error measure that quantifies the difference between the scatter points (perturbed steady states) and the solid curves (default steady states), as described below. 

The perturbations in the parameters cause some amount of variance in the system. We  measure the variance in the steady state and eigenvalues caused by perturbations in parameters by computing the root mean square relative error (RMSRE), \cite{despotovic2016evaluation,webber2017canopy,goccken2016integrating} defined as 
\begin{equation}
    \text{RMSRE} = \sqrt{\mathbb{E} \left[\left(\frac{\bm{y} - \bm{o}}{\bm{o}}\right)^2\right]}.
    \label{eq:RRMSE}
\end{equation}
where $\bm{y}$ and $\bm{o}$ are the steady state value or the largest eigenvalue (evaluated at the steady states) for default parameters and perturbed parameters, respectively. In this section, the RMSRE is computed for 50 values of $\beta\in[0.0001,1]$. 

For visualization purposes, the GP predictive mean, denoted by $\bm{\hat{y}}$, is shown by the dashed curves and interpreted as an estimate of the conditional expectation
$
    \bm{\hat{y}}(\theta_k) \approx \mathbb{E}\!\left[ \bm{o} \mid \theta_k \right],
$
where $\theta_k$ is one of the nine parameters of interest \cite{ratto2007state,saltelli2008}.

 All the figures  confirm that although the RMSRE increases as the perturbation factor $\beta$ grows (top-right panels), the steady-state values of the ABp system still exist and remain locally stable under the considered parameter fluctuations (green curves in the bottom-right panels). 
Overall, the system behavior remains robust, as explained next.

Figure~\ref{fig:pump_perturb}, in comparison with Figure~\ref{fig:eigval_p1}, confirms that in the presence of parameter fluctuations, increasing the pump rate $\pumprate$ continues to decrease the steady-state volumes.
\begin{figure}[h!]
    \centering
    \includegraphics[width=0.45\linewidth]{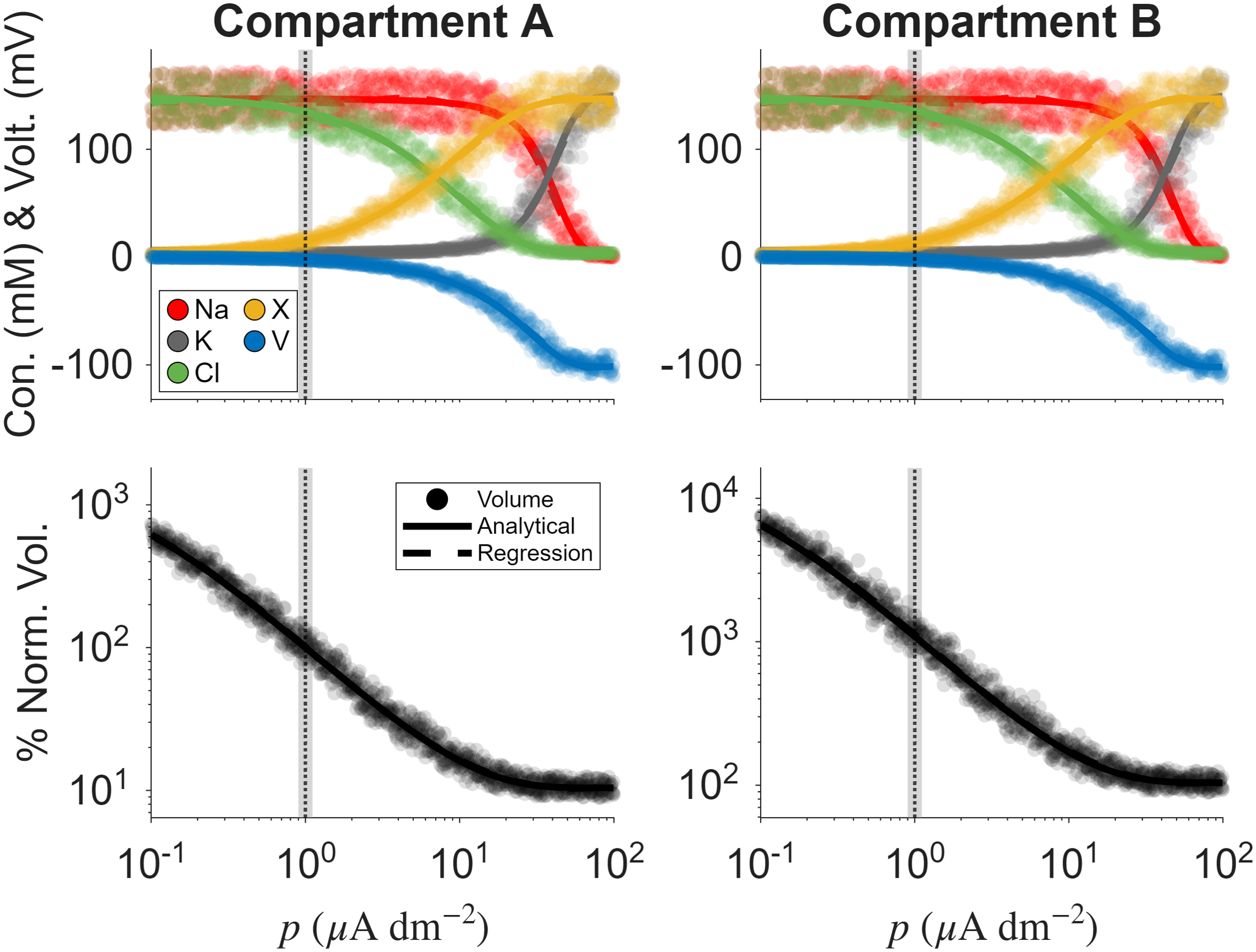}\qquad
    \includegraphics[width=0.22\linewidth]{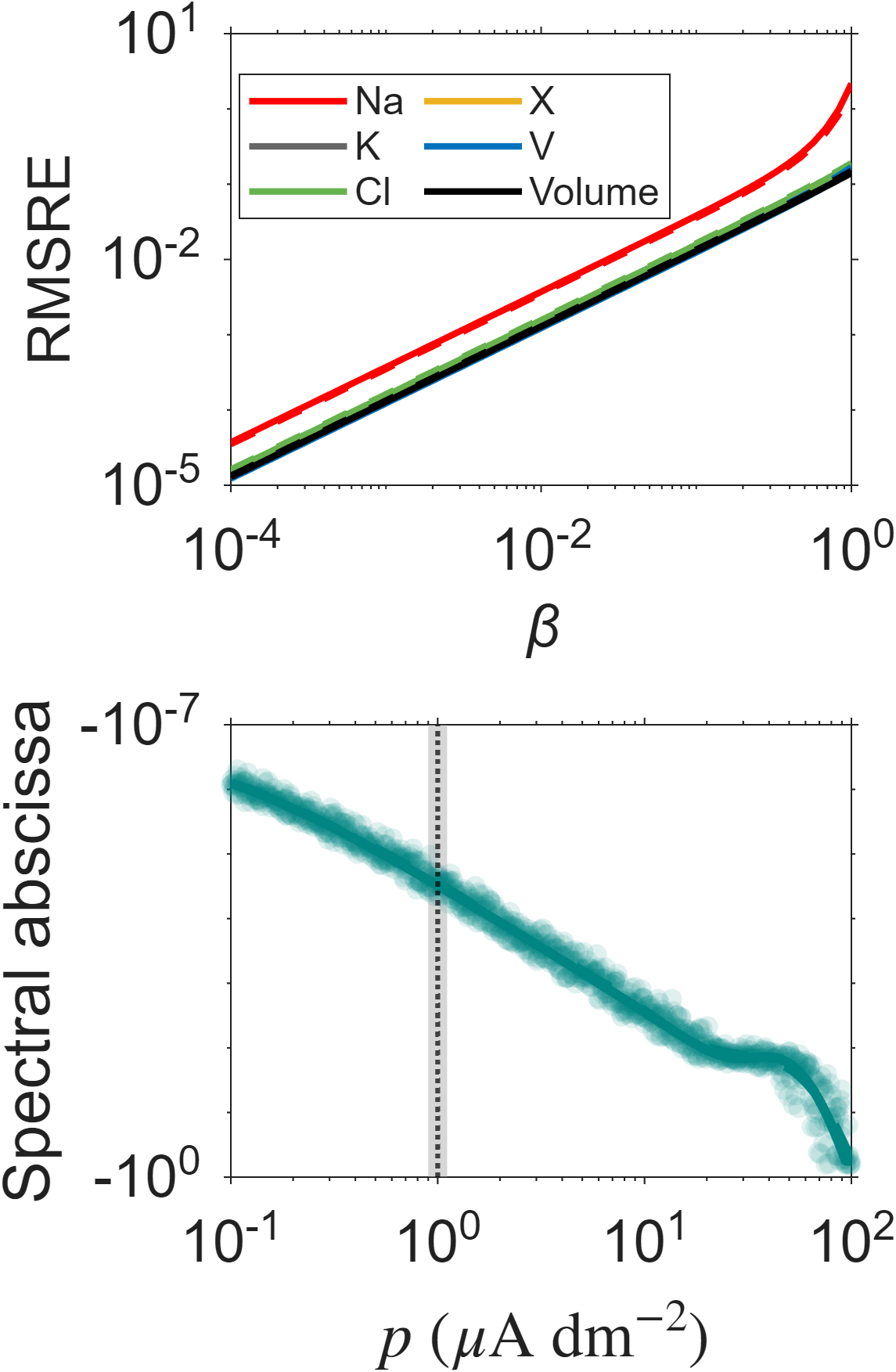}
   \caption{
\textbf{(Left \& Middle.)} Steady-state values of the ABp system are shown as functions of the pump rate $\pumprate \in (0,100)$: solid curves correspond to the default parameter values listed in Tables~\ref{tab:extracellular} and~\ref{tab:conductances}; scattered points correspond to $1000$ parameter vectors of dimension~9 sampled via LHS; and dashed curves represent GP regression fits to the scattered data. 
\textbf{(Bottom Right.)} Spectral abscissa of the Jacobian evaluated at the steady states shown in the left and middle panels. 
\textbf{(Top Right.)} Root mean square relative error (RMSRE), defined in Equation~\eqref{eq:RRMSE}, for the state variables in compartments $A$ (solid) and $B$ (dashed), plotted as a function of the scaling factor $\beta$, which increases the magnitude of the perturbations in $\ggionany$, $T$, and $\nacle$.
}
\label{fig:pump_perturb}
\end{figure}

Furthermore, Figures~\ref{fig:robustness_gna1} and \ref{fig:robustness_gna2_gnap}, when compared with the black arrows in Figures~\ref{fig:heatmapA_p_v_gna1}, \ref{fig:heatmapA_p_v_gna2}, and \ref{fig:heatmapA_p_v_gnaP}, confirm that under relatively weak NKA pump activity and in the presence of appropriate parameter fluctuations, decreasing the sodium conductances on the basolateral membrane and the paracellular pathway ($\ggnabl$ and $\ggnapara$) continues to decrease the steady-state volumes. Similarly, decreasing the sodium conductance on the apical surface ($\gnaap$) increases the steady-state volume in compartment $B$ while has little decreasing effect in compartment $A$.

We perform the same robustness analysis for intermediate and high pump rates ($\pumprate = 40, 90$ $\mu$A dm$^{-2}$)\ZA{, where the ion concentrations reach more physiologically realistic values—i.e., high potassium, low sodium and chloride, and low membrane voltage,} and observe similar behavior. The results are provided in Appendix~\ref{sec:robustness_mid_hi_pump}.

\begin{figure}[h!]
    \centering
    \includegraphics[width=0.45\linewidth]{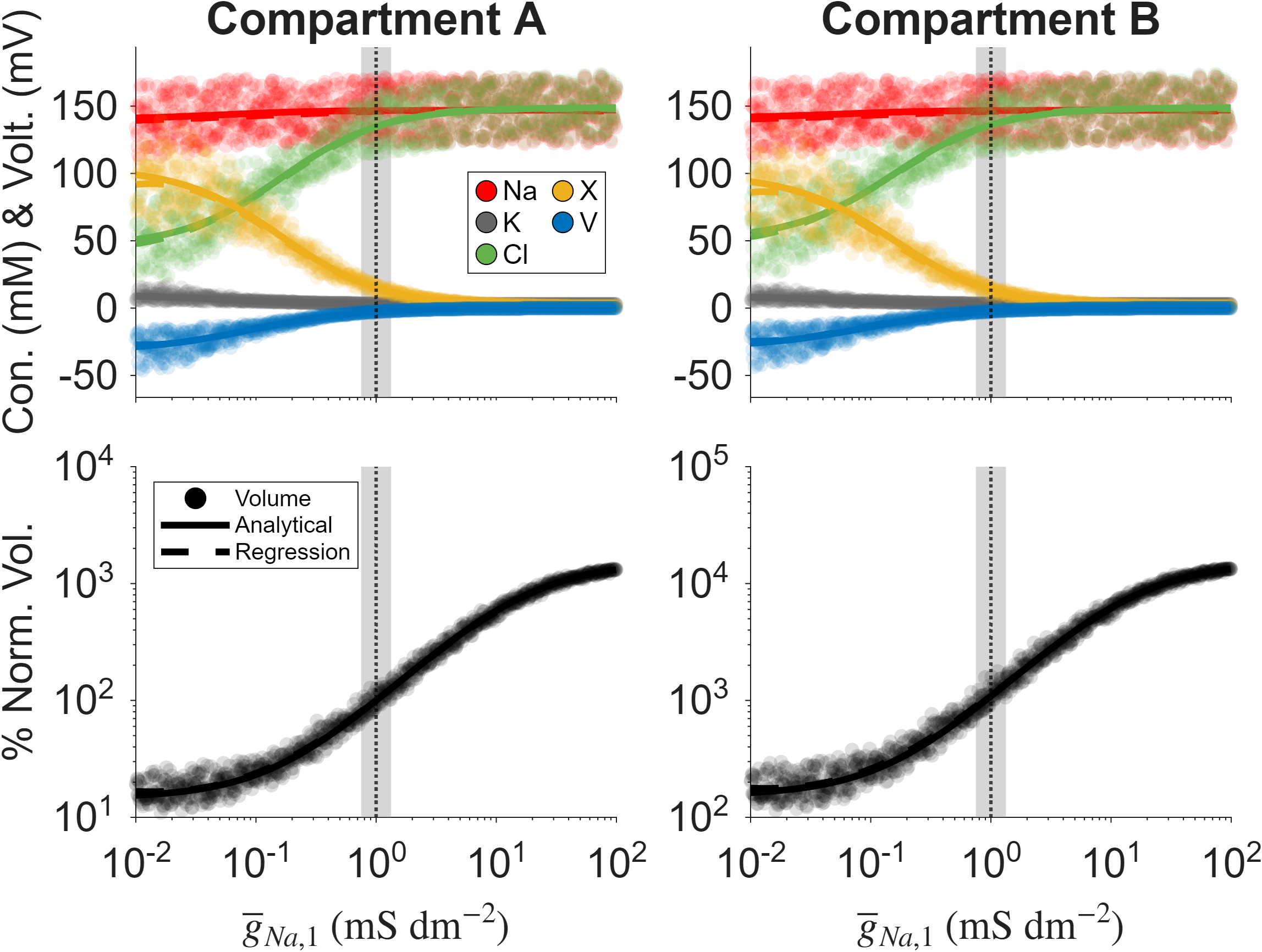}\qquad
    \includegraphics[width=0.22\linewidth]{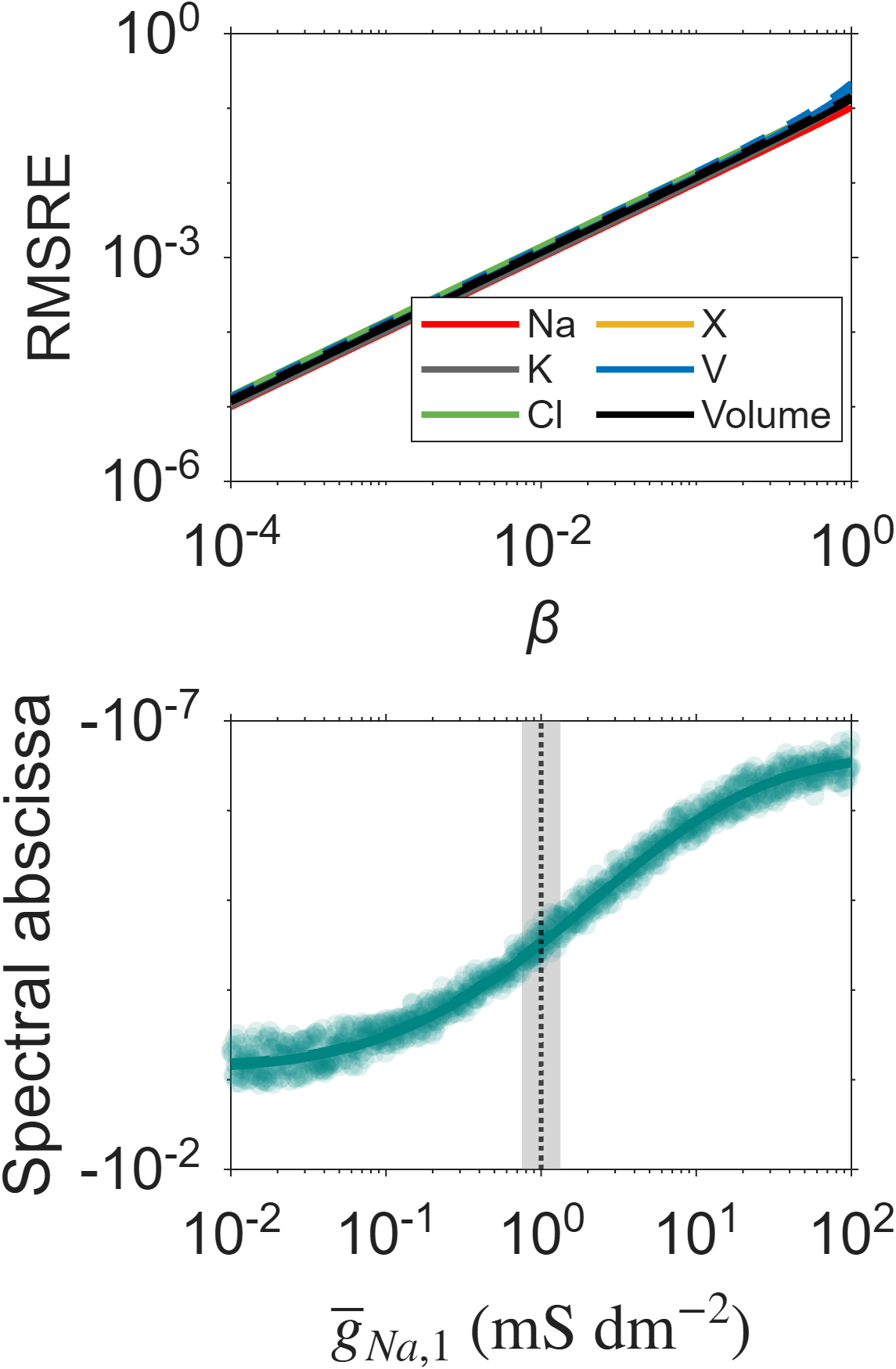}
    \caption{
       Steady-state values of the ABp system and corresponding spectral abscissa are shown as functions of the sodium conductance $\ggnabl$. The gray shaded vertical bands mark the perturbation described in \eqref{eq:parameter_interval} around the default value (dashed vertical line) of $\ggnabl$. 
        The RMSRE between scattered points and analytically computed steady states show that the error increases as the perturbation factor $\beta$ increases. 
    }
    \label{fig:robustness_gna1}
\end{figure}

\begin{figure}[h!]
    \centering
    \includegraphics[width=0.9\linewidth]{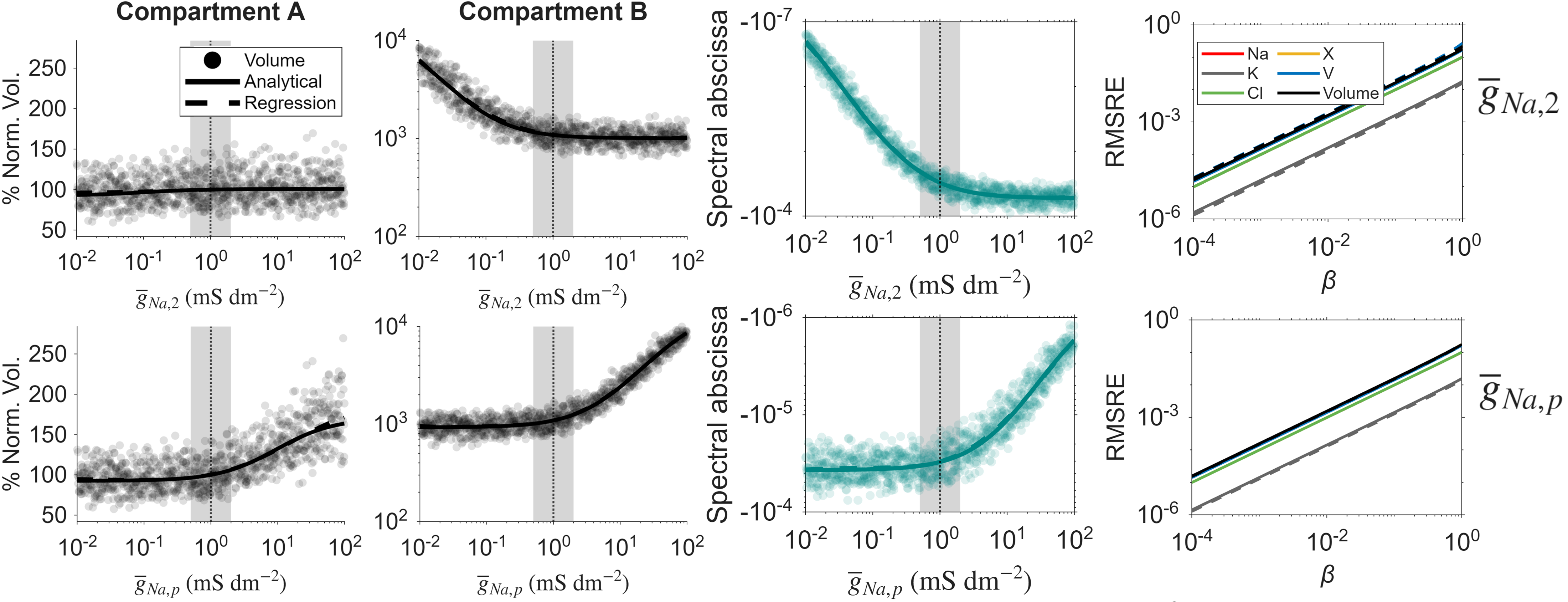}
    \caption{
    A comparison of the steady-state  values of the volumes and corresponding spectral abscissa are shown as functions of $\ggnaap$ (top) or $\ggnapara$ (bottom). Varying these conductances has little effect on concentrations and voltage (not shown).}
    \label{fig:robustness_gna2_gnap}
\end{figure}

Finally, Figure~\ref{fig:robustness_Nae} examines the robustness of the effects of extracellular concentrations on ABp system behavior, as previously discussed in Figure~\ref{fig:heatmap_p_vs_Na_e}. Under relatively weak NKA pump ($\pumprate\in[0.9,\,1.\overline{1}]$) and in the presence of parameter fluctuations, increasing the extracellular sodium and chloride concentrations (i.e., increasing $\nacle$) continues to decrease the steady-state volume, consistent with the black arrow shown in the right panel of Figure~\ref{fig:heatmap_p_vs_Na_e}.
 
\begin{figure}[h!]
    \centering
    \includegraphics[width=0.45\linewidth]{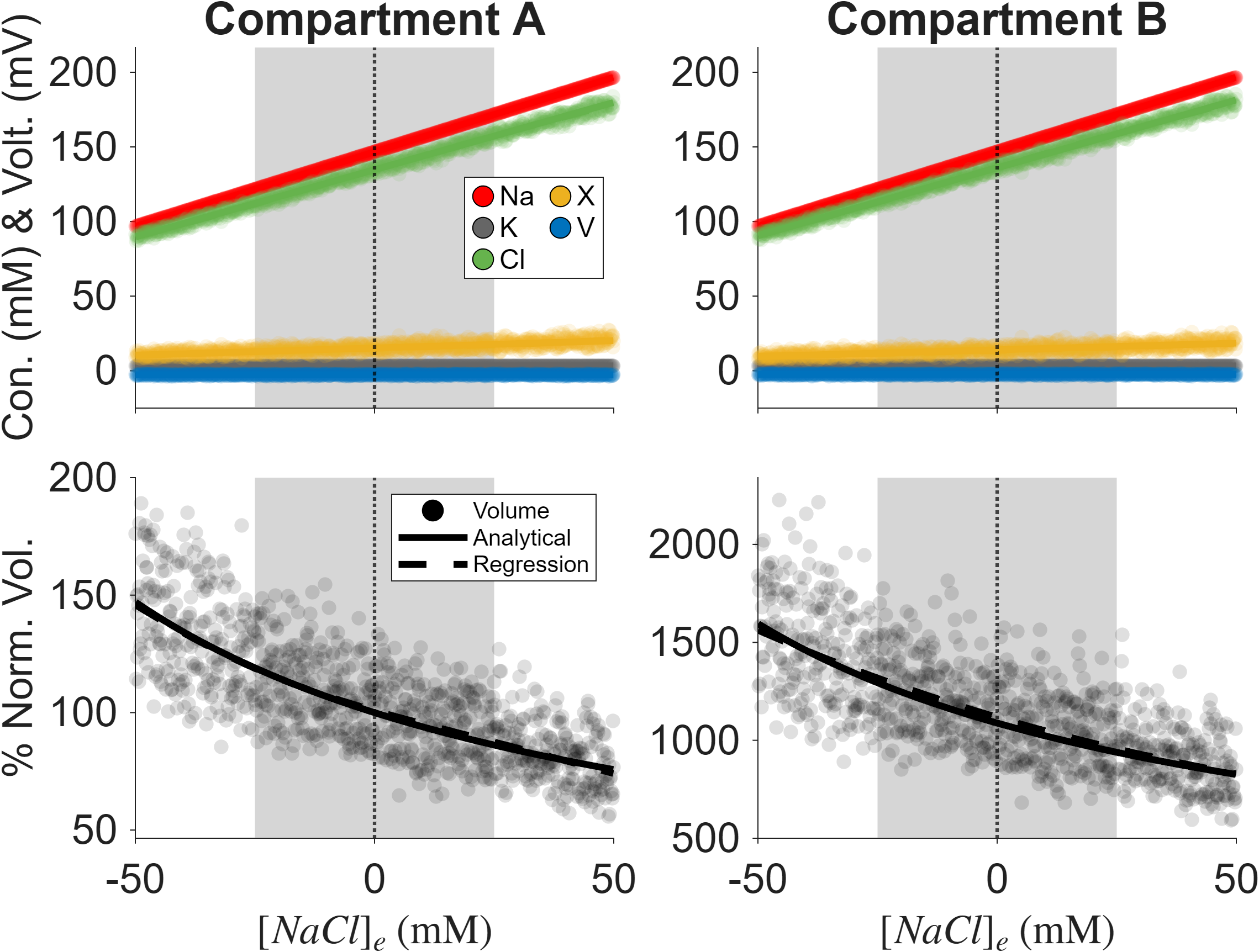}\qquad
    \includegraphics[width=0.22\linewidth]{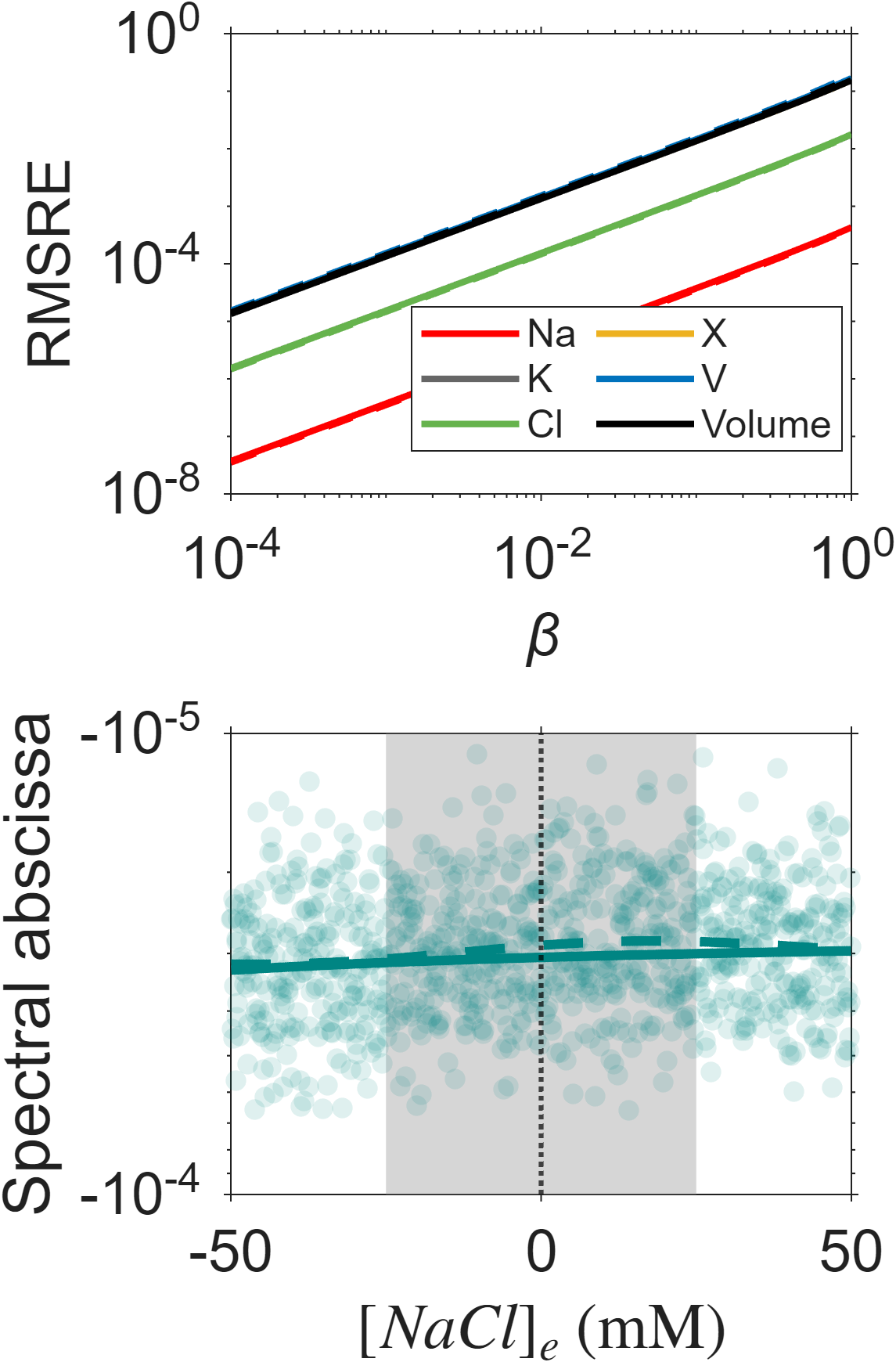}
    \caption{
    Steady-state values of the ABp system and corresponding spectral abscissa are shown as functions of $\nacle$ and the corresponding RMSRE are shown as functions of the perturbation factor $\beta$.
    }
    \label{fig:robustness_Nae}
\end{figure}

\subsection{Parameters that preserve homeostasis}
\label{subsection:homeostasis} 

In the sensitivity analysis, we observed that the potassium conductances 
$\ggkbl$, $\ggkap$, and $\ggkpara$ contribute negligibly to the variance of the 
steady-state values; see the short bars associated with these parameters in 
Figure~\ref{fig:sobol_wA__0}. This insensitivity is further illustrated in 
Figure~\ref{fig:robustness_gk1}, where all three potassium conductances are 
perturbed simultaneously and the resulting steady states are projected onto each 
conductance (left to right panels). Across the tested parameter ranges, the 
scatter points remain tightly clustered around the analytical steady-state 
curve, indicating that the steady state is effectively unchanged under combined 
fluctuations in $\ggkany$. Thus, in contrast to sodium conductances, the ABp 
system exhibits {an insensitivity to variations of the potassium conductance.}  

\begin{figure}[h!]
    \centering
    \includegraphics[width=0.33\linewidth]{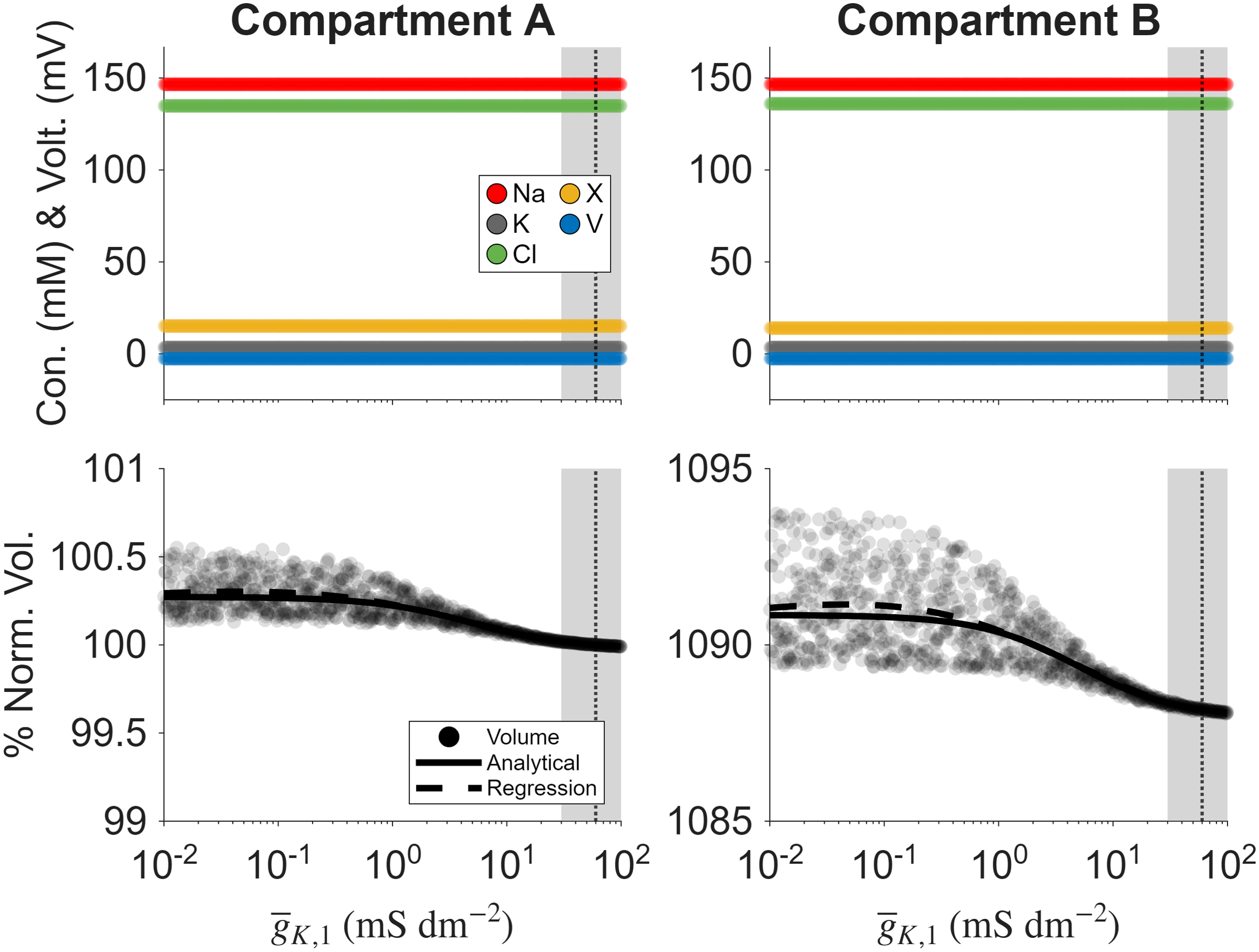}
    \includegraphics[width=0.33\linewidth]{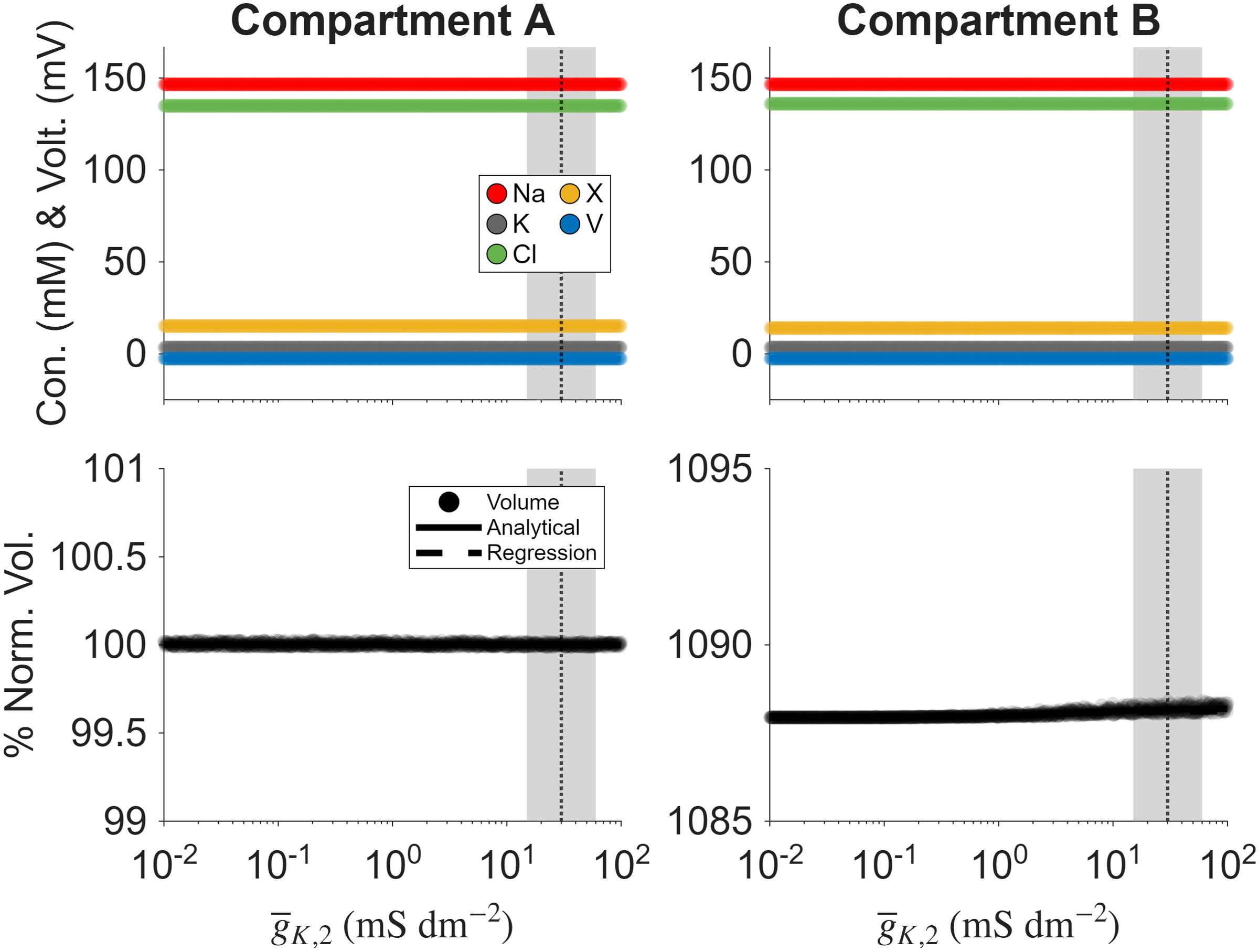}\includegraphics[width=0.33\linewidth]{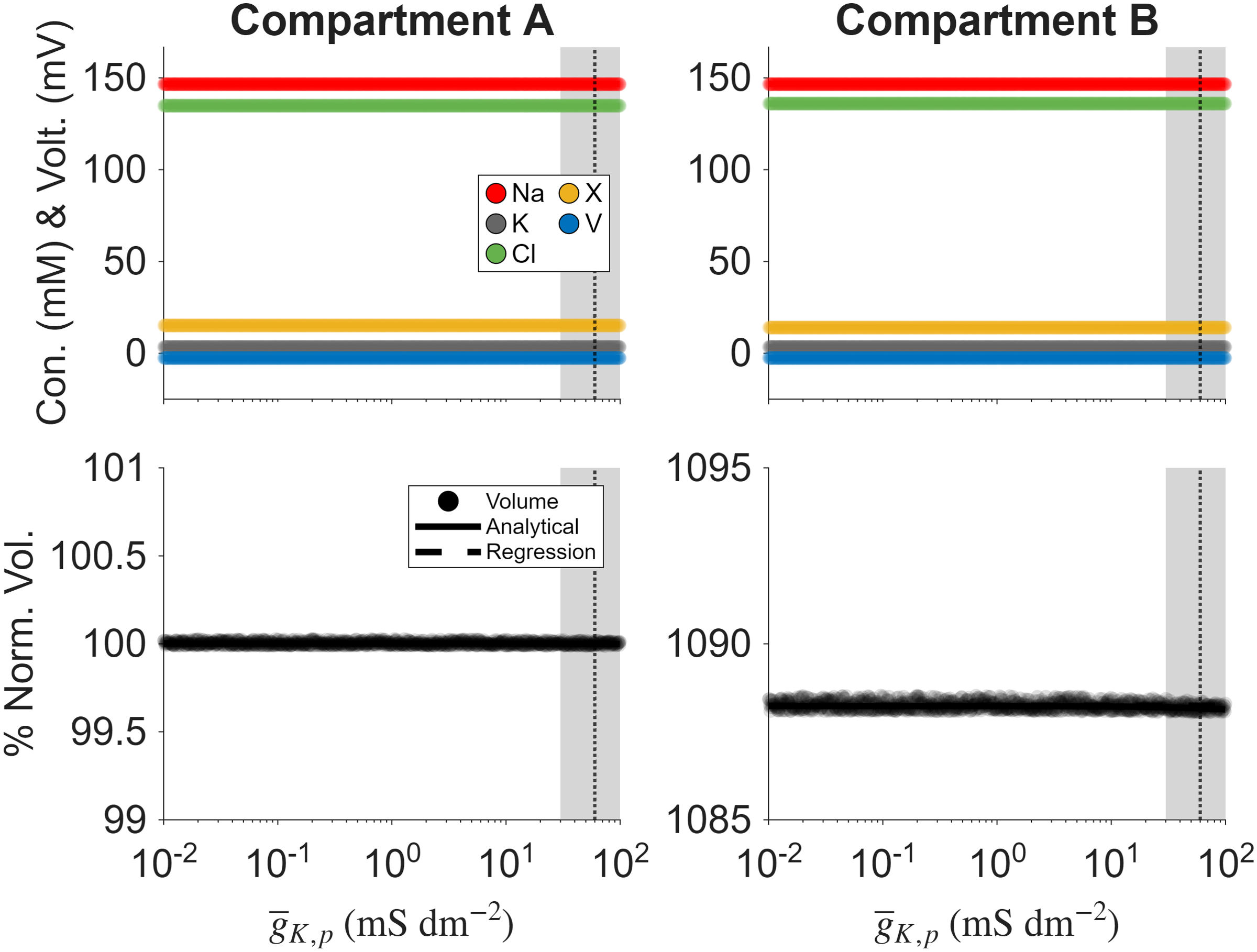}
    \caption{
\textbf{Homeostasis of steady states with respect to potassium conductances.} 
Steady-state responses of the ABp system under simultaneous perturbations of the potassium conductances 
$\ggkbl$, $\ggkap$, and $\ggkpara$. 
In each panel, one conductance is varied across its full tested range while the remaining two are perturbed within their sampling intervals. 
}
    \label{fig:robustness_gk1}
\end{figure}

Chloride conductances do not appear in the closed-form steady state expressions derived in Section \ref{sec:active}, so varying $\ggclbl$, $\ggclap$, or $\ggclpara$ does not alter the steady state ion concentrations, voltages, or volumes. In this sense, for fixed values of the remaining parameters, the ABp steady state is invariant under sufficiently small fluctuations in $\ggclany$, and the system is homeostatic with respect to these conductances. The chloride conductances do, however, enter the Jacobian \eqref{eq:Jacobian} and therefore influence the eigenvalues and transient dynamics. Assumption \ref{assumption:equilibrium} ensures that no more than one of $\ggclbl$, $\ggclap$, and $\ggclpara$ is fixed at zero. In addition, within the admissible parameter set, perturbations in the chloride conductances change the magnitude of the spectral abscissa but not its sign, so local stability is preserved. For this reason, we hold $\ggclany$ fixed at their default values in our sensitivity and robustness analyses.

%~~~~~~~~~~~~~~~~~~~~~~~~~~~~~~~~~~~~~~~~~~~~~~~~~~~~~~~~~~~~~~~~~~~
\section{Active Two-Compartment Systems:  NKA Pump on  Apical Surface} \label{subsection:pumpAS} 

In a less common configuration, the NKA pump mechanism can be found on the apical surface, interfacing compartments $A$ and $B$, without a corresponding pump mechanism on the basolateral surface. {This is true in the case of the epithelial cells of the choroid plexus, which secrete cerebrospinal fluid into the ventricles of the brain \cite{brown2004molecular}.}
%For example, in the lining of epithelial cells and ventricles of the choroid plexus, sodium ions are pumped into the ventricle as potassium ions are pumped into the cell to establish a concentration gradient \AK{replace the previous sentence with 'This is true in the case of the epithelial cells of the choroid plexus, which secrete cerebrospinal fluid into the ventricles of the brain.'  }
We extend the PLEs to cases where the passive and active transporters generate ionic gradients along the apical surface, while the basolateral membrane solely possesses ionic channels. In Equations \eqref{eq:ode_water}--\eqref{eq:ode_con_b}, we let 
\[
\pnabl = 0,\quad \pkbl = 0, \quad \pnaap=-\gamma_{Na} \, \pumprate \, \Arap , \quad \pkap = \gamma_{K} \, \pumprate \, \Arap.
\]
The NKA pump hydrolyzes ATP in compartment $A$, expending energy to pump $\gamma_{Na}$ Na\textsuperscript{+} ions into compartment $B$ for every $\gamma_K$ K\textsuperscript{+} pumped from $B$ into $A$. The ionic channels and paracellular pathway facilitate the passive transport of ions and water between regions. From our steady state equations in Section \ref{sec:active} and numerical solutions of the ODE system \eqref{eq:ode_water}--\eqref{eq:ode_con_b},
we compare the effects of pump rates on ionic flux at steady state and in a time series.
\medskip 

\noindent\textbf{Existence of steady states. }
The form of the steady states for the PLEs with an apical surface pump mechanism is very similar to that of the PLEs with a basolateral membrane pump mechanism derived in Section~\ref{sec:existence}, with the following modifications. In Lemma~\ref{lemma:existence_ss} and Proposition~\ref{proposition:steadystate}, replace $\pumprate \, \Arbl$ by $\pumprate \, \arap$, and replace $G_{\text{ion},A}$ and $G_{\text{ion},B}$ by  
\begin{subequations}
\begin{align}
    &\GionAap  = \dfrac{g_{\text{ion},p}}{g_{\text{ion},1} g_{\text{ion},2}  + g_{\text{ion},1} g_{\text{ion},p}  + g_{\text{ion},2} g_{\text{ion},p}} \label{eq:Giona2p2} \\
    &\GionBap  = \dfrac{-g_{\text{ion},1}}{g_{\text{ion},1} g_{\text{ion},2}  + g_{\text{ion},1} g_{\text{ion},p}  + g_{\text{ion},2} g_{\text{ion},p}} . \label{eq:Gionb2p2}
\end{align}
\label{eq:Gionp2}
\end{subequations}

\noindent\textbf{The possible range of the pump rate. }
In this case, unlike the ABp system with the NKA pump located on the basolateral membrane (see Figure~\ref{fig:eigval_p1}), the ABp system becomes highly sensitive to small changes in the pump rate $\pumprate$. In particular, the steady state of compartment~$B$ exists only over a restricted range of pump values, namely for $\pumprate < \PmaxBap$, where $\PmaxBap$ is substantially smaller than $\PmaxAap$ for compartment~$A$. See Equation~\eqref{eq:fj1} and Figure~\ref{fig:fj2fig} for a comparison of $\PmaxBap$ and $\PmaxAap$. Consequently, the existence of a stable steady state with non-negative concentrations and volume is confined to a relatively narrow range of $\pumprate$.

\begin{figure}[h!]
    \centering
    \includegraphics[width=0.6\linewidth]{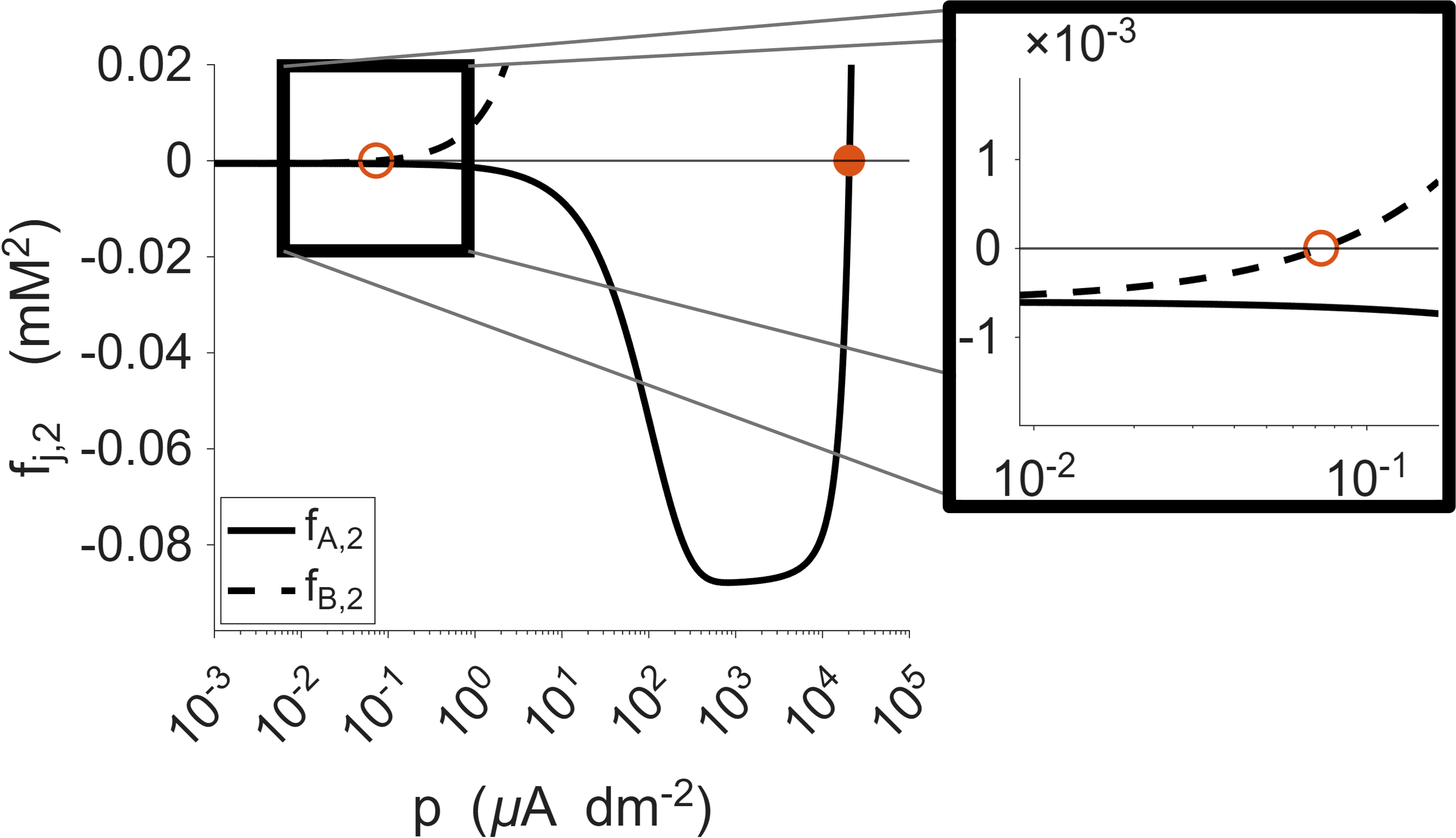} 
    \caption{
    The admissible range of $\pumprate$ is computed numerically using the default parameter values by plotting $f_{j,2}(\pumprate) := 4 \, \Cpbl(\pumprate) - \Ose^2$ and identifying the interval $(0, \pMax)$ where $f_{j,2}(\pumprate) < 0$. In the zoomed-in panel on the right, the open circle denotes the root of $f_{B,2}$.
    }
    \label{fig:fj2fig}
\end{figure}

\medskip 

\noindent\textbf{Explicit form of steady states and their stability. } 
Figure~\ref{fig:papab} shows the steady-state values of an ABp system as functions of $\pumprate$, where the pump is located on the apical surface and $\pumprate \leq \PmaxBap$. All steady-state concentrations and volumes remain non-negative across this range since the condition $f_{j,2}(\pumprate) := 4 \, \Cpbl(\pumprate) - \Ose^2<0$ remains valid.

\begin{figure}[h!]
    \centering
     \includegraphics[width=.8\linewidth]{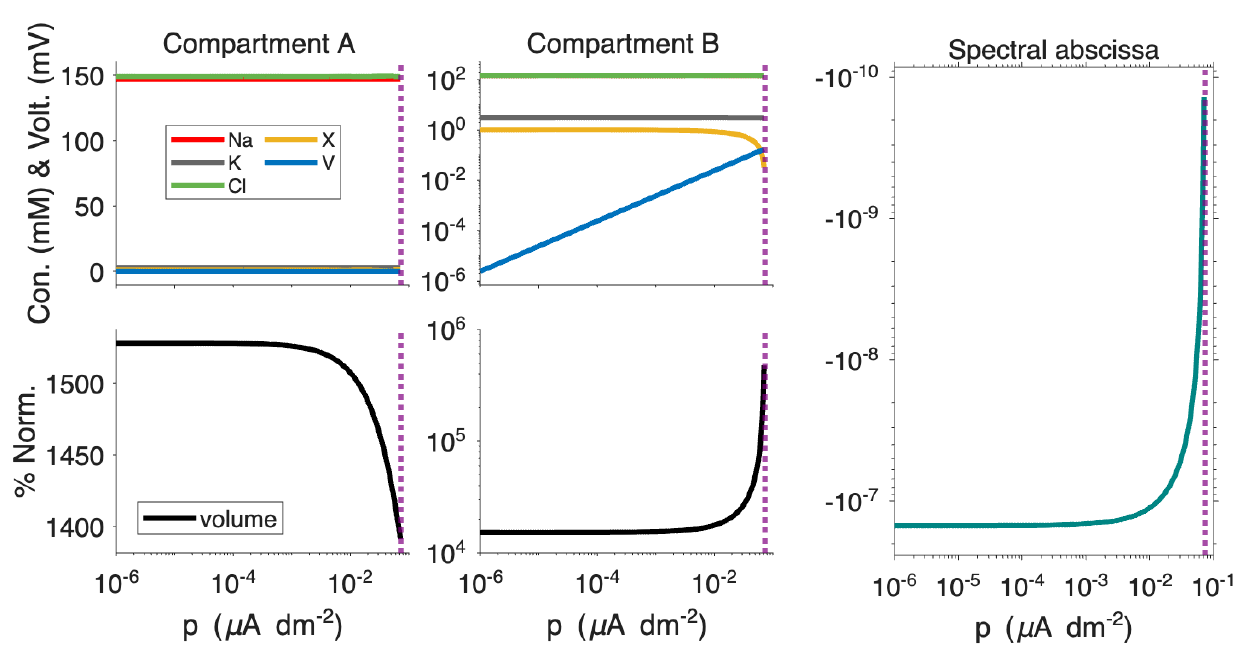}
    \caption{
\textbf{
ABp system with the NKA pump on the apical surface.}
Steady states and their corresponding spectral abscissae are plotted as functions of the pump rate~$\pumprate$. The purple vertical line marks $\PmaxBap$, the largest value of $\pumprate$ for which a steady state with non-negative volume exists. {$\naB^{ss}$ and $\clB^{ss}$ remain approximately the same over this range of $\pumprate$ and are plotted on top of one another in the middle panel.}
}
    \label{fig:papab}
\end{figure}

Although steady states of the ABp system with non-negative concentrations and volumes for both compartments $A$ and $B$ exist only for $\pumprate < \PmaxBap$, compartment $A$ remains stable for $\PmaxBap < \pumprate < \PmaxAap$, while compartment $B$ exhibits unbounded growth. Figure~\ref{fig:apical_timeseries_SS_pmaxA} (left panels) shows the time series of a numerically computed solution of the coupled PLEs with $\pumprate = 2 \times \PmaxBap$. As illustrated, all variables in compartments $A$ and $B$ approach finite steady-state values except for the volume of $B$, which diverges. This behavior persists for all $\pumprate < \PmaxAap$, as demonstrated in the right panel of Figure~\ref{fig:apical_timeseries_SS_pmaxA}. 
Note that $\wB^{ss}$ appears as a negative value for $\pumprate > \PmaxBap$ because the condition $f_{B,2} < 0$ is violated; in practice, the volume does not become negative but instead grows without bound.

\begin{figure}[h!]
    \centering
    \includegraphics[width=.95\linewidth]{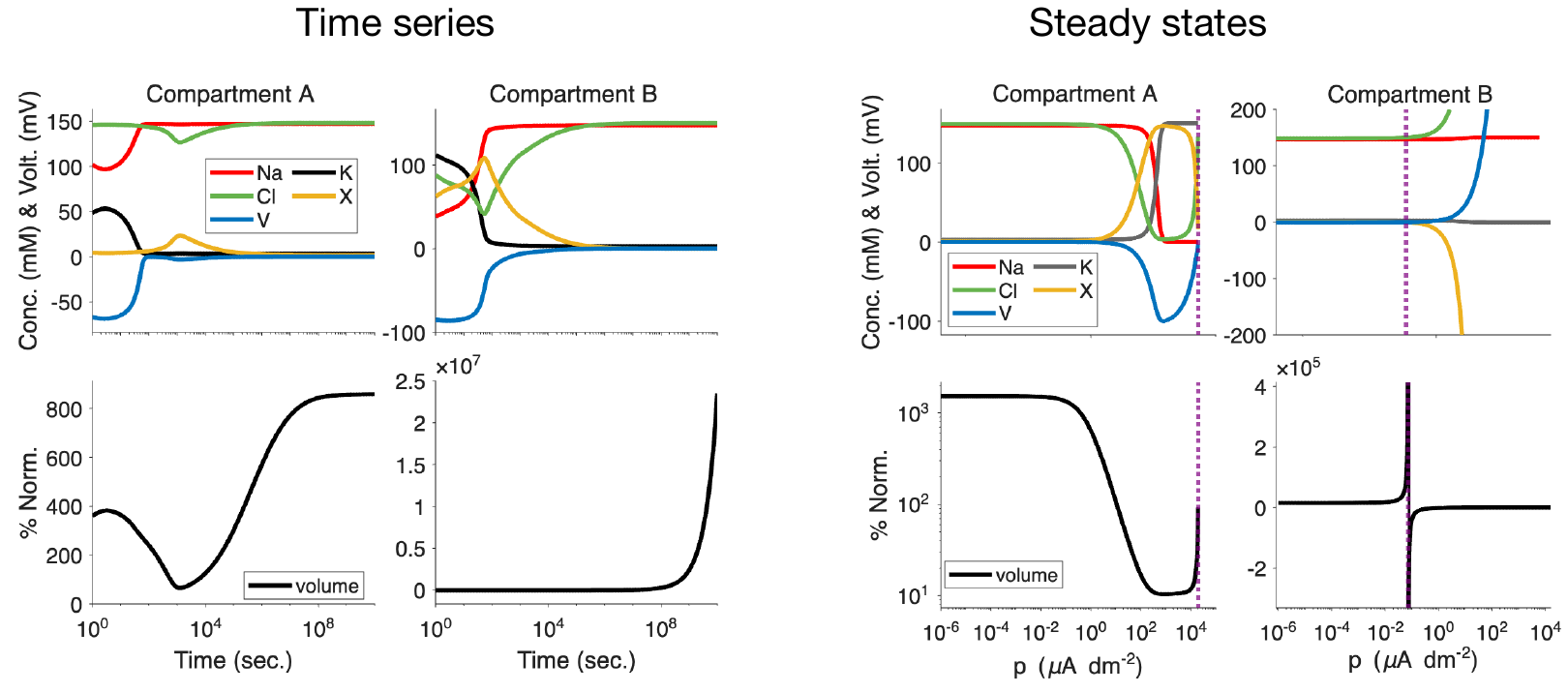}\quad
    \caption{(Left) Time series of a numerically computed solution of the coupled PLEs with $\pumprate = 2 \times \PmaxBap$. (Right) Steady state values are plotted as functions of the pump rate $\pumprate$ for a larger range of $\pumprate$, $0<\pumprate<\PmaxAap$ which violates the condition for existence of steady states with non-negative volumes. So on the right-bottom panel, $\wB^{ss}$ increases to $\infty$ for $\pumprate<\PmaxBap$ and becomes negative for $\pumprate>\PmaxBap$. {This corresponds to epithelial transport in the time series where $\wB(t)$ grows without bound.}
        }
    \label{fig:apical_timeseries_SS_pmaxA}
\end{figure}

\medskip

\noindent\textbf{A note on ion flows in an ABp system at steady state.}
For an ABp system with the NKA pump located on the apical surface, the ionic fluxes across the three interfaces are
\begin{subequations}
\begin{align}
    &\fluxionbl = \zion \, \gionbl \left( \vA - \eionA \right), \\
    &\fluxionap = - \zion \, \gionap \left[ \left( \vA - \eionA \right) - \left( \vB - \eionB \right) \right] + \pionap, \\
    &\fluxionpara = - \zion \, \gionpara \left( \vB - \eionB \right),
\end{align}
\label{eq:fluxes}
\end{subequations}
where 
$ \pnaap=-\gamma_{Na} \, \pumprate \, \Arap$, $\pkap = \gamma_{K} \, \pumprate \, \Arap,$ and $\pclap=0.$

Using the steady-state expressions in Equation~\eqref{eq:jss_p1}, together with $\GionAap$ and $\GionBap$ from Equation~\eqref{eq:Gionp2}, one verifies that the absolute values of the total ion---and water---fluxes are equal across the basolateral, apical, and paracellular pathways for $\pumprate < \PmaxBap$; that is, 
\[
\fluxionbl^{ss} = \fluxionap^{ss} = \fluxionpara^{ss}.
\]

As illustrated in the two left panels of Figure~\ref{fig:Apical_Na_flux}, sodium flow forms a clockwise loop in this configuration: sodium moves from compartment~$B$ to the ISF through the paracellular pathway, enters compartment~$A$ across the basolateral membrane, and returns to compartment~$B$ via the apical surface. The total fluxes across the three membranes are plotted in the second-left panel as functions of $\pumprate < \PmaxBap$. As $\pumprate$ increases, the magnitude of the flux increases across all pathways while remaining equal, confirming conservation of flow at steady state. Unlike in Figure~\ref{fig:total_flow}, the fluxes here are negative, indicating opposite flow directions. 

In addition, in Figure~\ref{fig:Apical_Na_flux} (right panels), we plot the time series of $\fluxnabl, \fluxnaap, \fluxnapara$ for some $\PmaxBap<\pumprate<\PmaxAap$ and observe that 
for large time $t$, i.e., at the steady state, the $Na^+$ fluxes become nearly constant and \[\fluxnabl(t)\approx\fluxnaap(t)<\fluxnapara(t)<0.\] Thus, the magnitude of the paracellular pathway $Na^+$ flux $B\rightarrow\text{ISF}$ is smaller than the magnitude of the apical surface $Na^+$ flux $A\rightarrow B$. This asymmetry produces a net $Na^+$ gain in compartment $B$, so $\naB(t)$ increases. The resulting osmotic imbalance drives water into compartment $B$ to compensate for the growing osmolarity, and $\wB(t)$ increases. The $K^+$ flux exhibits an analogous bottleneck in compartment $B$.

\begin{figure}[h!]
    \centering
     \includegraphics[width=.95\linewidth]{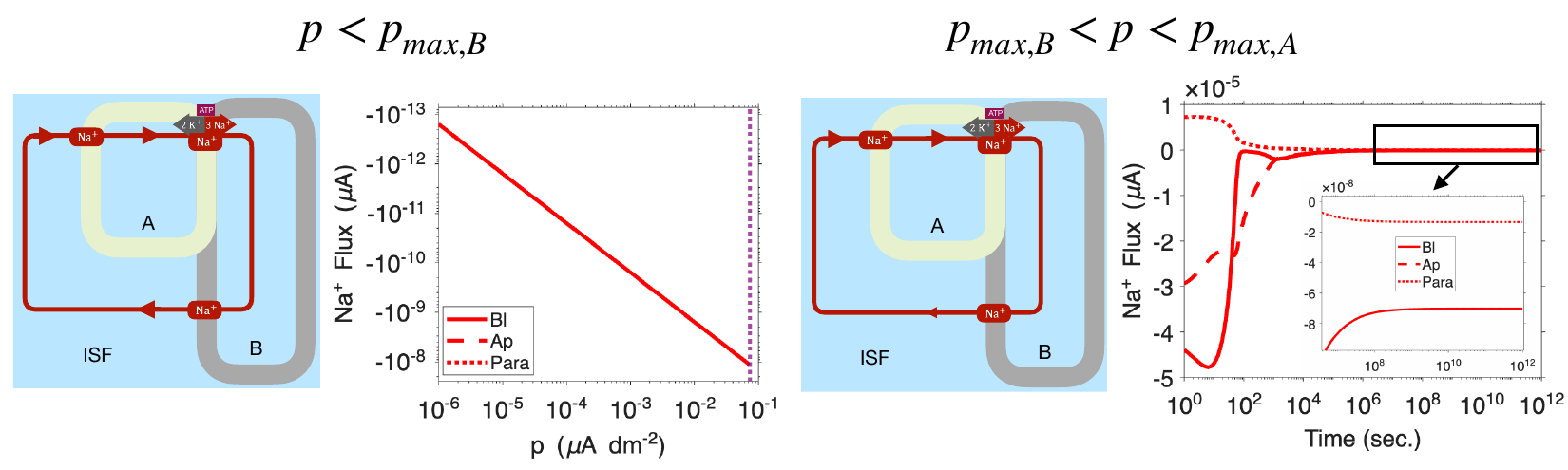}
\caption{
    \textbf{Sodium transport across the membranes at steady state for an ABp system with the NKA pump on the apical surface.}
    The schematic diagrams illustrate the clockwise sodium transport loop. For $\pumprate < \PmaxBap$, the total sodium fluxes across the three membranes are equal at steady state (second panel from left). For $\pumprate > \PmaxBap$, the flux through the paracellular pathway becomes smaller than the apical and basolateral fluxes, leading to water accumulation in compartment~$B$ and resulting in its instability (right panel).
}
    \label{fig:Apical_Na_flux}
\end{figure}

\medskip 

Apical NKA transport requires that the $K^+$ entering compartment A across the apical surface not be immediately lost across the basolateral membrane. If the basolateral $K^+$ conductance is too high, it dissipates the $K^+$ gradient generated by the pump and allows $K^+$ to leak out across the basolateral membrane. Keeping basolateral $\ggkbl$ smaller maintains a high intracellular $K^+$ activity set by pump activity and allows $K^+$ recycling rather than $K^+$ loss \cite{zeuthen1981epithelial, macaulay2022cerebrospinal}. This is the same energy-conservation logic that underlies classic pump-leak models, but here it is applied to the less common apical NKA configuration. Consistent with this physiology, the mathematical ABp system admits a stable steady state when the apical $K^+$ conductance is large ($\ggkap$ high) and the basolateral $K^+$ conductance is small ($\ggkbl$ low). A large $\ggkap$ provides the local $K^+$ recycling pathway needed to balance apical NKA current, while a small $\ggkbl$ prevents the basolateral side from short-circuiting the apical pump and draining intracellular $K^+$. Unlike the ABp system with basolateral NKA, where a sizable basolateral K$^+$ leak helps regulate the system, in the apical NKA case, the same leak becomes destabilizing because it competes directly with the pump for control of intracellular $K^+$ and volume. Not shown here, $\ggkap=300$ mS dm$^{-2}$ and $\ggkbl=30$ mS dm$^{-2}$ are parameter values consistent with biological systems, in which$K^+$ conductance on the apical surface can be roughly 10 times larger than on the serosal surface \cite{zeuthen1981epithelial}.

%~~~~~~~~~~~~~~~~~~~~~~~~~~~~~~~~~~~~~~~~~~~~~~~~~~~~~~~~~~~~~~~~~~~
\section{Discussion} \label{sec:discussion}

\ZA{
This work provides a systematic extension of the classical pump--leak framework to a coupled, two-compartment system, referred to as the ABp system, demonstrating how spatial organization fundamentally alters the qualitative behavior of ionic concentrations, membrane potential, and volume.}
\AK{Our development of the ABp model extends the reach of the PLM to multicellular assemblies and serves as a prelude to the mathematical analysis of epithelial systems that transport fluid.
}

\ZA{
Despite the resulting 10-dimensional system of coupled algebraic–differential equations and the presence of many parameters, we derive explicit steady-state formulas for both passive and pump-driven active regimes under a constant pump mechanism and characterize their dependence on physiologically meaningful parameters. We further show that assuming a constant pump rate preserves the qualitative behavior of the system while rendering the analysis analytically accessible. Obtaining explicit expressions for steady states is especially important in this setting, as the system exhibits stiffness and strong parametric dependence, making purely numerical determination of steady states both challenging and unreliable across large regions of parameter space.  We further establish local stability of these steady states using the Hartman–Grobman Theorem. \AK{Previously Mori \cite{mori2012mathematical} demonstrated how a free energy formulation of the PLM using Lyapunov functions could be used to study the \ZA{global} stability of the PLEs \ZA{for a single cell}. We leave the global asymptotic analysis of the ABp system for future studies. }

Although the steady states are derived explicitly, their expressions depend nonlinearly on numerous model parameters, making qualitative prediction of ABp system behavior challenging as parameters vary. To address this complexity, we perform a global sensitivity analysis that reveals a pronounced low-dimensional structure underlying the high-dimensional parameter space. Despite the large number of parameters, variance-based Sobol indices show that steady-state volume, voltage, and ionic concentrations are governed primarily by a small subset of parameters, notably pump strength, sodium conductances, and extracellular sodium concentration. This clear separation between dominant and negligible parameter directions suggests opportunities for systematic model reduction and provides a rigorous explanation for the robustness of volume regulation observed across wide parameter ranges.
}

\AK{The model developed here is for a rather specialized system, a stable epithelial vesicle. This model encompasses a wide range of natural biological structures as well as organoids. We initially developed the ABp to model a rather specialized biological structure, the scolopidium (or chordotonal organ), which serve as stretch receptors in insects \cite{Field_Matheson_1998}. % add to bib
 Scolopidia have a specific glial cell, a scolopale cell that surrounds the dendrites of a sensory neuron. The scoloplale cell creates a closed fluid filled compartment around the sensory dendrites  with an elevated $K^+$ concentration. Another example of an ABp system is the scala media in the mammalian cochlea. In this case a multilayer epithelium, the stria vascularis produces endolymph which covers the apical surface of the organ of Corti, and has a high $K^+$ concentration \cite{Nin_2012}. 

 It is worth noting that biomechanical models have been proposed  to simulate the development of epithelial vesicles  \cite{Rejniak_Anderson_2008} 
and there is considerable interest in the formation of epithelial organoids in vitro  \cite{Lu_2025, Pedersen_1999}. 
 
In the ABp model it is possible for cyclic flows of $Na^+$ and $K^+$  to be generated by the operation of the NKA. Without an active $Cl^-$ transporter there are no net flows of this ion. In the model as implemented no cyclic water flows can occur. However, when ions move through the channels they have an obligatory hydration shell, the water can cycle too. In the case of voltage gated $Na^+$ channels water, water accompanies the ion in its passage through the channel, whereas in the case of $K^+$ channels water is stripped from the ions \cite{roux2017ion}. If this is true for $Na^+$ leak channels, water will cycle through the ABp system. This is what is classically called electroosmosis \cite{finkelstein1987water}.  \ZA{See Figure~\ref{fig:total_flow}.}}

\ZA{
From a modeling perspective, although the ABp model introduced here provides a tractable extension of the classical PLEs and captures the essential dynamics of a cell coupled to a lumen—both reproducing known regulatory mechanisms and predicting new ones—it nevertheless has limitations that warrant further investigation. In particular, we model ionic currents using linear Ohm's law, whereas a more biophysically detailed description could be obtained using the nonlinear Goldman--Hodgkin--Katz (GHK) formulation. Similarly, water transport is modeled using a Starling-type relation without explicit incorporation of hydrostatic \AK{effects, or ion-water interactions.} 

In addition, it remains to be shown rigorously under what conditions an ABp system provides an accurate reduced description of a multicompartment epithelial architecture of the form $A_1\dots A_NBp$. Extending the model to include ion cotransporters represents another important direction toward greater physiological realism.

From a mathematical perspective, while our analysis focuses on the existence, local stability, and robustness of steady states, the explicit formulas derived here lay the groundwork for deeper analytical investigations. The strong dependence of steady states on model parameters suggests the presence of rich bifurcation structures. A systematic bifurcation analysis, as well as the development of Lyapunov or free-energy methods adapted to two-compartment settings, may provide insight into global stability, transitions between stable regimes, and the onset of pathological behaviors such as luminal volume divergence. 

Incorporating GHK-type nonlinearities will likely preclude closed-form expressions for steady states; consequently, the development of accurate and efficient numerical methods for computing steady states and their stability properties will be essential for advancing these models.
}

\section*{Materials and methods}
All computations were performed using MATLAB, release 25.2.0.2998904 (R2025b) (The MathWorks, Inc., Natick, MA). A small capacitance allows rapid changes in voltage. In this case, the PLEs become ``stiff,'' and special numerical solvers are required to solve this system. In this work, for the case of the non-linear NKA models, the system of equations was solved using the MATLAB stiff differential equation solver ode15s or ode23tb (the results are similar, only their
speeds are different for different parameters). 
To reduce the sensitivity of time series solutions to stiff solvers, we use the voltage formulation algebraically derived in Appendix \ref{sec:voltage}, which yields more numerically stable trajectories. \KT{The code is currently maintained in a private repository and is available upon reasonable request. The code will be made publicly available upon acceptance of the associated manuscript.}

\section*{Acknowledgment}
This work was supported by National Science Foundation grant 2037828 (KT, AK, ZA) and The Simons Foundation MPS-TSM-00008005 (ZA). Generative AI tools were used exclusively for copy-editing purposes--such as improving grammar, phrasing, and flow--and for error checking.

%~~~~~~~~~~~~~~~~~~~~~~~~~~~~~~~~~~~~~~~~~~~~~~~~~~~~~~~~~~~~~~~~~~~
%\printbibliography
%\bibliographystyle{plain}
%\bibliography{bibliography}

%~~~~~~~~~~~~~~~~~~~~~~~~~~~~~~~~~~~~~~~~~~~~~~~~~~~~~~~~~~~~~~~~~~~
\section{Appendix}

\subsection{Default values of the coupled PLEs model parameters} \label{sec:appendix_1}

In Tables~\ref{tab:constants}--\ref{tab:pump}, we introduce the parameters used in the coupled PLEs. We also provide \textit{default} values for these parameters as a reference set. However, we emphasize that a main objective of this work is to avoid relying on a single parameter set and instead explore the system's behavior over a wide range of parameter values.

\begin{table}[h!]
    \centering
    \begin{tabular}{| c || c | c | c |}
        \hline
         & Equations & Default value & Description \\
        \hline
        \hline
        $\Arbl$ && $2\pi 10^{-7}$ $\text{dm}^2$ & Basolateral membrane area \\
        % \hline
        $\Arap$ && $2\pi 10^{-7}$ $\text{dm}^2$ & Apical surface area  \\
        % \hline
        $\Arpara$ && $2\pi 10^{-8}$ $\text{dm}^2$ & Paracellular pathway area   \\
        % \hline
        $w_{A,0}$ && $\frac{4}{3} \pi \cdot 125\cdot10^{-15}$ dm$^{3}$ & Initial volume of compartment $A$  \\
        % \hline
        $w_{B,0}$ && $\frac{4}{3} \pi \cdot 125\cdot10^{-14}$ dm$^{3}$ & Initial volume of compartment $B$   \\
        % \hline
        $X_A$ &\eqref{eq:alg_volt}, \eqref{eq:osj}& $5\cdot10^{-3} \times w_{A,0}$ mol & Intra. impermeant molecules in $A$ \\
        % \hline
        $X_B$ &\eqref{eq:alg_volt}, \eqref{eq:osj}& $ 50\cdot 10^{-3} \times w_{B,0}$ mol & Intra. impermeant molecules in $B$  \\
        % \hline
        $\zA$ & \eqref{eq:alg_volt} & -1 & Charge of impermeant molecules in $A$ \\
        % \hline
        $\zB$ & \eqref{eq:alg_volt} & -1 & Charge of impermeant molecules in $B$ \\
        % \hline
        $C_{m,A}$ & \eqref{eq:alg_volt} & $10^{-4} \text{F}$ & Total capacitance \\
        % \hline
        $C_{m,B}$ & \eqref{eq:alg_volt} & $10^{-4} \text{F}$ & Total capacitance \\
        % \hline
        $R$ & \eqref{eq:E_ion_j} & $8.314 \frac{\text{J}}{\text{mol K}}$ & Universal gas constant \\
        % \hline
        % \hline
        $F$  & \eqref{eq:alg_volt}, \eqref{eq:E_ion_j} & $96485 \frac{\text{C}}{\text{mol}}$ & Faraday constant \\
        $\nu_1$ &\eqref{eq:dwA}& $126 \cdot 10^{-5} \Arbl $ ${\text{dm}^6\cdot\text{mol}^{-1}\cdot\text{s}^{-1}}$ & osmotic permeability $\times$ molar volume of water\\
        &&& on basolateral membrane \\
	% \hline
	$\nu_2$ &\eqref{eq:ode_water}& $126 \cdot 10^{-5} \Arap $ ${\text{dm}^6\cdot\text{mol}^{-1}\cdot\text{s}^{-1}}$ & osmotic permeability $\times$ molar volume of water\\ 
    &&& on apical surface \\
	% \hline
	$\nu_p$ &\eqref{eq:dwB}& $126 \cdot 10^{-5} \Arpara $ ${\text{dm}^6\cdot\text{mol}^{-1}\cdot\text{s}^{-1}}$ & osmotic permeability $\times$ molar volume of water\\
        &&& on para. pathway \\
        \hline
    \end{tabular}
    \caption{
    {Parameters used in Equations \eqref{eq:osj}--\eqref{eq:E_ion_j} and their default values.}
    }
    \label{tab:constants}
\end{table}

\begin{table}[h!]
    \centering
    \begin{tabular}{| c || c | c | c |}
        \hline
        & Equation & Value & Description \\
        \hline
        \hline
        $\Ose$ & \eqref{eq:ose}, \eqref{eq:cle}, \eqref{eq:ode_water} & 300 mM  & Extra. osmolarity  \\
        % \hline
        $\ye$ & \eqref{eq:sum_e=0}, \eqref{eq:ose}, \eqref{eq:ione} & 1 mM & Extra. impermeant molecule \\
        % \hline
        $\zy$ & \eqref{eq:sum_e=0}, \eqref{eq:ione}  & -1      & Charge of extra. impermeant molecule \\
        % \hline
        $\nae$ & \eqref{eq:E_ion_j}, \eqref{eq:sum_e=0}, \eqref{eq:ose}, \eqref{eq:nae} & 147 mM & Extra. Sodium concentration \\
        % \hline
        $\ke$ & \eqref{eq:E_ion_j}, \eqref{eq:sum_e=0}, \eqref{eq:ose}, \eqref{eq:nae} & 3 mM & Extra. Potassium concentration \\
        % \hline
        $\cle$ & \eqref{eq:E_ion_j}, \eqref{eq:sum_e=0}, \eqref{eq:ose}, \eqref{eq:ione} & 149 mM & Extra. Chloride concentration \\
        $\nacle$ & \eqref{eq:nacle} & 0 mM & Extra. Salt concentration \\
        \hline
            \end{tabular}
    \caption{
    {Parameters of the extracellular space that describe the extracellular concentrations.} Equation \eqref{eq:sum_e=0} is used to update $\Ose$, and Equation \eqref{eq:nacle} is used to update $\nae$ and $\cle$. 
    }
    \label{tab:extracellular}
\end{table}

\begin{table}[h!]
    \centering
    \begin{tabular}{| c || c | c | c |}
        \hline
         & Equations & Value & Description  \\
        \hline
        \hline
        $T$ & \eqref{eq:E_ion_j} & $37$ $^{\circ}$C & Temperature \\
        $\ggnabl$ &\eqref{eq:dnaA}& 1 mS dm$^{-2}$ & Na\textsuperscript{+} conductance on basolateral membrane \\
	% \hline
        $\ggnaap$ &\eqref{eq:dION}& 1 mS dm$^{-2}$ & Na\textsuperscript{+} conductance on apical surface \\
	% \hline
	$\ggnapara$ &\eqref{eq:dnaB}& 1 mS dm$^{-2}$ & Na\textsuperscript{+} conductance on para. pathway\\
    % \hline
    $\ggkbl$ &\eqref{eq:dkA}& 60 mS dm$^{-2}$& K\textsuperscript{+} conductance on basolateral membrane\\ 
	% \hline
	$\ggkap$ &\eqref{eq:dION}& 30 mS dm$^{-2}$ &  K\textsuperscript{+} conductance on the apical surface\\
	% \hline
	$\ggkpara$ &\eqref{eq:dkB}& 60 mS dm$^{-2}$ & K\textsuperscript{+} conductance on para. pathway\\
        % \hline
        $\ggclbl$ &\eqref{eq:dclA}& 2 mS dm$^{-2}$ & Cl\textsuperscript{-} conductance on basolateral membrane\\
	% \hline
	$\ggclap$ &\eqref{eq:dION}& 300 mS dm$^{-2}$ &Cl\textsuperscript{-} conductance on apical surface\\
	% \hline
	$\ggclpara$ &\eqref{eq:dclB}& 10 mS dm$^{-2}$ & Cl\textsuperscript{-} conductance on para. pathway\\
        % \hline
        \hline
        \end{tabular}
    \caption{Default parameter values for temperature and conductances. 
    Total conductance is defined as $\gionany=\ggionany\;\arany$.
    }
    \label{tab:conductances}
\end{table}

\begin{table}[h!]
    \centering
    \begin{tabular}{| c | c | c |}
        \hline 
         & Default value & Description \\
        \hline
        \hline
        % % \hline
        $\pumprate$ &   1 $\mu$A dm$^{-2}$ & NKA pump rate per unit area \\
        % \hline
        $\gamma_{Na}$&   3 & Sodium stoichiometry \\
        % \hline
        $\gamma_{K}$&   2 & Potassium stoichiometry \\
        \hline
    \end{tabular}
    \caption{{NKA pump rate and stoichiometry parameters. Parameters are given in equations at the beginning of Sections \ref{sec:passive} and \ref{sec:active}.} The pump rate does not possess a default value; however, we are interested in compensatory mechanisms for when $\pumprate$ is small, which is given here the value of $\pumprate$ = 1 $\mu$A dm$^{-2}$.
    }
    \label{tab:pump}
\end{table}

%\vskip2in

\subsection{Motivation to study a larger range for each parameter} \label{sec:motivation}

{Table~\ref{tab:compare_cond} compares our default parameter values with those used in \cite{weinstein1992mathematical}, \cite{weinstein2012mathematical}, \cite{weinstein2015mathematical}, and \cite{edwards2014effects}. We converted their permeabilities (reported as cm/sec) to conductances per unit area (mS/dm$^2$). We calculated conductance by differentiating the Goldman-Hodgkin-Katz (GHK) current-voltage equation with respect to voltage and evaluated it at the reversal potential. The reversal potentials were computed using the ion concentrations reported in each reference. The comparison for extracellular concentrations includes epithelial cells from the rat proximal tubule (\cite{weinstein1992mathematical,weinstein2015mathematical}), loop of Henle (\cite{weinstein2015mathematical}), thick ascending limb and macula densa (\cite{edwards2014effects}), and a general biophysical representation of kidney function (\cite{weinstein2012mathematical}).}

\begin{table}[h!]
    \centering
    \small
    \begin{tabular}{|l|c|c|c|c|c|}
        \hline
         & \textbf{Ours} & 
         %\textbf{Weinstein 1992} 
         \cite{weinstein1992mathematical} & %\textbf{2012} 
         \cite{weinstein2012mathematical} & %\textbf{2015}
         \cite{weinstein2015mathematical} & %\textbf{Edwards 2014} 
         \cite{edwards2014effects} \\
        \hline
        $\ggnabl$ (mS dm$^{-2}$) & 1 & 0.12497 & 113.30--12948.04 & 0.64--4.49 & \\
        $\ggnaap$ (mS dm$^{-2}$) & 1 & 0 &  & 0.00--0.00 & \\
        $\ggnapara$ (mS dm$^{-2}$) & 1 &  &  & 4165.83--8331.66 & \\
        $\ggkbl$ (mS dm$^{-2}$) & 30 & 24.608 & 79.65--5881.49 & 132.88--885.89 & 11.2262 \\
        $\ggkap$ (mS dm$^{-2}$) & 3 & 110.736 &  & 16.61--110.74 & 1389.75 \\
        $\ggkpara$ (mS dm$^{-2}$) & 60 &  &  & 1784.08--3568.17 & \\ 
        $\ggclbl$ (mS dm$^{-2}$) & 2 & 0 & 93.31--24526.06 & 0.00--0.00 & 348.196 \\
        $\ggclap$ (mS dm$^{-2}$) & 300 & 0 &  & 0.00--0.00 & \\
        $\ggclpara$ (mS dm$^{-2}$) & 10 &  &  & 1495.42--2990.84 & \\
        $[$Na$^+]_e$ (mM) & 146 & 140 & 140 & & 144 \\
        $[$K$^+]_e$ (mM) & 3 & 4.9 & 4.9 & & 5 \\
        $[$Cl$^-$]$_e$ (mM) & 147 & 114 & 113.2 & & 123 \\
        \hline
    \end{tabular}
    \caption{Comparison of conductances and concentrations.}
    \label{tab:compare_cond}
\end{table}

%%%%%%%%%%%%%%%%%%%%%%%%
\subsection{Voltage in time series} \label{sec:voltage}

Voltage variables are determined at each time step by the electroneutrality constraint. For each compartment $j\in\{A,B\}$, we assume
\[0 = \sum_{\text{ion}} F \, \zion \, \ionj + F \, \zj \, \xj\]
so that the total intracompartmental charge concentration is zero. Multiplying this expression by $\wj$ and differentiating with respect to time gives
\[ 0 = \sum_{\text{ion}} F \, \zion \frac{d\left(\wj \ionj\right)}{dt} \]
where $\wj \, \xj$ is constant. Using the right-hand sides of Equations \eqref{eq:ode_con_a} and \eqref{eq:ode_con_b}, and collecting terms in $\vA$ and $\vB$, we obtain a linear system of the form
\[
\alpha \vA + \omega \,\vB = h, \quad \quad \omega \, \vA + \delta \vB = k,
\]
where $\alpha$, $\omega$, $h$, and $k$ depend on reversal potentials, which in turn depend on state variables:
\begin{align*}
    &\alpha=-(\gnabl+\gnaap +\gkbl+\gkap +\gclbl+\gclap) \\
    &\omega=(\gnaap+\gkap+\gclap) \\
    &\delta=-(\gnapara+\gnaap +\gkpara+\gkap +\gclpara+\gclap) \\
    &h=-\enaA(\gnabl+\gnaap) - \ekA(\gkbl+\gkap) - \eclA(\gclbl+\gclap) \\
    &\quad + \gnaap\enaB + \gkap \ekB + \gclap\eclB + (\nu - \kappa)\,\pumprate \, ({ss}_1 \, \Arbl+{ss}_2 \, \Arap) \\
    &k=-\enaB(\gnapara+\gnaap) - \ekB(\gkpara+\gkap) - \eclB(\gclpara+\gclap) \\
    &\quad + \gnaap\enaA + \gkap\ekA + \gclap\eclA - (\nu - \kappa)\,\pumprate \, ({ss}_2 \, \Arap).
\end{align*}
The terms ${ss}_1$ and ${ss}_2$ are used here to denote pump switches for the basolateral membrane and apical surface, respectively. Solving this system yields
\begin{align*}
    &\vA = \frac{\delta h - \omega k}{\alpha \delta - \omega^2} \\
    &\vB = \frac{\alpha k - \omega h}{\alpha \delta - \omega^2}.
\end{align*}
Thus, in our time-series simulations, the voltages are computed algebraically at each time step from the electroneutrality condition.

\subsection{Sampling procedure in VBGSA}\label{sec:appendix_sampling_VBGSA}

Each marginal $\Theta_k$ is stratified into $N$ 
 equiprobable bins; one value is drawn without replacement from each bin. Independent random permutations are applied per parameter, and the draws are assembled to form $N$ $9$-dimensional samples \cite{mckay1979,helton2006survey}. 

We use Latin hypercube sampling (LHS) to sample from $\mathcal{H}^9$. For each $k$, write $\Theta_k=[a_k,b_k]$ and define the monotone transform
\[
\phi_k(t)=
\begin{cases}
\log t, & k\notin\{T,\nacle\},\\
t, & k\in\{T,\nacle\},
\end{cases}
\qquad t\in\Theta_k,
\]
with inverse $\phi_k^{-1}(y)=\exp(y)$ for the log case and $\phi_k^{-1}(y)=y$ otherwise.
Partition $[0,1]$ into $N$ strata $I_i=\bigl((i-1)/N,i/N\bigr]$. For each parameter $k$, draw $U_{i,k}\sim\mathrm{Unif}(I_i)$ independently across $i$, apply an independent random permutation $\pi_k$ of $\{1,\dots,N\}$, and set for $i=1,\dots,N$
\[
\theta_k^{(i)}
=\phi_k^{-1}\!\Big(\phi_k(a_k) + U_{\pi_k(i),k}\,[\,\phi_k(b_k)-\phi_k(a_k)\,]\Big).
\]
The $N$ vectors $\theta^{(i)}=(\theta^{(i)}_1,\dots,\theta^{(i)}_9)$ form a Latin hypercube on $\mathcal{H}^9$ with uniform marginals on $\Theta_k$ for $k\in\{T,\nacle\}$ and uniform marginals on $\log\Theta_k$ for $k\notin\{T,\nacle\}$ \cite{mckay1979,helton2006survey,DELA2022111159}. Positivity of the log-scaled bounds is guaranteed by $0<\beta\,\epsilon_k\ll1$ in \eqref{eq:parameter_interval_log}, where $\beta$ and $\epsilon_k$ are
\begin{equation}\nonumber
\begin{aligned}
    &\beta = 1, \quad \epsilon_1 = 0.10, \quad \epsilon_2=0.25, \quad \epsilon_3=0.5, \quad \epsilon_4=0.5,\\ 
    &\epsilon_5=0.5, \quad \epsilon_6=0.5, \quad \epsilon_7=0.5, \quad \epsilon_8=25 \, ^\circ\text{C}, \quad \epsilon_9=25 \, \text{mM}.
\end{aligned}
\end{equation}
where the index $k=1,\ldots,9$ refers respectively to the parameters 
\[\pumprate,\; 
\ggnabl,\; 
\ggkbl,\;
\ggnaap,\;  
\ggkap,\;
\ggnapara,\;  
\ggkpara,\;
T,\;
\nacle.
\]

In the figures, we plot responses versus the predefined $\Theta_k$ domain for the parameter indexed by $k$. When helpful for visualization, the domain $\Theta_k$ for the parameter of interest may be extended beyond $\Theta_k$, while for the remaining eight parameters the sampling domain $\mathcal{H}^9$ in \eqref{eq:parameter_hypercube} remains unchanged. In the same plots against parameter $k$ we will highlight the original domain $\Theta_k$ with a light gray vertical band centered around $\theta_k^*$. Unless otherwise noted, we set $\beta=1$.

\subsection{Robustness analysis for intermediate and high pump rates} \label{sec:robustness_mid_hi_pump}

The robustness analysis on $\ggnabl$ is repeated for $\pumprate=40$ and $90 \; \mu\text{A dm}^{-2}$. The overall trends are the same as Figure \ref{fig:robustness_gna1} when $\pumprate=1\;\mu\text{A dm}^{-2}$ -- increasing $\ggnabl$ will increase $\naj$, $\clj$, $\vj$, and $\wj$ while decreasing $\kj$ and $\xj$. Compared to the case in Figure \ref{fig:robustness_gna1}, there is a difference in the values of the steady state and spectral abscissa. For the intermediate and high pump rate, the $\naj^{ss}$ and $\clj^{ss}$ will increase from near 0 mM to about 150 mM. $\vj^{ss}$ increases from roughly -100 mV to near 0 mV. $\xj^{ss}$ and $\kj^{ss}$ decrease from roughly 150 mM to about 0 mM. $\wA^{ss}$ increases from $\sim 10 \%$ to $\sim 100 \%$, and $\wB^{ss}$ increases from $\sim 100 \%$ to $\sim 1100$ (volumes are normalized by $\wnorm$). Notice that the trends observed here are the same as those seen in Figure \ref{fig:heatmapA_p_v_gna1}. The spectral abscissa increases from $-1$ to $-10^{-5}$. $\naj^{ss}$ attains values near zero, making RMSRE unstable; standardized mean square error (SMSE) is a better metric for the error \cite{williams2006gaussian}, where SMSE is defined as \[\text{SMSE}=\frac{\mathbb{E}\!\left[(\bm{y}-\bm{o})^2\right]}{\operatorname{Var}(\bm{o})}.\] These two cases are summarized in Figure \ref{fig:g_Na1_p40_and_p90}.

\medskip 

\begin{figure}[h!]
    \centering
    \includegraphics[width=0.45\linewidth]{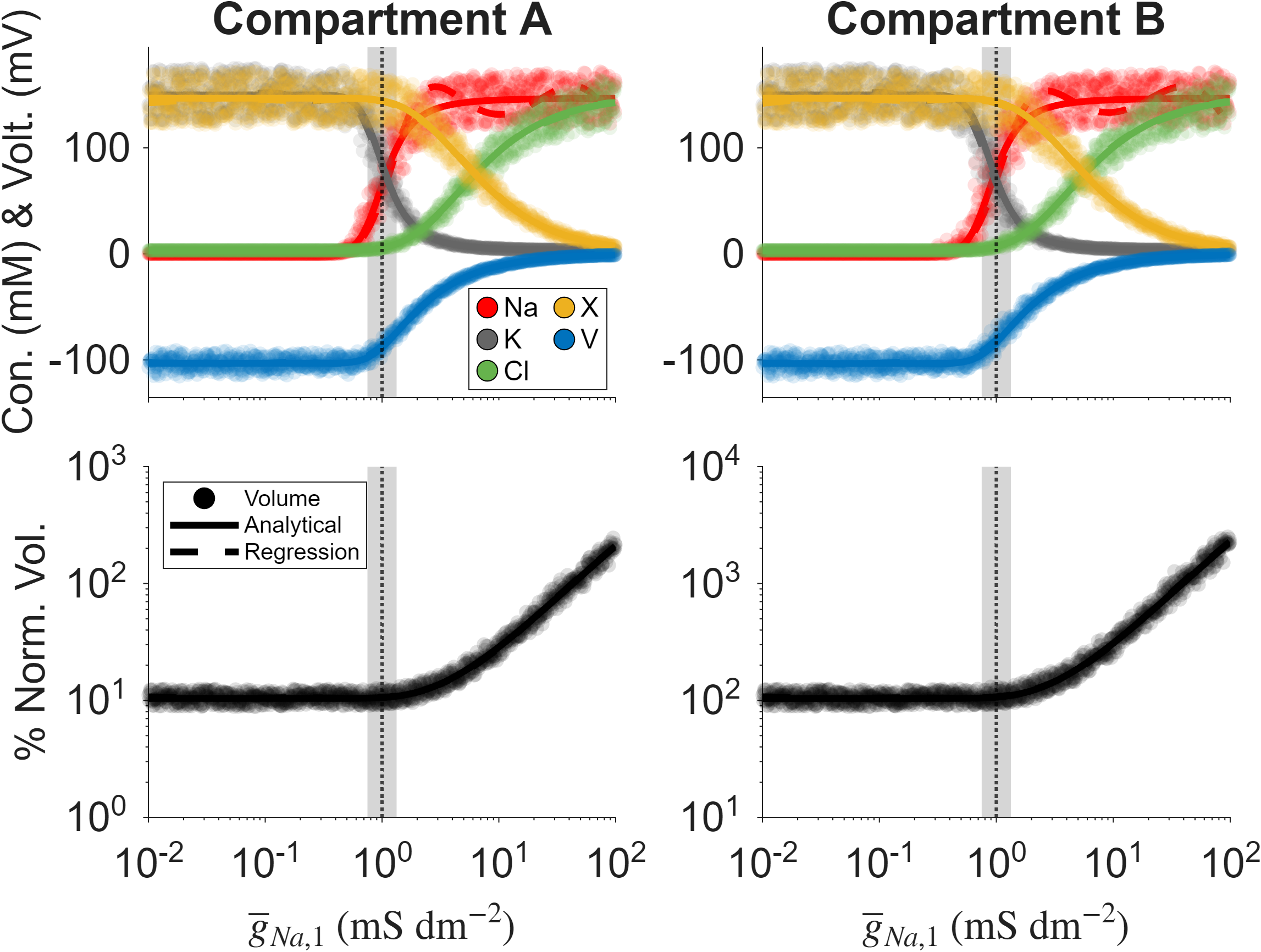} 
    \includegraphics[width=0.22\linewidth]{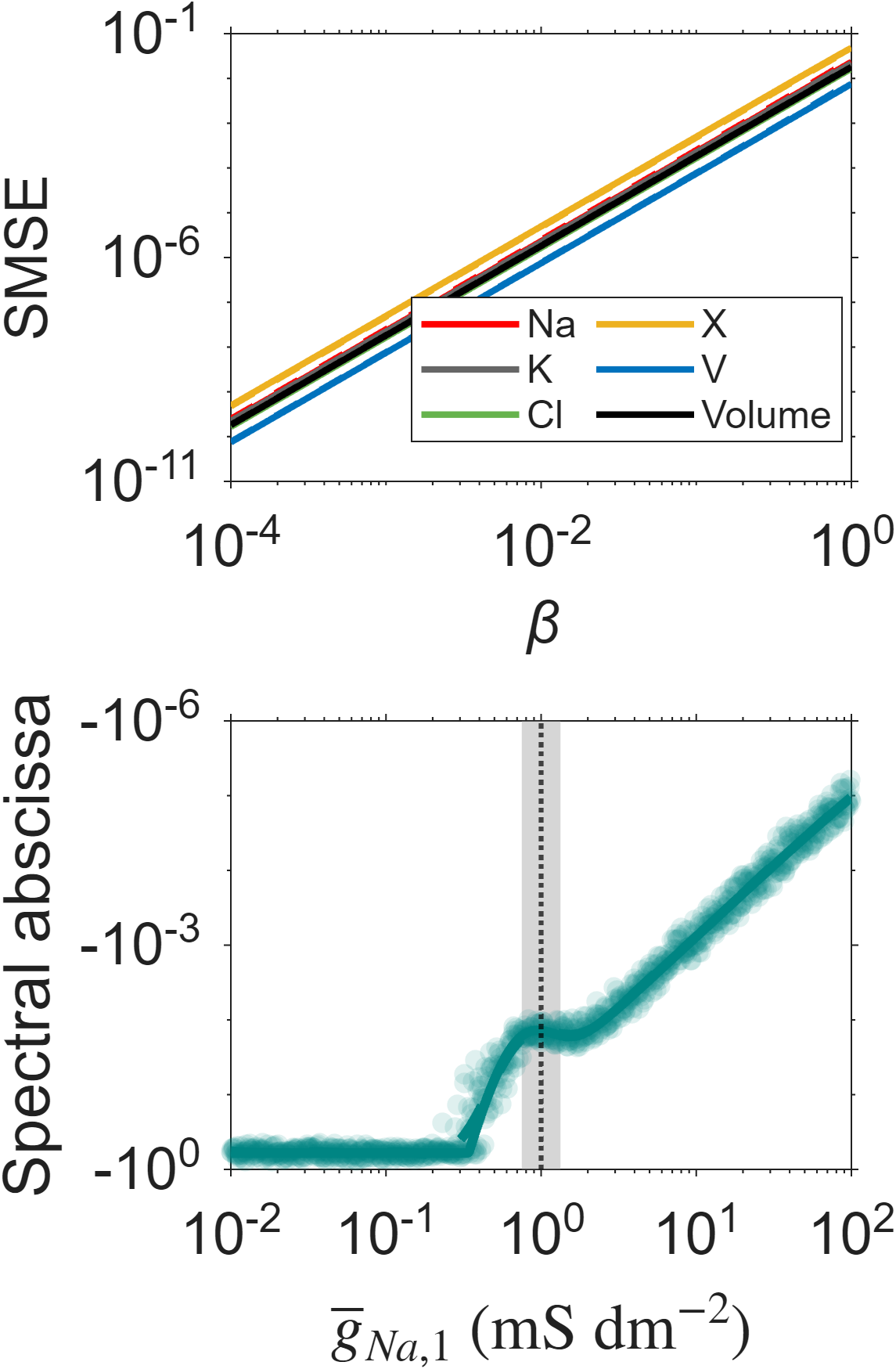} 
    \bigskip 
    
    \bigskip 
    
    \includegraphics[width=0.45\linewidth]{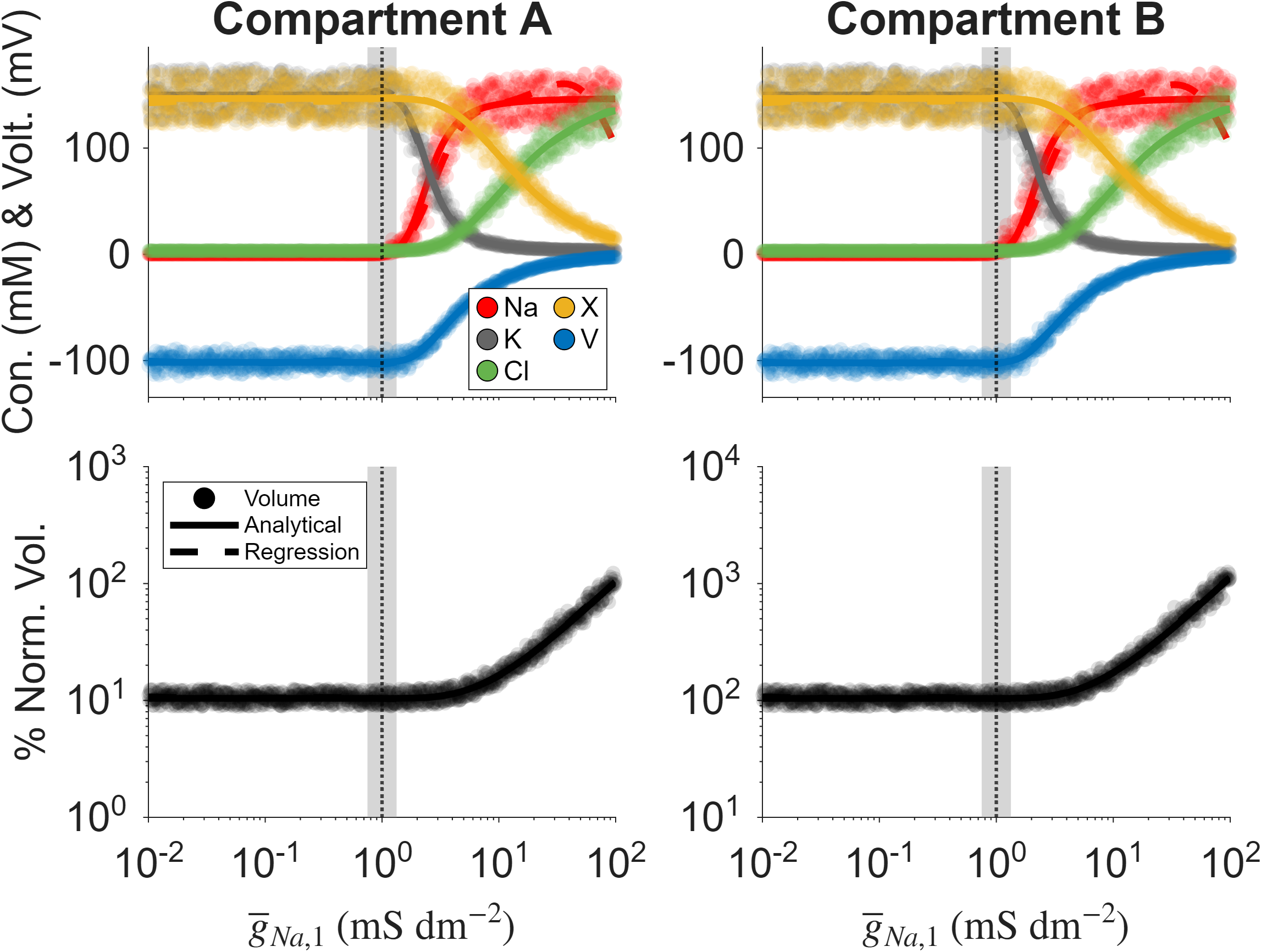} 
    \includegraphics[width=0.22\linewidth]{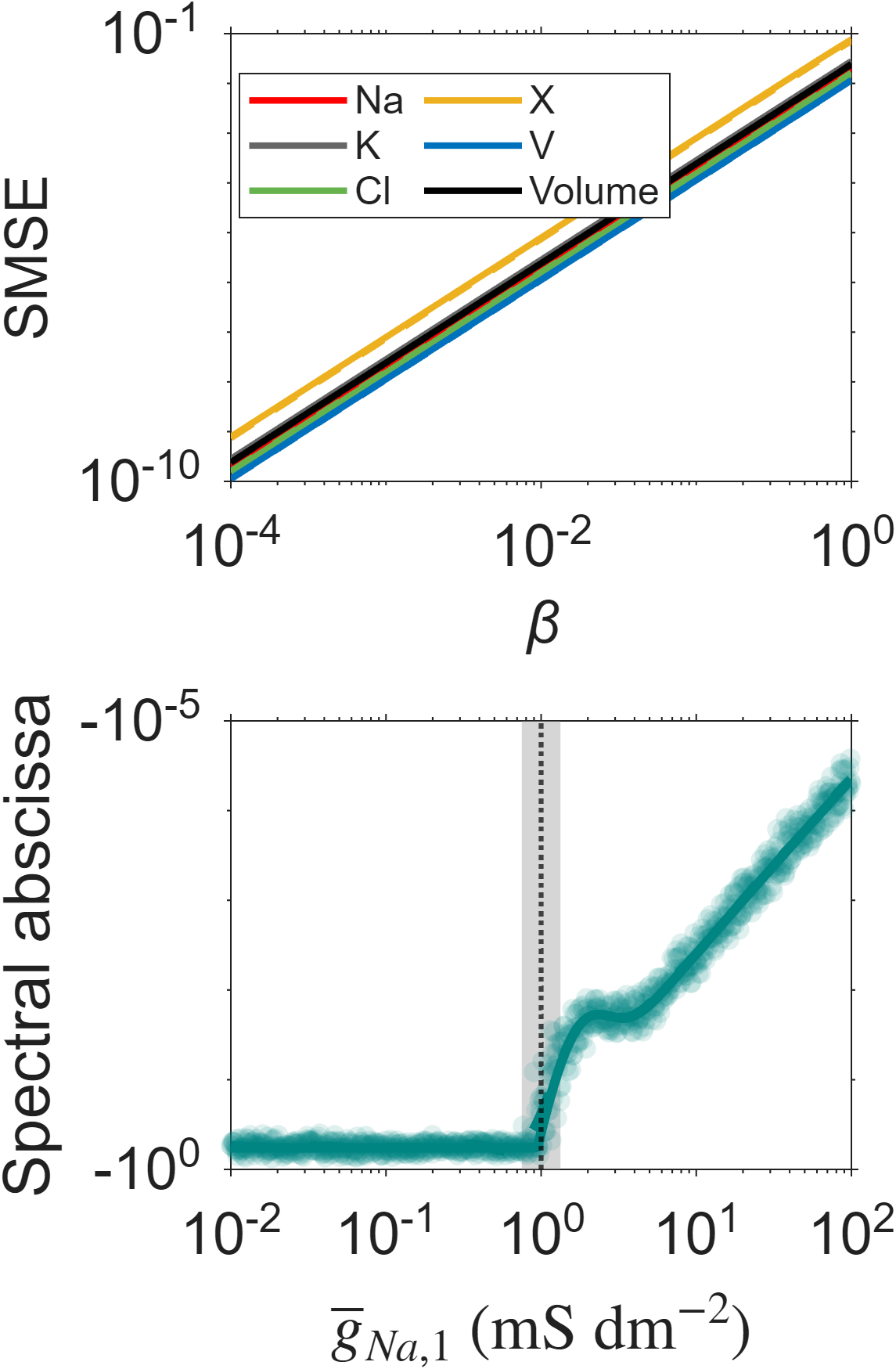}
    \caption{The steady states of the system perturbed around $\pumprate=40$ and $\pumprate=90$.}
    \label{fig:g_Na1_p40_and_p90}
\end{figure}

\end{document}